\documentclass[12pt]{article}
\usepackage{amssymb}
\usepackage{amsfonts}
\usepackage{makeidx}
\usepackage{amsmath}
\usepackage{sw20elba}
\usepackage{hyperref}

\newtheorem{theorem}{Theorem}[section]

\newtheorem{corollary}[theorem]{Corollary}

\newtheorem{definition}[theorem]{Definition}
\newtheorem{example}[theorem]{Example}

\newtheorem{lemma}[theorem]{Lemma}

\newtheorem{proposition}[theorem]{Proposition}
\newtheorem{remark}[theorem]{Remark}

\newenvironment{proof}[1][Proof]{\textbf{#1.} }{\ \rule{0.5em}{0.5em}}
\input{tcilatex}
\begin{document}

\title{Relativistic point dynamics and Einstein formula as a property of
localized solutions of a nonlinear Klein-Gordon equation.}
\author{Anatoli Babin and Alexander Figotin \\
University of California at Irvine}
\maketitle

\begin{abstract}
Einstein's relation $E=Mc^{2}$ between the energy $E$ and the mass $M$ is
the cornerstone of the relativity theory. This relation is often derived in
a context of the relativistic theory for closed systems which do not
accelerate. By contrast, Newtonian approach to the mass is based on an
accelerated motion. We study here a particular neoclassical field model of a
particle governed by a nonlinear Klein-Gordon (KG) field equation. We prove
that if a solution to the nonlinear KG equation and its energy density
concentrate at a trajectory, then this trajectory and the energy must
satisfy the relativistic version of Newton's law with the mass satisfying
Einstein's relation.\ Therefore the internal energy of a localized wave
affects its acceleration in an external field as the inertial mass does in
Newtonian mechanics. We demonstrate that the "concentration" assumptions
hold for a wide class of rectilinear accelerating motions.
\end{abstract}

\section{Introduction}

One of the goals of this paper is to demonstrate that a certain field model
of a particle provides for a deeper understanding of the Einstein
energy-mass relation $E=M\mathrm{c}^{2}$. This field model is an integral
part of our recently introduced neoclassical electromagnetic (EM)\ theory 
\cite{BF4}-\cite{BF7} of charges, and the relativistic aspects of that
theory has been explored in \cite{BF8}. An idea that a particle can be
viewed as a field excitation carrying a certain amount of energy is a rather
old one. Einstein \ and Infeld wrote in \cite[p. 257]{Einstein Infeld}: "
What impresses our senses as matter is really a great concentration of
energy into a comparatively small space. We could regard matter as the
regions in space where the field is extremely strong." But implementation of
this idea in a mathematically sound theory is a challenging problem.
Einstein remarks in his letter to Ernst Cassirer in March 16, 1937, \cite[%
pp. 393-394]{Stachel EBZ}: "One must always bear in mind that up to now we
know absolutely nothing about the laws of motion of material points from the
standpoint of "classical field theory." For the mastery of this problem,
however, no special physical hypothesis is needed, but "only" the solution
of certain mathematical problems".

We start with recalling basic facts from the relativistic mechanics
including the relativistic dynamics of a mass point and the relativistic
field theory, see for instance, \cite{Anderson}, \cite{Barut}, \cite{Moller}%
, \cite{Pauli RT}, \cite{Sexl}. In a relativistic field theory the
relativistic field dynamics is derived from a relativistic covariant
Lagrangian. The field equations, the energy and the momentum, the forces and
their densities are naturally defined in terms of the Lagrangian both in the
cases of closed and non closed (with external forces) systems. For a closed
system the total energy-momentum $\left( E,\mathbf{P}\right) $ transforms as
4-vector, \cite[Sec. 7.1-7.5]{Anderson}, \cite[Sec. 3.1-3.3, 3.5]{Moller}, 
\cite[Sec. 37]{Pauli RT}, \cite[Sec. 4.1]{Sexl}. In particular, the total
momentum $\mathbf{P}$ has a simple form $\mathbf{P}=M\mathbf{v}$ where the
constant velocity $\mathbf{v}$ originates from the corresponding parameter
of the Lorentz group. Then one can naturally define and interpret the mass $%
M $ for a closed relativistic system as the coefficient of proportionality
between the momentum $\mathbf{P}$ and the velocity $\mathbf{v}$. As to the
energy of a closed system, the celebrated Einstein energy-mass relation holds%
\begin{equation}
E=M\mathrm{c}^{2},\quad M=m_{0}\gamma ,\quad \gamma =\left( 1-\mathbf{v}^{2}/%
\mathrm{c}^{2}\right) ^{-1/2},  \label{fre3a}
\end{equation}%
where $m_{0}$ is the rest mass and $\gamma $ is the Lorentz\textbf{\ }%
factor. Observe that the above definition of mass is based on the
relativistic argument for a uniform motion with a constant velocity $\mathbf{%
v}$ without acceleration. An immediate implication of Einstein's mass-energy
relation (\ref{fre3a}) is that the rest mass of a closed system is
essentially equivalent to the internal energy of the system.

Let us turn now to the relativistic dynamics of a mass point which
accelerates under the action of a force $\mathbf{f}\left( t,\mathbf{r}%
\right) $. This dynamics\ is governed by the relativistic version of
Newton's equation of the form, \cite{Barut}, \cite{Pauli RT}, \cite{Moller}: 
\begin{equation}
\frac{\mathrm{d}}{\mathrm{d}t}\left( M\mathbf{v}\left( t\right) \right) =%
\mathbf{f}\left( t,\mathbf{r}\right) ,\quad \mathbf{v}\left( t\right) =\frac{%
\mathrm{d}\mathbf{r}\left( t\right) }{\mathrm{d}t},\quad M=m_{0}\gamma .
\label{fre1}
\end{equation}%
Note that in (\ref{fre1}) the rest mass $m_{0}$ is prescribed as an
intrinsic property of the mass point, for the mass point does not have
internal degrees of freedom.\ The equation (\ref{fre1}), just as in the case
of classical Newtonian mechanics, suggests that the mass $M$\ is a measure
of inertia, that is the coefficient that relates the acceleration to the
known force $\mathbf{f}$. To summarize, according to the relativity
principles the rest mass is naturally defined for a uniform motion, whereas
in the Newtonian mechanics the concept of inertial mass is introduced
through an accelerated motion. Note that the relativistic and
non-relativistic\ masses are sometimes considered to be "rival and
contradictory", \cite{EriksenV76}.

A principal problem we want to take on here is as follows. We would like to
construct a field model of a charge where the internal energy of a localized
wave affects its acceleration in an external field the same way the inertial
mass does in Newtonian mechanics. We could not find such a model in
literature and we introduce and study it here. The model allows, in
particular, to consider the uniform motion in the absence of external forces
and the accelerated motion in the same framework. Hence within the same
framework (but in different regimes) we can determine the mass either from
the analysis of the 4-vector of the energy-momentum or using the Newtonian
approach. When both the approaches are relevant, they agree, see Remark \ref%
{Rfreemass}.

The proposed model is based on our manifestly relativistic Lagrangian field
theory \cite{BF4}-\cite{BF7}. In the case of a single charge its state is
described by a complex-valued scalar field (wave function) $\psi \left( t,%
\mathbf{x}\right) $ of the time $t$ and the position vector $\mathbf{x}\in 
\mathbb{R}^{3}$ and its time evolution is governed by the following
nonlinear Klein-Gordon (KG) equation%
\begin{equation}
-\frac{1}{\mathrm{c}^{2}}\tilde{\partial}_{t}^{2}\psi +\tilde{\nabla}%
^{2}\psi -G^{\prime }\left( \left\vert \psi \right\vert ^{2}\right) \psi -%
\frac{m^{2}\mathrm{c}^{2}}{\chi ^{2}}\psi =0,  \label{KGex}
\end{equation}%
where $m$ is a positive mass parameter and $\chi $ is a constant which in
physical applications coincides with (or is close to) the Planck constant $%
\hbar $. The covariant derivatives in equation (\ref{KGex}) are defined by 
\begin{equation}
\tilde{\partial}_{t}=\partial _{t}+\frac{\mathrm{i}q}{\chi }\varphi ,\quad 
\tilde{\nabla}=\nabla -\frac{\mathrm{i}q}{\chi \mathrm{c}}\mathbf{A},
\label{dtex}
\end{equation}%
where $q$ is the value of the charge and $\varphi \left( t,\mathbf{x}\right)
,\mathbf{A}\left( t,\mathbf{x}\right) $ are the potentials of the\emph{\
external EM field. }The external EM field acts upon the charge resulting
ultimately in the Lorentz force when Newton's point mass equation is
relevant. In this paper, for simplicity, we treat the case where only
electric forces are present and the magnetic potential is set to be zero,
that is 
\begin{equation}
\mathbf{A=0,}\quad \tilde{\nabla}=\nabla .  \label{Aeq0}
\end{equation}%
The physical treatment of the case with a non-zero external magnetic field
is provided in \cite{BF8}.

The nonlinear term $G^{\prime }$\ in the KG (\ref{KGex}) provides for the
existence of localized solutions for resting or uniformly moving charges.
The nonlinearity $G$ and its properties are considered in Section \ref%
{snonlin}. Importantly, the nonlinearity $G$ involves a size parameter\ $a$
that determines the spatial scale of the charge when at rest. Notice that $%
\left\vert \psi \right\vert ^{2}$ is interpreted as a charge distribution
and not as a probability density.

A sketch of our line of argument is a follows. From the Lagrangian of the KG
equation introduced below we derive a usual expression for the energy
density which allows to define the energy involved in Einstein's formula. To
be able to apply the Newtonian approach it is necessary to relate point
trajectories $\mathbf{r}\left( t\right) $ to wave solutions $\psi $ and
trace point dynamics to the field equations. The latter is accomplished
based on a concept of solutions "concentrating at a trajectory".

Roughly speaking (see Section \ref{SConctraj} for the exact definitions)\
solutions concentrate at a given trajectory $\mathbf{\hat{r}}\left( t\right) 
$ if their energy densities $\mathcal{E}\left( \mathbf{x},t\right) $
restricted to $R_{n}$-neighborhoods of $\mathbf{\hat{r}}\left( t\right) $
locally converge to $E\left( t\right) \delta \left( \mathbf{x}-\mathbf{\hat{r%
}}\left( t\right) \right) $ as $R_{n}\rightarrow 0$.\ A concise formulation
of the main result of the paper, Theorem \ref{TconcEinstein}, is as follows.
We prove that \emph{if a sequence of solutions of the KG\ equation
concentrates at a trajectory }$\mathbf{\hat{r}}\left( t\right) $\emph{\ then
the restricted energy }$\mathcal{\bar{E}}_{n}$\emph{\ of the solutions
converges to a function }$\mathcal{\bar{E}}\left( t\right) $\emph{\ so that
the following relativistic version of Newton's equation holds:} 
\begin{equation}
\frac{\mathrm{d}}{\mathrm{d}t}\left( \frac{\mathcal{\bar{E}}}{\mathrm{c}^{2}}%
\mathbf{\hat{v}}\right) =\mathbf{f}\left( t,\mathbf{\hat{r}}\right) ,\quad 
\mathbf{\hat{v}=}\frac{\mathrm{d}}{\mathrm{d}t}\mathbf{\hat{r}.}
\label{Newt0}
\end{equation}%
The electric force $\mathbf{f}$\ in (\ref{Newt0}) is defined by formula 
\begin{equation}
\mathbf{f}\left( t,\mathbf{\hat{r}}\right) =-\bar{\rho}\nabla \varphi \left(
t,\mathbf{\hat{r}}\right)  \label{fLor0}
\end{equation}%
where $\varphi $ is the electric potential as in the KG\ equation (\ref{dtex}%
), $\bar{\rho}$ is a constant describing the limit charge. The limit
restricted energy $E\left( t\right) $ satisfies relation 
\begin{equation}
\frac{\mathcal{\bar{E}}\left( t\right) }{\mathrm{c}^{2}}=M_{0}\gamma ,\quad
\gamma =\left( 1-\mathbf{\hat{v}}^{2}/\mathrm{c}^{2}\right) ^{-1/2},
\label{En0}
\end{equation}%
where $M_{0}$ is a constant which plays the role of a generalized rest mass.
Observe that \ \emph{equation (\ref{Newt0}) takes the form of the
relativistic version of Newton's law if the mass is defined by Einstein's
formula}\ $M\left( t\right) =\frac{1}{\mathrm{c}^{2}}\mathcal{\bar{E}}\left(
t\right) $. The relation between the generalized rest mass $M_{0}$ and the
rest mass $m_{0}$ of resting solutions is discussed in Remark \ref{Rfreemass}%
. Note that the same KG equation (\ref{KGex}) with the same value of the
mass parameter $m$ has rest solutions with different energies, and
consequently different rest masses,\ see Sections (\ref{Snonlinrest}) and %
\ref{Suniform}.\ 

Therefore we can make the following conclusion: \emph{in the framework of
the KG field theory, the relativistic material point dynamics is represented
by concentrating solutions of the KG equation with the mass determined by
Einstein's formula from the limit restricted energy of the solutions.}

We can add to the above outline a few guiding points to the mathematical
aspects of our approach. Equation (\ref{Newt0}) in Theorem \ref%
{TconcEinstein} produces a \emph{necessary condition} for solutions of KG\
equation to concentrate to a trajectory and reveals their asymptotic
point-like dynamics. To obtain this necessary condition we have to derive
the point dynamics governed by an ordinary differential equation (\ref{Newt0}%
)\ from the dynamics of waves governed by a partial differential equation.
The concept of "\emph{concentration of functions at a given trajectory} $%
\mathbf{\hat{r}}\left( t\right) $", see Definition \ref{Dlocconverge} below,
is the first step in relating spatially localized fields $\psi $ to point
trajectories. The definition of concentration of functions has a sufficient
flexibility to allow for general regular trajectories $\mathbf{\hat{r}}%
\left( t\right) \ $and plenty of functions localized about the trajectory,
see Example \ref{Rboundint}. But if a sequence of functions concentrating at
a given trajectory \emph{are also solutions of the KG equation then,
according to Theorem \ref{TconcEinstein}, the trajectory and the limit
energy must satisfy the relativistic Newton's equations together with
Einstein's formula}. To derive the equations (\ref{Newt0}) and (\ref{En0}),
we introduce the energy $\mathcal{\bar{E}}_{n}\left( t\right) $ restricted
to a narrow tubular neighborhood of the trajectory with radius $R_{n}$ \ and
adjacent ergocenters $\mathbf{r}_{n}\left( t\right) $ of the concentrating
solutions; then\ we infer integral equations for the restricted energy and
adjacent ergocenters from the energy and momentum conservation laws and the
continuity equation, \ and then we pass to the limit as $R_{n}\rightarrow 0$%
, see Theorems \ref{Ldiscrepancy1} and \ref{Ldiscrepancy2} \ and the proof
of Theorem \ref{TconcEinstein}.\ The intrinsic length scales $a$ and $a_{%
\mathrm{C}}=\frac{\chi }{m\mathrm{c}}$ of concentrating solutions \ are much
smaller than the radius $R_{n}$, namely $R_{n}/a\rightarrow \infty ,$ $%
R_{n}/a_{\mathrm{C}}\rightarrow \infty $. The determination of restricted
energy and charge and adjacent ergocenters involves integration over a large
relative to $a$ and $a_{\mathrm{C}}$ spatial domain of radius $R_{n}$.
Therefore, equation (\ref{Newt0}) inherits integral, non-local
characteristics of concentrating solutions \ which cannot be reduced to
their behavior on the trajectory, namely the limit restricted energy $%
\mathcal{\bar{E}}\left( t\right) $ and restricted charge $\bar{\rho}_{\infty
}$. Therefore it is natural to call our method of determination of the point
trajectories \emph{semi-local}. This semi-local feature allows to capture
Einstein's relation between mass and energy.

Note that in many problems of physics and mathematics a common way to
establish a relation between point dynamics and wave dynamics is by means of
the WKB method, see for instance \cite{MaslovFedorjuk}, \cite[Sec. 7.1]%
{Nayfeh}.\ We remind that the WKB\ method is based on the quasiclassical
ansatz for solutions to a hyperbolic partial differential equation and their
asymptotic expansion. The leading term of the expansion results in the
eikonal equation; wavepackets and their energy propagate along its
characteristics. The characteristics represent point dynamics and are
determined from a system of ODE which can be interpreted as a law of motion
or a law of propagation. The construction of the characteristics involves
only local data. The proposed here approach also relates waves governed by
certain PDE's in asymptotic regimes to the point dynamics but it differs
signficantly from the WKB method. In particular, our approach is not based
on any specific ansatz and it is not not entirely local but rather it is
semi-local.

\ A simple and important example of concentrating solutions is provided by a
resting or uniformly moving charge in Example \ref{Euniform}. Examples of 
\emph{accelerating} charges are constructed in Section \ref{secAR}, where we
show that all the assumptions on concentrating solutions imposed in the
Theorem \ref{TconcEinstein} are satisfied. For a given accelerating
rectilinear translational motion and a fixed Gaussian shape of $\left\vert
\psi \right\vert $ we construct an electric potential $\varphi $ consisting
of (i) an explicitly written principal component yielding desired
acceleration and (ii) an additional vanishingly small "balancing" component
allowing for the shape $\left\vert \psi \right\vert $ to be exactly
preserved. The construction of the balancing component is reduced to solving
a system of characteristic ODE which allows for a detailed analysis.

The rest of the paper is structured as follows. In Section \ref{snonlin} we
describe the nonlinearities $G$. In Section \ref{sreldis}\ we introduce the
Lagrangian for the KG equation and related densities of charge, energy and
momentum, and verify then the corresponding conservation laws. In Sections %
\ref{Snonlinrest} and \ref{Suniform} solutions for resting and uniformly
moving charge are considered and their energies are found. In Section \ref%
{SReldyn} we define our key concepts and derive then the equations (\ref%
{Newt0}), (\ref{En0})\ as necessary conditions for solutions of the KG
equation to concentrate at a trajectory. In Section \ref{secAR} an example
is provided for a relativistically accelerating charge which satisfies all
the assumptions imposed in the theorem on concentrating solutions. In
Section \ref{Sconclin}\ we discuss application of the results of the paper
to linear KG equations. Note that the generalization of the results of the
paper to higher than $3$ space dimensions is straightforward.

\section{Basic properties of the Klein-Gordon equation\label{Cbas}}

\subsection{Nonlinearity, its basic properties and examples\label{snonlin}}

In this section we describe the nonlinearity $G^{\prime }\left( \left\vert
\psi \right\vert ^{2}\right) $\ which enters KG equation (\ref{KGex}). We
assume that\ $G^{\prime }\left( s\right) $ is of class $C^{1}\left( \mathbb{%
R\setminus }0\right) $, it may have a\ mild singularity point at $s=0.$
Namely, we assume that\ the function $G^{\prime }\left( \left\vert \psi
\right\vert ^{2}\right) \psi ,$ $\psi \in \mathbb{C}$, can be extended to a
function which belongs to the Holder class $C^{\alpha }\left( \mathbb{C}%
\right) $\ with any $\alpha $ such that $1>\alpha >0$\ and the
antiderivative $G\left( \left\vert \psi \right\vert ^{2}\right) $ belongs to 
$C^{1+\alpha }\left( \mathbb{C}\right) $. Below we give several explicit
examples of the function $G^{\prime }\left( |\psi |^{2}\right) $ which
allows for resting, time-harmonic localized solutions\ of (\ref{KGex}) and
satisfy the above assumptions. In the examples we define the nonlinearity $G$
and its dependence on the size parameter $a$ based on the ground\ state $%
\mathring{\psi}\geq 0$.\ The dependence of the ground state $\mathring{\psi}$
on the size parameter $a>0$ is as follows: 
\begin{equation}
\mathring{\psi}\left( r\right) =\mathring{\psi}_{a}\left( r\right) =a^{-3/2}%
\mathring{\psi}_{1}\left( a^{-1}r\right) ,r=\left\vert x\right\vert \geq 0
\label{nrac5}
\end{equation}%
where $\mathring{\psi}_{1}\left( r\right) $ is a given function. The
dependence on $a$ is chosen so that $L^{2}$-norm $\Vert \mathring{\psi}%
_{a}\left( |x|\right) \Vert $ does not depend on $a,$ hence the function $%
\mathring{\psi}_{a}\left( r\right) $ satisfies the normalization condition 
\begin{equation*}
\Vert \mathring{\psi}_{a}\left( |x|\right) \Vert =\nu ,\nu >0
\end{equation*}%
with a fixed $\nu $\ for every $a>0$.\ \ The function $\mathring{\psi}%
_{a}\left( r\right) $ is assumed to be a smooth (three times continuously
differentiable) positive monotonically decreasing function of $r\geq 0$
which is square integrable with weight $r^{2},$ we assume that its
derivative $\mathring{\psi}_{a}^{\prime }\left( r\right) $ is negative for $%
r>0$ and we assume it to satisfy a charge normalization condition.

We assume that $\mathring{\psi}_{a}$ satisfies the charge equilibrium\
equation: 
\begin{equation}
\nabla ^{2}\mathring{\psi}_{a}=G_{a}^{\prime }(\mathring{\psi}_{a}^{2})%
\mathring{\psi}_{a}.  \label{stp}
\end{equation}%
This equation is obtained from (\ref{KGex})\ by substitution\ $\psi
=e^{-i\omega _{0}t}\mathring{\psi}\left( |x|\right) ,$ $\omega _{0}=m\mathrm{%
c}^{2}/\chi $; this equation can be used to define the nonlinearity.
Obviously, it is sufficient to define $G_{a}^{\prime }$\ for $a=1$\ and then
set $G_{a}^{\prime }\left( s\right) =a^{-2}G_{1}^{\prime }\left(
a^{3}s\right) $ . From the equation (\ref{stp}) we express $G_{1}^{\prime
}\left( s\right) $ with $s=\mathring{\psi}_{1}^{2}\left( r\right) ;$ since $%
\mathring{\psi}_{1}^{2}\left( r\right) $ is a \emph{monotonic} function,\ we
can find its inverse $r=r\left( s\right) ,\ $and we obtain an\ explicit
expression 
\begin{equation}
G_{1}^{\prime }\left( s\right) =\nabla ^{2}\mathring{\psi}_{1}\left( r\left(
s\right) \right) /\mathring{\psi}_{1}\left( r\left( s\right) \right) ,\ 0=%
\mathring{\psi}_{1}^{2}\left( \infty \right) \leq s\leq \mathring{\psi}%
_{1}^{2}\left( 0\right) .  \label{intps}
\end{equation}%
Since $\mathring{\psi}_{1}\left( r\right) $ is smooth and $\partial _{r}%
\mathring{\psi}_{1}<0$, $G^{\prime }(|\psi |^{2})$ is smooth for $0<|\psi
|^{2}<\mathring{\psi}_{1}^{2}\left( 0\right) $.\ We assume that $\mathring{%
\psi}_{a}\left( r\right) $ is a smooth function of class $C^{3}$, and we
assume that we define an extension of $G^{\prime }\left( s\right) $ for $%
s\geq \mathring{\psi}_{1}^{2}\left( 0\right) $ as a function of class $C^{1}$
for all $r>0$. Though $G^{\prime }\left( s\right) $ may have a singular
point at $s=0,$ we assume that\ the function $G^{\prime }\left( \left\vert
\psi \right\vert ^{2}\right) \psi $ can be extended to a function of Holder
class $C^{\alpha }\left( \mathbb{C}\right) $\ for any $\alpha ,$ $0<\alpha
<1 $\ and the antiderivative $G_{1}\left( \left\vert \psi \right\vert
^{2}\right) ,$%
\begin{equation*}
G_{1}\left( s\right) =\int_{0}^{s}G_{1}^{\prime }\left( s^{\prime }\right)
ds^{\prime },\ G_{a}\left( s\right) =a^{-5}G_{1}\left( a^{3}s\right) ,
\end{equation*}%
is a function of class $C^{1+\alpha }\left( \mathbb{C}\right) $ for any $%
\alpha <1$ with respect to the variable $\psi \in \mathbb{C}$.\ In
particular, we assume that\ for every $C_{1}$ there exists such constant $C$
that the following inequalities hold: 
\begin{equation}
\left\vert G_{1}(\left\vert \psi \right\vert ^{2})\right\vert \leq
C\left\vert \psi \right\vert ^{1+\alpha }\text{,\ \ }\left\vert
G_{1}^{\prime }(\left\vert \psi \right\vert ^{2})\right\vert \left\vert \psi
\right\vert \leq C\left\vert \psi \right\vert ^{\alpha }\text{\ \ for\ \ }%
\left\vert \psi \right\vert \leq C_{1}.  \label{Gless}
\end{equation}

\begin{example}
\label{Ex1} Let the form factor $\mathring{\psi}_{1}\left( r\right) $ decay
as a power law, namely\ $\mathring{\psi}_{1}\left( r\right) =c_{\mathrm{pw}%
}\left( 1+r^{2}\right) ^{-p},$ $p>3/4$.\ Then 
\begin{equation}
G_{1}^{\prime }\left( \mathring{\psi}_{1}^{2}\right) =-2p\left( \mathring{%
\psi}_{1}/c_{\mathrm{pw}}\right) ^{2/p}\left( \left( \mathring{\psi}_{1}/c_{%
\mathrm{pw}}\right) ^{1/p}+\left( 2p+2\right) \left( \mathring{\psi}_{1}/c_{%
\mathrm{pw}}\right) ^{2/p}\right)  \label{expsi1}
\end{equation}%
where $c_{\mathrm{pw}}$ is the normalization factor. The same formula (\ref%
{expsi1}) can be used to obtain extension $G^{\prime }\left( s\right) $ for
all $s\geq 0$. In this example\ $G_{1}^{\prime }\left( s\right) $ is
differentiable for all $s\geq 0$ if $3/p\geq 2$.
\end{example}

\begin{example}
\label{Ex2}Consider an exponentially decaying form factor $\mathring{\psi}%
_{1}$ of the form 
\begin{equation}
\mathring{\psi}_{1}\left( r\right) =c_{\mathrm{e}}\mathrm{e}^{-\left(
r^{2}+1\right) ^{p}}.  \label{grs}
\end{equation}%
A direct computation yields%
\begin{eqnarray}
G^{\prime }\left( \mathring{\psi}^{2}\right) &=&-2p\left( \left( 2p+1\right)
\ln ^{\left( p-1\right) /p}\left( c_{\mathrm{e}}/\mathring{\psi}\right)
-2p\ln ^{\left( 2p-1\right) /p}\left( c_{\mathrm{e}}/\mathring{\psi}\right)
\right)  \label{grs1} \\
&&-2p\left( 2p\ln ^{\left( 2p-2\right) /p}\left( c_{\mathrm{e}}/\mathring{%
\psi}\right) +\left( 2-2p\right) \ln ^{\left( p-2\right) /p}\left( c_{%
\mathrm{e}}/\mathring{\psi}\right) \right) .  \notag
\end{eqnarray}%
In particular, for $p=1/2$ we obtain an exponentially decaying ground state\ 
$\mathring{\psi}_{1}\left( r\right) =c_{\mathrm{e}}\mathrm{e}^{-\left(
r^{2}+1\right) ^{1/2}}$\ and for $s\leq c_{\mathrm{e}}^{2}\mathrm{e}^{-2}$ 
\begin{equation}
G_{1}^{\prime }\left( s\right) =1-\frac{4}{\ln \left( c_{\mathrm{e}%
}^{2}/s\right) }-\frac{4}{\ln ^{2}\left( c_{\mathrm{e}}^{2}/s\right) }-\frac{%
8}{\ln ^{3}\left( c_{\mathrm{e}}^{2}/s\right) }.\text{ }  \label{ggkap}
\end{equation}%
We can extend it for larger $s$ as follows:%
\begin{equation}
G_{1}^{\prime }\left( s\right) =G_{1}^{\prime }\left( c_{\mathrm{e}}^{2}%
\mathrm{e}^{-2}\right) =-3\text{ if }s\geq 2c_{\mathrm{e}}^{2}\mathrm{e}%
^{-2}.  \label{ggkap1}
\end{equation}%
and in the interval $c_{\mathrm{e}}^{2}\mathrm{e}^{-2}\leq s\leq 2c_{\mathrm{%
e}}^{2}\mathrm{e}^{-2}$\ interpolate to obtain a smooth function for $s>0$.
The function $G_{1}^{\prime }\left( s\right) $ defined by (\ref{ggkap})\ has
a limit at $s=0$ and can be extended to a continuous function, but is not
differentiable at $s=0$.
\end{example}

\begin{example}
\label{Ex3}The \emph{Gaussian} ground state is given by the formula 
\begin{equation}
\mathring{\psi}\left( r\right) =C_{g}\mathrm{e}^{-r^{2}/2},\quad C_{g}=\pi
^{-3/4}.  \label{Gaussp}
\end{equation}%
Such a ground state is called \emph{gausson} in \cite{Bialynicki}. The
expression for $G^{\prime }$ can be derived from (\ref{grs1}) in the
particular case $p=1$, or can be evaluated directly: 
\begin{equation*}
\nabla ^{2}\mathring{\psi}\left( r\right) /\mathring{\psi}\left( r\right)
=r^{2}-3=-\ln \left( \mathring{\psi}^{2}\left( r\right) /C_{g}^{2}\right) -3.
\end{equation*}%
Hence, we define the nonlinearity by the formula%
\begin{equation}
G^{\prime }\left( |\psi |^{2}\right) =-\ln \left( |\psi
|^{2}/C_{g}^{2}\right) -3,  \label{Gaussg}
\end{equation}%
and refer to it as the \emph{logarithmic nonlinearity. }Dependence on the
size parameter $a>0$ is given by the formula 
\begin{equation}
G_{a}^{\prime }\left( |\psi |^{2}\right) =-a^{-2}\ln \left( a^{3}|\psi
|^{2}/C_{g}^{2}\right) -3a^{-2}  \label{Gpa}
\end{equation}%
with the antiderivative 
\begin{equation}
G\left( s\right) =G_{a}\left( s\right) =-a^{-2}s\left[ \ln \left(
a^{3}s\right) +\ln \pi ^{3/2}+2\right] ,\quad s\geq 0.  \label{paf3}
\end{equation}
\end{example}

\subsection{Conservation laws for Klein-Gordon equation\label{sreldis}}

The Lagrangian for the KG equation (\ref{KGex}) is the following
relativistic and gauge invariant expression 
\begin{equation}
\mathcal{L}_{1}\left( \psi \right) =\frac{\chi ^{2}}{2m}\left\{ \frac{1}{%
\mathrm{c}^{2}}\tilde{\partial}_{t}\psi \tilde{\partial}_{t}^{\ast }\psi
^{\ast }-\tilde{\nabla}\psi \cdot \tilde{\nabla}\psi ^{\ast }-\kappa
_{0}^{2}\psi ^{\ast }\psi -G\left( \psi ^{\ast }\psi \right) \right\}
\label{paf1}
\end{equation}%
where $\psi \left( t,\mathbf{x}\right) $ is a complex valued wave function,
and $\psi ^{\ast }$ is its complex conjugate. In the expression (\ref{paf1}) 
$\mathrm{c}$ is the speed of light, 
\begin{equation}
\kappa _{0}=m\mathrm{c}/\chi ,  \label{paf2}
\end{equation}%
and the covariant derivatives in (\ref{paf1}) are defined by (\ref{dtex}).\
The nonlinear KG equation (\ref{KGex}) is the Euler-Lagrange field equation
for the Lagrangian (\ref{paf1}).

Using the Noether theorem and the invariance of the Lagrangian (\ref{paf1})
with respect to Lorentz and gauge transformations one can derive expressions
for a number of conserved quantities, \cite{BF4}-\cite{BF7}. In particular,
the charge and the current densities (if $\mathbf{A=0}$) are as follows:%
\begin{gather}
\rho =-\frac{\chi q}{m\mathrm{c}^{2}}\left\vert \psi \right\vert ^{2}\func{Im%
}\frac{\tilde{\partial}_{t}\psi }{\psi }=-\left( \frac{\chi q}{m\mathrm{c}%
^{2}}\func{Im}\frac{\partial _{t}\psi }{\psi }+\frac{q^{2}}{m\mathrm{c}^{2}}%
\varphi \right) \left\vert \psi \right\vert ^{2},  \label{paf6} \\
\mathbf{J}=\frac{\chi q}{2m}\left( \psi ^{\ast }\nabla \psi -\psi \nabla
\psi ^{\ast }\right) =\frac{\chi q}{m}\left\vert \psi \right\vert ^{2}\func{%
Im}\frac{\nabla \psi }{\psi },  \label{paf6J}
\end{gather}%
and the above expressions are exactly the charge and the current sources in
the Maxwell equations, \cite{BF4}-\cite{BF7}. The expression for the energy,
momentum and external force densities\ are respectively 
\begin{equation}
\mathcal{E}\left( \psi \right) =\frac{\chi ^{2}}{2m}\left[ \frac{1}{\mathrm{c%
}^{2}}\tilde{\partial}_{t}\psi \tilde{\partial}_{t}^{\ast }\psi ^{\ast
}+\nabla \psi \cdot \nabla \psi ^{\ast }+G\left( \psi ^{\ast }\psi \right)
+\kappa _{0}^{2}\psi \psi ^{\ast }\right] ,  \label{emtn3}
\end{equation}%
\begin{equation}
\mathbf{P}=-\frac{\chi ^{2}}{2m\mathrm{c}^{2}}\left( \tilde{\partial}%
_{t}\psi \tilde{\nabla}^{\ast }\psi ^{\ast }+\tilde{\partial}_{t}^{\ast
}\psi ^{\ast }\tilde{\nabla}\psi \right) ,  \label{emtn4}
\end{equation}%
\begin{equation}
\mathbf{F}=\frac{\partial }{\partial x}\mathcal{L}_{1}=\frac{\chi q}{2m%
\mathrm{c}^{2}}\mathrm{i}\left( -\psi ^{\ast }\tilde{\partial}_{t}\psi +%
\tilde{\partial}\psi _{t}^{\ast }\psi ^{\ast }\right) \nabla \varphi =-\rho
\nabla \varphi .  \label{FLor}
\end{equation}%
Consequently, the total charge $\bar{\rho}$ and the energy $\mathcal{\bar{E}}
$ are 
\begin{equation}
\bar{\rho}=\int_{\mathbb{R}^{3}}\rho \left( t,\mathbf{x}\right) \,\mathrm{d}%
^{3}x,\quad \mathcal{\bar{E}}=\frac{\chi ^{2}}{2m}\int_{\mathbb{R}^{3}}%
\mathcal{E}\,\mathrm{d}^{3}x.  \label{Etot}
\end{equation}%
The charge conservation/continuity equation 
\begin{equation}
\partial _{t}\rho +\nabla \cdot \mathbf{J}=0  \label{contex}
\end{equation}%
is readily verified by multiplying the both sides of equation (\ref{KGex})
by $\psi ^{\ast }\ $and\ taking the imaginary part. The equation (\ref%
{contex}) in turn implies the total charge conservation:%
\begin{equation}
\int_{\mathbb{R}^{3}}\rho \left( t,\mathbf{x}\right) \,\mathrm{d}^{3}x=\bar{%
\rho}=\limfunc{const}=q.  \label{rhoq1}
\end{equation}%
The value $\bar{\rho}$ $=q$ is taken to ensure that Coulomb's potential with
this density has asymptotics $q/\left\vert \mathbf{x}\right\vert $ as $%
\left\vert \mathbf{x}\right\vert \rightarrow \infty $, \cite{BF4}--\cite{BF6}%
.

The energy conservation equation%
\begin{equation}
\partial _{t}\mathcal{E}=-\mathrm{c}^{2}\nabla \cdot \mathbf{P-}\nabla
\varphi \cdot \mathbf{J}  \label{equen2}
\end{equation}%
can be verified by multiplying (\ref{KGex}) by $\tilde{\partial}_{t}^{\ast
}\psi ^{\ast }$ defined by (\ref{dtex}), taking the real part and carrying
out elementary transformations, see Appendix \ref{apprelmath}. Similarly,
the the momentum conservation law takes the form%
\begin{equation}
\partial _{t}\mathbf{P}-\mathbf{F}+\nabla \mathcal{L}_{1}\left( \psi \right)
+\frac{\chi ^{2}}{2m}\tsum\nolimits_{j}\partial _{j}\left( \partial _{j}\psi
\nabla \psi ^{\ast }+\partial _{j}\psi ^{\ast }\nabla \psi \right) =0,
\label{momeq1}
\end{equation}%
where\ $\mathcal{L}_{1}\left( \psi \right) \ $ and $\mathbf{F}$ are defined
respectively by (\ref{paf1}) and (\ref{FLor}). It can be verified by
multiplying (\ref{KGex}) by $\tilde{\nabla}^{\ast }\psi ^{\ast }$, taking
the real part and carrying out elementary transformations, see Appendix \ref%
{apprelmath}. Integrating (\ref{equen2})\ and (\ref{momeq1}) with respect to 
$\mathbf{x}$\ we obtain equations for the total energy and momentum, namely 
\begin{equation}
\partial _{t}\mathcal{\bar{E}}=\mathbf{-}\int_{\mathbb{R}^{3}}\nabla \varphi
\cdot \mathbf{J}\,\mathrm{d}^{3}x,\;\partial _{t}\int_{\mathbb{R}^{3}}%
\mathbf{P}\,\mathrm{d}^{3}x=\int_{\mathbb{R}^{3}}\mathbf{F}\,\mathrm{d}^{3}x.
\label{enint1}
\end{equation}%
Obviously, the total energy and momentum are preserved if the external force
field $\nabla \varphi =\mathbf{0}$.

\subsection{Rest states and their energies \label{Snonlinrest}}

When the external EM field vanishes, that is $\varphi =0$, $\mathbf{A}=0$,
we define the rest states $\psi $ as time harmonic solutions to the KG
equation (\ref{KGex})%
\begin{equation}
\psi \left( t,\mathbf{x}\right) =\mathrm{e}^{-\mathrm{i}\omega t}\breve{\psi}%
\left( \mathbf{x}\right)  \label{psiom}
\end{equation}%
where $\breve{\psi}\left( \mathbf{x}\right) $ is assumed to be
central-symmetric. The substitution of (\ref{psiom}) in the KG equation (\ref%
{KGex}) yields the following \emph{nonlinear eigenvalue problem} 
\begin{equation}
\nabla ^{2}\breve{\psi}=G_{a}^{\prime }(|\breve{\psi}|^{2})\breve{\psi}+%
\mathrm{c}^{-2}\left( \omega _{0}^{2}-\omega ^{2}\right) \breve{\psi}=0.
\label{eiga}
\end{equation}%
The solution $\breve{\psi}$ must also satisfy the charge normalization
condition (\ref{rhoq1}) which takes the form 
\begin{equation}
\int |\breve{\psi}|^{2}\,\mathrm{d}^{3}x=\frac{\omega _{0}}{\omega },\qquad
\omega _{0}=\frac{m\mathrm{c}^{2}}{\chi }.  \label{normom}
\end{equation}%
The energy defined by (\ref{emtn3}), (\ref{Etot}) yields for the standing
wave (\ref{psiom}) the following expression 
\begin{equation}
\mathcal{\bar{E}}=\frac{\chi ^{2}}{2m}\int_{\mathbb{R}^{3}}\left[ \frac{1}{%
\mathrm{c}^{2}}\omega ^{2}\breve{\psi}\breve{\psi}^{\ast }+\kappa _{0}^{2}%
\breve{\psi}\breve{\psi}^{\ast }+\nabla \breve{\psi}\nabla \breve{\psi}%
^{\ast }+G_{a}(\breve{\psi}\breve{\psi}^{\ast })\right] \,\mathrm{d}^{3}x.
\label{Estand}
\end{equation}%
The problem (\ref{eiga}), (\ref{normom}) has a sequence of solutions with
the corresponding sequence of frequencies $\omega $. Their energies $%
\mathcal{\bar{E}}_{0\omega }$ are related to the frequency $\omega $ by the
formula 
\begin{equation}
\mathcal{\bar{E}}_{0\omega }=\chi \omega \left( 1+\Theta \left( \omega
\right) \right) ,  \label{enomega}
\end{equation}%
\begin{equation}
\Theta \left( \omega \right) =\Theta _{0}\frac{a_{\mathrm{C}}^{2}}{a^{2}}%
\frac{\omega _{0}^{2}}{\omega ^{2}},\qquad a_{\mathrm{C}}=\frac{\chi }{m%
\mathrm{c}}.  \label{Thomega}
\end{equation}%
Here the coefficient $\Theta _{0}=\frac{1}{3}\left\Vert \nabla \mathring{\psi%
}_{1}\right\Vert ^{2}$ depends on the shape of the rest charge, the\
parameter $a_{\mathrm{C}}=\frac{\hbar }{m\mathrm{c}}=\frac{\lambda _{\mathrm{%
C}}}{2\pi }$ coincides in the physical applications with the \emph{reduced
Compton wavelength} of a particle with a mass $m$ if $\chi =\hbar $. For the
case\ of the logarithmic nonlinearity the ground state Gaussian shape is
given by (\ref{Gaussp}) and\ $\Theta _{0}=1/2$.

In the case of the logarithmic nonlinearity the original nonlinear
eigenvalue problem (\ref{eiga}) can be reduced by a change of variables to\
the following nonlinear eigenvalue problem with only one eigenvalue
parameter $\xi $ and the parameter-independent constraint: 
\begin{equation}
\nabla ^{2}\breve{\psi}_{1}=G_{1}^{\prime }(|\breve{\psi}_{1}|^{2})\breve{%
\psi}_{1}-\xi \breve{\psi}_{1},\quad \int_{\mathbb{R}^{3}}|\breve{\psi}%
_{1}|^{2}\,\mathrm{d}^{3}x=1.  \label{eipro1}
\end{equation}%
The parameter $\xi $ is related to the parameters in (\ref{eiga}) by the
formula 
\begin{equation}
\xi =\frac{a^{2}}{a_{\mathrm{C}}^{2}}\left( \frac{\omega ^{2}}{\omega
_{0}^{2}}-1\right) -\frac{1}{2}\ln \frac{\omega ^{2}}{\omega _{0}^{2}}.
\label{ksi}
\end{equation}%
The eigenvalue problem (\ref{eipro1}) has infinitely many solutions $\left(
\xi _{n},\psi _{1n}\right) $, $n=0,1,2...$, representing localized charge
distributions. The energy of $\psi _{n}$, $n>0$, is higher than the energy
of the Gaussian ground state which corresponds to $\xi =\xi _{0}=0$ and has
the lowest possible energy. These solutions coincide with critical points of
the energy\ functional under the constraint, for mathematical details see 
\cite{Cazenave83}, \cite{BerestyckiLions83I}, \cite{BerestyckiLions83II}.
The next two values of $\xi $ for the radial rest states are approximately $%
2.17$ and $3.41$ according to \cite{Bialynicki1}.

Applying to the rest solution the Lorentz transformation, one can easily
obtain a solution which represents the charge moving with a constant
velocity $\mathbf{v}$ (see \cite{BF4}, \cite{BF5} and the following Section %
\ref{Suniform}).

\subsection{Uniform motion of a charge\label{Suniform}}

The free motion of a charge is governed by the KG equation (\ref{KGex}) with
vanishing external EM field, that is $\varphi =0$, $\mathbf{A}=0$. Since the
KG equation is relativistic covariant, the solution can be obtained from a
rest solution defined by (\ref{psiom}), (\ref{eiga}) by applying the Lorentz
boost transformation as in \cite{BF4}, \cite{BF6}. Consequently, the
solution to the KG equation (\ref{KGex}) for a free charged particle that
moves with a constant velocity $\mathbf{v}$ takes the form 
\begin{equation}
\psi \left( t,\mathbf{x}\right) =\psi _{\mathrm{free}}\left( t,\mathbf{x}%
\right) =\mathrm{e}^{-\mathrm{i}\left( \gamma \omega t-\mathbf{k}\cdot 
\mathbf{x}\right) }\breve{\psi}\left( \mathbf{x}^{\prime }\right) ,
\label{mvch1}
\end{equation}%
with $\breve{\psi}\left( \mathbf{x}^{\prime }\right) $ satisfying equation (%
\ref{eiga}) and%
\begin{equation}
\breve{\psi}\left( \mathbf{x}^{\prime }\right) =\breve{\psi}_{a}\left( 
\mathbf{x}^{\prime }\right) =a^{-3/2}\breve{\psi}_{1}\left( \mathbf{x}%
^{\prime }/a\right) ,  \label{psia}
\end{equation}%
\begin{equation}
\mathbf{x}^{\prime }=\mathbf{x}+\frac{\left( \gamma -1\right) }{v^{2}}\left( 
\mathbf{v}\cdot \mathbf{x}\right) \mathbf{v}-\gamma \mathbf{v}t,\quad 
\mathbf{k}=\gamma \omega \frac{\mathbf{v}}{\mathrm{c}^{2}},  \label{mvch4}
\end{equation}%
where $\gamma $ is the Lorentz factor 
\begin{equation}
\gamma =\left( 1-\mathbf{\beta }^{2}\right) ^{-1/2},\quad \mathbf{\beta }=%
\mathrm{c}^{-1}\mathbf{v}.  \label{gam}
\end{equation}%
Consequently, all quantities of interest for a free charge can be written
explicitly. Namely, the charge density $\rho $ defined by the relation (\ref%
{paf6}), the total charge $\bar{\rho}$ and the total energy $\mathsf{E}$
equal respectively 
\begin{equation}
\rho =\gamma q|\breve{\psi}\left( \mathbf{x}^{\prime }\right) |^{2},\quad 
\bar{\rho}=\tint \rho \left( \mathbf{x}\right) \,\mathrm{d}^{3}x=q,\quad 
\mathcal{\bar{E}}=\gamma m\mathrm{c}^{2}\left( 1+\Theta \left( \omega
\right) \right) ,  \label{rhofree}
\end{equation}%
where $\Theta \left( \omega \right) $ is given by (\ref{Thomega}). The
current density and the total current $\mathbf{\bar{J}}$ for the free charge
equal respectively%
\begin{equation}
\mathbf{J}=\frac{q}{m}\chi \func{Im}\frac{\nabla \psi }{\psi }\left\vert
\psi \right\vert ^{2}=\gamma \chi \frac{q}{m}\omega \frac{1}{\mathrm{c}^{2}}%
\mathbf{v|}\breve{\psi}|^{2}\left( \mathbf{x}^{\prime }\right) ,  \label{Pnr}
\end{equation}%
\begin{equation}
\mathbf{\bar{J}}=\int_{\mathbb{R}^{3}}\mathbf{J}\left( \mathbf{x}\right) \,%
\mathrm{d}^{3}x=q\mathbf{v.}  \label{Jfree}
\end{equation}%
The total momentum is given by the following formula 
\begin{equation}
\mathbf{\bar{P}}=\gamma \mathbf{v}m\left( 1+\Theta \left( \omega \right)
\right) =M\mathbf{v},\quad M=\gamma m\left( 1+\Theta \left( \omega \right)
\right) .  \label{momfree}
\end{equation}%
Given the above kinematic representation $\mathbf{\bar{P}}=M\mathbf{v}$ of
the momentum, it is natural, \cite[Sec. 3.3]{Moller}, \cite[Sec. 37]{Pauli
RT}, \cite{BenciF}, to identify the coefficient $M$ as the mass and to
define the \emph{rest mass} $m_{0}$\ of the charge by the formula 
\begin{equation}
m_{0}=m\left( 1+\Theta \left( \omega \right) \right) ,\quad \Theta \left(
\omega \right) =\Theta _{0}a_{\mathrm{C}}^{2}/a^{2},  \label{Mfree}
\end{equation}%
and the expression (\ref{rhofree}) for the energy takes the form of
Einstein's mass-energy relation $\mathcal{\bar{E}}=M\mathrm{c}^{2}$.

A direct comparison shows that the above definition based on the Lorentz
invariance of a uniformly moving free charge is fully consistent with the
definition of the inertial mass which is derived from the analysis of the
accelerated motion of localized charges in an external EM field in the
following section, see Remark \ref{Rfreemass} for a more detailed
discussion.\ 

\section{Relativistic dynamics of localized solutions\label{SReldyn}}

The point dynamics described by (\ref{fre1}) involves only fields exactly at
the location of the point. Therefore it is natural to make assumptions of
localization of KG equations and their solutions which should involve only
behavior of the fields in a vicinity of the trajectory.\ \ The speed of
light $\mathrm{c}$, the charge $q$\ and the mass parameter $m$ are assumed
to be fixed.\ There are two intrinsic length scales relevant for the
nonlinear KG equation: the size parameter $a$ that enters the nonlinearity
and the quantity $a_{\mathrm{C}}=\frac{\chi }{m\mathrm{c}}$ known as the
reduced Compton wavelength (if $\chi $ equals the Planck constant $\hbar $).
In our analysis we suppose the parameter $a_{\mathrm{C}}$ and the size
parameter $a$ to become vanishingly small by taking a sequence of values $%
\chi _{n}\rightarrow 0$, $a_{n}\rightarrow 0$, $n=1,2,...$,\ when the
corresponding values of the potential $\varphi _{n}$\ and solutions $\psi
_{n}\left( t,\mathbf{x}\right) $ of the KG\ equation are defined in
contracting neighborhoods of a trajectory $\mathbf{\hat{r}}\left( t\right) $.

\subsection{ Concentration at a trajectory\label{SConctraj}}

In this section we define and give examples of \emph{solutions concentrating
at a trajectory}.\ \ A trajectory $\mathbf{\hat{r}}\left( t\right) $, $%
T_{-}\leq t\leq T_{+}$ is a twice continuously differentiable function with
values in $\mathbb{R}^{3}$ satisfying 
\begin{equation}
\left\vert \mathbf{\hat{v}}\right\vert ,\left\vert \partial _{t}\mathbf{\hat{%
v}}\right\vert \leq C,\text{\ where }\mathbf{\hat{v}}\left( t\right)
=\partial _{t}\mathbf{\hat{r}}\left( t\right) ,\text{\quad }T_{-}\leq t\leq
T_{+}.  \label{trbound}
\end{equation}%
In this section the trajectory $\mathbf{\hat{r}}\left( t\right) $ is assumed
to be fixed.\ Being given a trajectory $\mathbf{\hat{r}}\left( t\right) $,
we consider an associated with it family of neighborhoods contracting to it.
Namely, we introduce the ball of radius $R$ centered at $r=\mathbf{\hat{r}}%
\left( t\right) $:\ 
\begin{equation}
\Omega \left( \mathbf{\hat{r}}\left( t\right) ,R\right) =\left\{ \mathbf{x}%
:\left\vert \mathbf{x}-\mathbf{\hat{r}}\left( t\right) \right\vert ^{2}\leq
R^{2}\right\} \subset \mathbb{R}^{3},\text{\quad }R>0  \label{Omball}
\end{equation}%
and then for a sequence of positive $R_{n}\rightarrow 0$ we consider the
sequence 
\begin{equation}
\Omega _{n}=\Omega _{n}\left( t\right) =\Omega \left( \mathbf{\hat{r}}%
(t),R_{n}\right) .  \label{Omn}
\end{equation}%
In what follows we often make use of the following elementary identity%
\begin{equation}
\int_{\Omega _{n}}\partial _{t}f\left( t,\mathbf{x}\right) \,\mathrm{d}%
^{3}x=\partial _{t}\int_{\Omega _{n}}f\left( t,\mathbf{x}\right) \,\mathrm{d}%
^{3}x-\int_{\partial \Omega _{n}}f\left( t,\mathbf{x}\right) \mathbf{\hat{v}}%
\cdot \mathbf{\bar{n}}\,\mathrm{d}^{3}x  \label{eldt}
\end{equation}%
where $\mathbf{\bar{n}}$ is the external normal\ to\ $\partial \Omega _{n}$, 
$\mathbf{\hat{v}}=\partial _{t}\mathbf{\hat{r}}$.

\begin{definition}[concentrating neighborhoods]
\label{Dconcne}\emph{\ Concentrating neighborhoods} of a trajectory $\mathbf{%
\hat{r}}\left( t\right) $ are\ defined as a family of tubular domains 
\begin{equation}
\hat{\Omega}\left( \mathbf{\hat{r}}\left( t\right) ,R_{n}\right) =\left\{
\left( t,\mathbf{x}\right) :\left\vert x-\mathbf{\hat{r}}\left( t\right)
\right\vert ^{2}\leq R_{n}^{2}\right\} \subset \left[ T_{-},T_{+}\right]
\times \mathbb{R}^{3}  \label{Omhat}
\end{equation}%
where $R_{n}$ satisfy contraction condition:%
\begin{equation}
R_{n}\rightarrow 0\text{\ as\ }n\rightarrow \infty .  \label{Rnto0}
\end{equation}
\end{definition}

We will consider the KG equations and their solutions in concentrating
neighborhoods of a trajectory $\mathbf{\hat{r}}\left( t\right) $. We will
make certain regularity assumptions on the behavior of the potentials and
solutions restricted to the domain $\hat{\Omega}\left( \mathbf{\hat{r}}%
,R_{n}\right) $ around the trajectory\ $\mathbf{\hat{r}}\left( t\right) $.\
The assumptions relate the "microscopic" scales $a$ and $a_{\mathrm{C}}$ to
the "macroscopic" length scale of the order $1$ relevant to the trajectory $%
\mathbf{\hat{r}}(t)$ and the electric potential $\varphi \left( t,\mathbf{x}%
\right) $.

Below we often make use of the following notations%
\begin{equation}
\partial _{0}=\mathrm{c}^{-1}\partial _{t},  \label{dt0}
\end{equation}%
\begin{gather}
\nabla _{\mathbf{x}}\varphi =\nabla \varphi =\left( \partial _{1}\varphi
,\partial _{2}\varphi ,\partial _{3}\varphi \right) ,\text{\quad }\left\vert
\nabla _{\mathbf{x}}\varphi \right\vert ^{2}=\left\vert \nabla \varphi
\right\vert ^{2}=\left\vert \partial _{1}\varphi \right\vert ^{2}+\left\vert
\partial _{2}\varphi \right\vert ^{2}+\left\vert \partial _{3}\varphi
\right\vert ^{2},  \label{grad0} \\
\nabla _{0,\mathbf{x}}\varphi =\left( \partial _{0}\varphi ,\partial
_{1}\varphi ,\partial _{2}\varphi ,\partial _{3}\varphi \right) ,\text{\quad 
}\left\vert \nabla _{0,\mathbf{x}}\varphi \right\vert ^{2}=\left\vert
\partial _{0}\varphi \right\vert ^{2}+\left\vert \partial _{1}\varphi
\right\vert ^{2}+\left\vert \partial _{2}\varphi \right\vert ^{2}+\left\vert
\partial _{3}\varphi \right\vert ^{2}.  \notag
\end{gather}

\begin{definition}[localized KG equations]
\label{DKGseq}Let $\mathbf{\hat{r}}(t)$ be a trajectory with its
concentrating neighborhoods $\hat{\Omega}\left( \mathbf{\hat{r}}%
,R_{n}\right) $, and let $\chi =\chi _{n}$, $a=a_{n}$, $\varphi =\varphi
_{n} $\ be a sequence of parameters defining the KG equation (\ref{KGex}).
We call a sequence of the KG equations localized in $\hat{\Omega}\left( 
\mathbf{\hat{r}},R_{n}\right) $ if the following conditions are satisfied.
First of all, $\chi $ and $a$ become vanishingly small\ 
\begin{equation}
a=a_{n}\rightarrow 0,\text{\quad }\chi =\chi _{n}\rightarrow 0,
\label{anto0}
\end{equation}%
the ratio $\zeta =a_{\mathrm{C}}/a$ remains bounded, 
\begin{equation}
\zeta =\frac{a_{\mathrm{C}}}{a}\leq C,\text{\ where }a_{\mathrm{C}}=\frac{%
\chi }{m\mathrm{c}},  \label{acan}
\end{equation}%
and the ratio $\theta =R/a$\ grows to infinity 
\begin{equation}
\theta _{n}=\frac{R_{n}}{a_{n}}\rightarrow \infty \ \text{as\ }n\rightarrow
\infty .  \label{theninf}
\end{equation}%
Second of all, the electric potentials $\varphi _{n}\left( t,\mathbf{x}%
\right) $ are twice continuously differentiable in $\hat{\Omega}\left( 
\mathbf{\hat{r}},R_{n}\right) $ and satisfy the following constraints:

\begin{enumerate}
\item[(i)] the following uniform in $n$ estimate holds 
\begin{equation}
\max_{T_{-}\leq t\leq T_{+},x\in \Omega _{n}}(|\varphi _{n}\left( t,\mathbf{x%
}\right) \mathbf{|}+\mathbf{|}\nabla _{0,\mathbf{x}}\varphi _{n}\left( t,%
\mathbf{x}\right) \mathbf{\mathbf{|})}\leq C;  \label{Abound}
\end{equation}

\item[(ii)] the potentials $\varphi _{n}\left( t,\mathbf{x}\right) $ locally
converge to a linear potential $\varphi _{\infty }\left( t,\mathbf{x}\right) 
$, namely%
\begin{equation}
\max_{T_{-}\leq t\leq T_{+},x\in \Omega _{n}}(|\varphi \left( t,\mathbf{x}%
\right) -\varphi _{\infty }\left( t,\mathbf{x}\right) \mathbf{|}+|\nabla
_{0,x}\varphi \left( t,\mathbf{x}\right) -\nabla _{0,x}\varphi _{\infty
}\left( t,\mathbf{x}\right) \mathbf{\mathbf{\mathbf{|}})}\rightarrow 0\text{%
\ \ as\ \ }n\rightarrow \infty ,  \label{Alocr}
\end{equation}%
\ where 
\begin{equation}
\varphi _{\infty }\left( t,\mathbf{x}\right) =\varphi _{\infty }\left( t,%
\mathbf{\hat{r}}\right) +\left( \mathbf{x}-\mathbf{\hat{r}}\right) \nabla
\varphi _{\infty }\left( t\right) ,  \label{fiinf}
\end{equation}%
\ and the coefficients $\varphi _{\infty }\left( t,\mathbf{\hat{r}}\right) $
and $\nabla \varphi _{\infty }\left( t\right) $ are bounded 
\begin{equation}
\left\vert \varphi _{\infty }\left( t,\mathbf{\hat{r}}\right) \right\vert ,%
\text{\ }\left\vert \nabla \varphi _{\infty }\right\vert ,\text{\ }%
\left\vert \partial _{t}\nabla \varphi _{\infty }\right\vert \leq C\text{ for%
}\;T_{-}\leq t\leq T_{+}.  \label{fihatb}
\end{equation}
\end{enumerate}
\end{definition}

\ Since the parameters $m,\mathrm{c}$, $q$ are fixed, the following
inequalities always hold 
\begin{equation}
C_{0}\chi \leq a_{\mathrm{C}}\leq C_{1}\chi ,\text{\quad }\chi _{n}\leq
C_{2}a_{n}.  \label{acan1}
\end{equation}%
Based on the above inequalities we often replace $\chi $ with $a_{\mathrm{C}%
} $ in estimates without special comments. Throughout the paper we denote
constants which do not depend on $\chi _{n},$ $a_{n},$ and $\theta _{n}$ by
the letter $C\ $with\ different\ indices. Sometimes the same letter $C$ with
the same indices may denote in different formulas different constants. Below
we often omit index $n$ in $\chi _{n}$, $a_{n}$, $\varphi _{n}$\ etc.

Obviously, if the potentials $\varphi _{n}$\ are the restrictions of a fixed
twice continuously differentiable function $\varphi $ to $\hat{\Omega}\left( 
\mathbf{\hat{r}}(t),R_{n}\right) $ then conditions (\ref{Abound}) and (\ref%
{Alocr}) are satisfied with $\varphi _{\infty }\left( t,\mathbf{x}\right)
=\varphi \left( t,\mathbf{\hat{r}}\right) +\left( \mathbf{x}-\mathbf{\hat{r}}%
\right) \nabla \varphi \left( t,\mathbf{\hat{r}}\right) $.

\begin{definition}[concentrating solutions]
\label{Dlocconverge}Let $\mathbf{\hat{r}}(t)$ be a trajectory. We say that
solutions $\psi $ of\ the KG\ equation (\ref{KGex}) concentrate\ at\ the
trajectory $\mathbf{\hat{r}}(t)$ if the following conditions hold. First of
all, a sequence of concentrating neighborhoods $\hat{\Omega}\left( \mathbf{%
\hat{r}},R_{n}\right) $ and parameters $a=a_{n},$ $\chi =\chi _{n},$\
potentials$\ \varphi =\varphi _{n},\ $are selected to form a sequence of the
KG equations (\ref{KGex}) localized in $\hat{\Omega}_{n}\left( \mathbf{\hat{r%
}},R_{n}\right) $ as in Definition \ref{DKGseq}. Second of all, for every
domain $\hat{\Omega}\left( \mathbf{\hat{r}},R_{n}\right) $ there exists a
function $\psi =\psi _{n}\in C^{2}\left( \hat{\Omega}\left( \mathbf{\hat{r}}%
,R_{n}\right) \right) $ which is a solution to the KG equation (\ref{KGex})
in $\hat{\Omega}\left( \mathbf{\hat{r}},R_{n}\right) $, and this solution
satisfies the following constraints:

\begin{enumerate}
\item[(i)] there exists a constant $C$ which does not depend on $n$\ but may
depend on the sequence such that\ 
\begin{equation}
\max_{T_{-}\leq t\leq T_{+}}\int_{\Omega _{n}}\left[ a_{\mathrm{C}%
}^{2}\left\vert \nabla _{0,\mathbf{x}}\psi \left( t,\mathbf{x}\right)
\right\vert ^{2}+a_{\mathrm{C}}^{2}|G(\left\vert \psi \left( t,\mathbf{x}%
\right) \right\vert ^{2})|+\left\vert \psi \left( t,\mathbf{x}\right)
\right\vert ^{2}\right] \,\mathrm{d}^{3}x\leq C,  \label{locpsibound}
\end{equation}%
where $G=G_{a}$ is the nonlinearity in (\ref{paf1});

\item[(ii)] the restriction of functions $\psi $ to the boundary\ $\partial
\Omega _{n}=\left\{ \left\vert \mathbf{x}-\mathbf{\hat{r}}\right\vert
=R_{n}\right\} $ vanishes asymptotically: 
\begin{equation}
\max_{T_{-}\leq t\leq T_{+}}\int_{\partial \Omega _{n}}(a_{\mathrm{C}%
}^{2}\left\vert \nabla _{0,\mathbf{x}}\psi \left( t,\mathbf{x}\right)
\right\vert ^{2}+a_{\mathrm{C}}^{2}|G(\left\vert \psi \left( t,\mathbf{x}%
\right) \right\vert ^{2})|+\left\vert \psi \left( t,\mathbf{x}\right)
\right\vert ^{2})\,\mathrm{d}^{2}\sigma \rightarrow 0;  \label{psiouta}
\end{equation}

\item[(iii)] the restricted to $\Omega \left( \mathbf{\hat{r}}%
(t),R_{n}\right) $ energy 
\begin{equation}
\mathcal{\bar{E}}_{n}\left( t\right) =\int_{\Omega \left( \mathbf{\hat{r}}%
(t),R_{n}\right) }\mathcal{E}_{n}\,\mathrm{d}^{3}x,  \label{Endef0}
\end{equation}%
with $\mathcal{E}_{n}$ being the energy density defined by expression (\ref%
{emtn3}), is bounded from below\ for sufficiently large $n$: 
\begin{equation}
\mathcal{\bar{E}}_{n}\left( t\right) \geq c_{0}>0\text{ for}\;n\geq
n_{0},\;T_{-}\leq t\leq T_{+};  \label{Egrc}
\end{equation}

\item[(iv)] there exists $t_{0}\in \left[ T_{-},T_{+}\right] $ such that the
sequence of restricted energies $\mathcal{\bar{E}}_{n}\left( t_{0}\right) $\
converges: 
\begin{equation}
\lim_{n\rightarrow \infty }\mathcal{\bar{E}}_{n}\left( t_{0}\right) =%
\mathcal{\bar{E}}_{\infty }\left( t_{0}\right) .  \label{Econv}
\end{equation}
\end{enumerate}
\end{definition}

Notice that condition (i) provides for the boundedness of the restricted
energy over domains $\Omega _{n}$. Condition (ii) provides for a proper
confinement of $\psi $ to $\Omega _{n}$ and condition (iii) ensures that\
the sequence is non-trivial. According to (\ref{Eless}) and (\ref%
{locpsibound}), $\mathcal{\bar{E}}_{n}\left( t_{0}\right) $ is a bounded
sequence, and, consequently, it always contains a converging subsequence.
Hence, condition (iv) is not really an additional constraint but rather it
assumes that such a subsequence is selected. The choice of a particular \
subsequence limit $\mathcal{\bar{E}}_{\infty }\left( t_{0}\right) $ then can
be interpreted as a normalization, see Remark \ref{Runiq}. Obviously, this
condition describes the amount of energy which concentrates at the
trajectory at the time $t_{0}$.

\begin{example}
\label{Rboundint}The conditions of the Definition \ref{Dlocconverge} can be
easily verified for rather general sequences of functions which are
localized around an arbitrary trajectory $\mathbf{\hat{r}}(t)$ (if we do not
assume that the functions are solutions of the KG equation). In particular,
let a sequence of functions be defined by the formula 
\begin{equation}
\psi \left( t,\mathbf{x}\right) =a^{-3/2}\psi _{\mathrm{0}}\left( \left( 
\mathbf{x}-\mathbf{\hat{r}}(t)\right) /a\right) ,\quad a=a_{n},
\label{psiex}
\end{equation}%
$a$\ is the same parameter as in $G_{a}$, the function $\psi _{\mathrm{0}%
}\left( z\right) $ is a smooth fixed function which decays together with its
derivatives as $\left\vert \mathbf{z}\right\vert \rightarrow \infty $: 
\begin{equation}
\max_{\left\vert z\right\vert =\theta }\left( \left\vert \nabla _{z}\psi _{%
\mathrm{0}}\left( \mathbf{z}\right) \right\vert ^{2}+\left\vert \psi _{%
\mathrm{0}}\left( \mathbf{z}\right) \right\vert ^{2}\right) \leq Y^{2}\left(
\theta \right) ,\   \label{Yz}
\end{equation}%
where $Y\left( \theta \right) $ is a continuous function which decays fast
enough: 
\begin{equation}
Y\left( \theta \right) \leq C_{0}\theta ^{-N}\text{\ \ if\ }\theta \geq
1;\quad Y\left( \theta \right) \leq C_{0}\;\ \text{if}\quad \theta \leq 1,
\label{Ythe}
\end{equation}%
with sufficiently large $N>3/2$. We assume that $N$ and the sequences $%
\theta _{n}\rightarrow \infty $, $a=a_{n}\rightarrow 0$\ satisfy the
condition 
\begin{equation}
3-\left( 1+\alpha \right) N<0,\quad a^{-1}\theta _{n}^{2-\left( 1+\alpha
\right) N}\rightarrow 0  \label{anthen}
\end{equation}%
with $1>\alpha >0$, $\alpha $ close to $1$ as in (\ref{Gless}).\ To verify (%
\ref{locpsibound}) we change variables $\left( \mathbf{x}-\mathbf{\hat{r}}%
(t)\right) /a=\mathbf{z}$, take into account that according to (\ref{acan}) $%
\chi ^{2}/a^{2}\leq C$, and obtain inequalities 
\begin{gather*}
a_{\mathrm{C}}^{2}\int_{\Omega \left( \mathbf{\hat{r}}(t),R_{n}\right)
}\left\vert \nabla _{x}\psi \left( t,\mathbf{x}\right) \right\vert ^{2}+%
\frac{1}{\mathrm{c}^{2}}\left\vert \partial _{t}\psi \left( t,\mathbf{x}%
\right) \right\vert ^{2}\,\mathrm{d}^{3}x\leq C\int_{\mathbf{z}\in \Omega
\left( 0,\theta _{n}\right) }\left\vert \nabla _{z}\psi _{\mathrm{0}}\left(
t,\mathbf{z}\right) \right\vert ^{2}\left( 1+\left\vert \partial _{t}\mathbf{%
\hat{r}}\right\vert ^{2}\right) \,\mathrm{d}^{3}z\leq \\
\leq C_{1}+C_{1}\int_{1}^{\infty }r^{2}r^{-2N}dr\leq C_{1}^{\prime },
\end{gather*}%
\begin{equation*}
\int_{\Omega \left( \mathbf{\hat{r}}(t),R_{n}\right) }\left\vert \psi \left(
t,\mathbf{x}\right) \right\vert ^{2}\,\mathrm{d}^{3}x=\int_{\mathbf{z}\in
\Omega \left( 0,\theta _{n}\right) }\left\vert \psi _{\mathrm{0}}\left( t,%
\mathbf{z}\right) \right\vert ^{2}\,\mathrm{d}^{3}z\leq C_{2}.
\end{equation*}%
Now we estimate the integral which involves the nonlinearity%
\begin{equation*}
G_{a}(\left\vert \psi \left( t,\mathbf{x}\right) \right\vert
^{2})=a^{-5}G_{1}(\left\vert \psi _{\mathrm{0}}\left( \left( \mathbf{x}-%
\mathbf{\hat{r}}(t)\right) /a\right) \right\vert
^{2})=a^{-5}G_{1}(\left\vert \psi _{\mathrm{0}}\left( \mathbf{z}\right)
\right\vert ^{2}).
\end{equation*}%
According to (\ref{Gless}) and (\ref{anthen}) 
\begin{gather}
a_{\mathrm{C}}^{2}\int_{\Omega \left( \mathbf{\hat{r}}(t),R_{n}\right)
}\left\vert G_{a}(\left\vert \psi \left( t,\mathbf{x}\right) \right\vert
^{2})\right\vert \,\mathrm{d}^{3}x=\zeta ^{2}\int_{\Omega \left( 0,\theta
_{n}\right) }\left\vert G_{a}(\left\vert \psi _{\mathrm{0}}\left( t,\mathbf{z%
}\right) \right\vert ^{2})\right\vert \,\mathrm{d}^{3}z  \label{Gen} \\
\leq C\zeta ^{2}\int_{\Omega \left( 0,\theta _{n}\right) }\left\vert \psi _{%
\mathrm{0}}\left( t,\mathbf{z}\right) \right\vert ^{1+\alpha }\,\mathrm{d}%
^{3}z\leq \zeta ^{2}C_{1}+\zeta ^{2}C_{2}\int_{1}^{\theta _{n}}\theta
^{-N\left( 1+\alpha \right) +2}d\theta \leq \zeta ^{2}C_{3}.  \notag
\end{gather}%
Hence condition (\ref{locpsibound}) is fulfilled. To verify (\ref{psiouta})
we estimate integrals over the boundary using (\ref{anthen}) and (\ref{Gless}%
): 
\begin{gather*}
\int_{\partial \Omega \left( \mathbf{\hat{r}}(t),R_{n}\right) }a_{\mathrm{C}%
}^{2}\left\vert \nabla _{x}\psi \left( t,\mathbf{x}\right) \right\vert
^{2}+a_{\mathrm{C}}^{2}\frac{1}{\mathrm{c}^{2}}\left\vert \partial _{t}\psi
\left( t,\mathbf{x}\right) \right\vert ^{2}+\left\vert \psi \left( t,\mathbf{%
x}\right) \right\vert ^{2}\,\mathrm{d}^{2}\sigma \\
\leq Ca^{-1}\int_{\partial \Omega \left( 0,\theta _{n}\right) }\left\vert
\nabla _{z}\psi _{\mathrm{0}}\left( t,\mathbf{z}\right) \right\vert
^{2}+\left\vert \psi _{\mathrm{0}}\left( t,\mathbf{z}\right) \right\vert
^{2}\,\mathrm{d}^{2}\sigma \leq C_{1}\theta _{n}^{2}Y^{2}\left( \theta
_{n}\right) a^{-1}\leq C,
\end{gather*}%
\begin{gather*}
a_{\mathrm{C}}^{2}\int_{\partial \Omega \left( \mathbf{\hat{r}}%
(t),R_{n}\right) }|G_{a}(\left\vert \psi \left( t,\mathbf{x}\right)
\right\vert ^{2})|\,\mathrm{d}^{2}\sigma =\zeta ^{2}a^{-1}\int_{\partial
\Omega \left( 0,\theta _{n}\right) }|G_{1}(\left\vert \psi _{\mathrm{0}%
}\left( t,\mathbf{z}\right) \right\vert ^{2})|\,\mathrm{d}^{2}\sigma \\
\leq C\zeta ^{2}a^{-1}\int_{\partial \Omega \left( 0,\theta _{n}\right)
}\left\vert \psi _{\mathrm{0}}\left( t,\mathbf{z}\right) \right\vert
^{1+\alpha }\,\mathrm{d}^{2}\sigma \leq C_{1}\zeta ^{2}a^{-1}\theta
_{n}^{2}Y^{1+\alpha }\left( \theta _{n}\right) \leq C_{2}a^{-1}\theta
_{n}^{2}\theta _{n}^{-N\left( 1+\alpha \right) }\leq C_{2}^{\prime }.
\end{gather*}%
Consequently, we obtain\ the desired (\ref{psiouta}). The energy $\mathcal{%
\bar{E}}_{n}$ involves the positive term 
\begin{equation*}
\frac{\chi ^{2}}{2m}\kappa _{0}^{2}\int_{\Omega \left( \mathbf{\hat{r}}%
(t),R_{n}\right) }\left\vert \psi \left( t,\mathbf{x}\right) \right\vert
^{2}=\frac{m\mathrm{c}^{2}}{2}\int_{\Omega \left( 0,\theta _{n}\right)
}\left\vert \psi _{\mathrm{0}}\right\vert ^{2}d\mathbf{z}\rightarrow \frac{m%
\mathrm{c}^{2}}{2}\int_{\mathbb{R}^{3}}\left\vert \psi _{\mathrm{0}%
}\right\vert ^{2}d\mathbf{z\geq }\frac{1}{C_{4}}>0.
\end{equation*}%
According to (\ref{emtn3}), \ the only negative contribution to the energy
can come from the nonlinearity $G$, and, according to (\ref{Gen}), this
contribution cannot be greater than $\zeta ^{2}C_{3},$ \ hence condition (%
\ref{Egrc}) is satisfied if $\zeta $ is small enough.
\end{example}

In the following important example we consider uniform motion of charges
from Section \ref{Suniform} where the relativistic argument is used to
determine their mass. The uniformly moving solutions are also an example of
concentrating solutions.

\begin{example}
\label{Euniform}As a simple example of solutions which concentrate\ at\ a
trajectory $\mathbf{\hat{r}}(t)$ we take solutions defined by (\ref{mvch1})
in Section \ref{Suniform}.\ Then\ the trajectory $\mathbf{\hat{r}}(t)=%
\mathbf{v}t$ is a straight line, $\varphi =0$. We assume that the function $%
\breve{\psi}_{1}=\psi _{0}$ from (\ref{mvch1}) satisfies (\ref{Yz}), (\ref%
{Ythe}) (ground states from Examples \ref{Ex1}-\ref{Ex3} satisfy this
assumption). We take sequences $a_{n}\rightarrow 0,$ $\chi _{n}\rightarrow
0, $ $R_{n}\rightarrow 0,$ $\theta _{n}\rightarrow \infty $\ such that
conditions (\ref{anthen}) are satisfied. Since $\theta _{n}\rightarrow
\infty ,$ the restricted energy and charge defined as integrals over $\Omega
\left( \mathbf{\hat{r}}(t),R_{n}\right) $\ converge to\ the integrals over
the entire space and $\mathcal{\bar{E}}_{\infty },\bar{\rho}_{\infty }\ $%
are\ given by (\ref{rhofree}), namely 
\begin{equation*}
\mathcal{\bar{E}}_{n}\left( t\right) \rightarrow \gamma m\mathrm{c}%
^{2}\left( 1+\Theta \left( \omega \right) \right) ,\quad \bar{\rho}%
_{n}\left( t\right) \rightarrow q.
\end{equation*}%
Therefore (\ref{Egrc}) holds for large $n$.\ The solutions concentrate at $%
\mathbf{\hat{r}}\left( t\right) =\mathbf{v}t$\ and have the following
additional properties:(i) the energy density $\mathcal{E}$\ is
center-symmetric with respect to $\mathbf{\hat{r}}\left( t\right) =\mathbf{v}%
t$, hence the ergocenter $\mathbf{r}\left( t\right) $ defined by (\ref{rnom}%
) coincides with the center $\mathbf{\hat{r}}\left( t\right) $; (ii) the
charge density $\rho $ given by (\ref{rhofree}) and according to (\ref{psia}%
) it converges to delta-function $q\delta \left( \mathbf{x}-\mathbf{r}%
\right) \ $as $a\rightarrow 0$; (iii) the current $\mathbf{J}$ is given by (%
\ref{Pnr}), its components are center-symmetric and converge to the
corresponding components of $q\mathbf{v}\delta \left( \mathbf{x}-\mathbf{r}%
\right) $.
\end{example}

The following statement is the main result of this paper. It describes
trajectories relevant to concentrating solutions to the KG\ equations.

\begin{theorem}[relevant trajectories]
\label{TconcEinstein}Let solutions $\psi $ of\ the KG\ equation (\ref{KGex})
concentrate at $\mathbf{\hat{r}}(t)$. Then the restricted energies $\mathcal{%
\bar{E}}_{n}\left( t\right) $ converge to the limit restricted energy $%
\mathcal{\bar{E}}_{\infty }\left( t\right) $ which satisfies equation (\ref%
{Einf}).\ The limit energy $\mathcal{\bar{E}}_{\infty }\left( t\right) $ and
the trajectory $\mathbf{\hat{r}}\left( t\right) $\ satisfy equation 
\begin{equation}
\partial _{t}\left( \frac{1}{\mathrm{c}^{2}}\mathcal{\bar{E}}_{\infty
}\left( t\right) \partial _{t}\mathbf{\hat{r}}\right) =\mathbf{f}_{\infty }
\label{New2}
\end{equation}%
with the electric force $\mathbf{f}_{\infty }\left( t\right) $ given by the
formula\ 
\begin{equation}
\mathbf{f}_{\infty }\left( t\right) =-\bar{\rho}_{\infty }\nabla \varphi
_{\infty }\left( t\right)  \label{FLorinf}
\end{equation}%
where\ the charge $\bar{\rho}_{\infty }$ does not depend on $t$. If we
identify the coefficient at $\partial _{t}\mathbf{\hat{r}}$ in (\ref{New2})
with\ the mass $M$ by Einstein's formula%
\begin{equation}
M=\frac{1}{\mathrm{c}^{2}}\mathcal{\bar{E}}_{\infty }\left( t\right) ,
\label{EinM}
\end{equation}%
the following formula holds: 
\begin{equation}
M=\gamma M_{0},\quad \gamma =\left( 1-\left( \partial _{t}\mathbf{\hat{r}}%
\right) ^{2}/\mathrm{c}^{2}\right) ^{-1/2}  \label{Massfor}
\end{equation}%
where $M_{0}$ is a constant.
\end{theorem}

The proof of Theorem \ref{TconcEinstein} is given in Section \ref{Sproof}.

Note that the formula (\ref{EinM}) defines the mass\ $M$\ based on the
requirement that (\ref{New2}) takes the form of the relativistic version of
Newton's equation (\ref{fre1}). Therefore formula (\ref{New2})\ provides an
alternative Newtonian method of deriving Einstein's formula.

\begin{remark}
\label{Rfreemass}The constant $M_{0}$ in (\ref{Massfor}) equals the mass $M$%
\ if $\partial _{t}\mathbf{\hat{r}}=\mathbf{0}$. This allows to interpret $%
M_{0}$ as the rest mass as in (\ref{fre1}). We would like to stress that the
rest mass $M_{0}$\ in our treatment is not a prescribed quantity, but it is
derived in (\ref{Mconst}) as an integral of motion of an equation obtained
in asymptotic limit from the KG\ equation. As any integral of motion it can
take different values for different "trajectories" of the field. Notice that
the rest mass $M_{0}$ can take different values for different rest states
described in Section \ref{Snonlinrest}. The integral of motion $M_{0}$ can
be related to the mass $m_{0}$ of one of resting charges considered in
Section \ref{Suniform} by the identity%
\begin{equation}
M_{0}=m_{0}  \label{Meqm0}
\end{equation}%
if the velocity vanishes on a time interval or asymptotically as $%
t\rightarrow -\infty $ or $t\rightarrow \infty $. If the velocity $\partial
_{t}\mathbf{\hat{r}}$ vanishes just at a time instant $t_{0}$ then it is, of
course, possible to express the value of $M_{0}$ in terms of $\mathcal{\bar{E%
}}_{\infty }=\mathcal{\bar{E}}_{\infty }\left( \psi \right) $ by formulas (%
\ref{EinM}), (\ref{Mconst}), but the corresponding $\psi =\psi \left(
t_{0}\right) $ may have no relation to the rest solutions of the field
equation with a time independent profile $\left\vert \psi \right\vert ^{2}$.
It is also possible that $\partial _{t}\mathbf{\hat{r}}$\ never equals zero,
and in fact this is a general case since all three components of velocity
may vanish simultaneously only in very special situations. Hence, there is a
possibility of localized regimes where the value of the "rest mass" $M_{0}$
may differ from the rest mass of a free charge. In such regimes the value of
the rest mass cannot be derived based on the analysis of the uniform motion
as in Section \ref{Suniform}. This wide variety of possibilities makes even
more remarkable the fact that the inertial mass is well-defined and that
Einstein's formula (\ref{fre3a}) holds even in such general regimes where
the standard analysis based on the Lorentz invariance of the uniform motion
as in Section \ref{Suniform} does not apply.
\end{remark}

\begin{remark}
If the external electric field $\varphi $ is not zero, the system described
by (\ref{paf1}) is not a closed system, and there are principal differences
between closed and non-closed systems. In particular, the total momentum and
the energy of a closed system are preserved and form a 4-vector. For
non-closed systems the center of energy (also known as center of mass or
centroid) and the total energy-momentum are frame dependent and hence are
not 4-vectors, \cite[Sec. 7.1, 7.2]{Moller}, \cite[Sec. 24]{Lanczos VPM}.
Hence, one cannot identify the velocity parameter $\mathbf{v}$ of the
Lorentz group with the velocity of the system.
\end{remark}

\begin{remark}
\label{Runiq}The sequences $a_{n},R_{n},\theta _{n},\varphi _{n},\chi
_{n},\psi _{n}$\ enter the definition of a concentrating solution. But if we
take two different sequences which fit the definition \ for the same
trajectory, we obtain the same equation (\ref{New2})-(\ref{Massfor}) and the
limit energy $E_{\infty }\left( t\right) $. More than that, as long as the
gradient $\nabla \varphi _{\infty }\left( t\right) $ of the limit potential
is given and at a moment of time $t_{0}$ the position and velocity\ $\ 
\mathbf{\hat{r}}\left( t_{0}\right) \mathbf{,}\partial _{t}\mathbf{\hat{r}}%
\left( t_{0}\right) $ and\ the limit restricted energy and charge $E_{\infty
}\left( t_{0}\right) ,\bar{\rho}_{\infty }$ are fixed, then the\ trajectory $%
\mathbf{\hat{r}}\left( t\right) $ and the energy $E_{\infty }\left( t\right) 
$\ are uniquely defined as solutions of equations (\ref{New2})-(\ref{Massfor}%
), and consequently they\ do not depend on the particular sequences. Note
also that if in Definition \ref{Dlocconverge} we would eliminate point (iv),
we could pick a subsequence for which (\ref{Econv}) holds. For any such a
subsequence we obtain (\ref{New2}),\ but $\mathcal{\bar{E}}_{\infty }\left(
t_{0}\right) $ could be different. Taking into account (\ref{Massfor}) and
Corollary \ref{Cuniq0} we see that the choice of different subsequences
leads to multiplication of equations (\ref{New2}), (\ref{Massfor}) by a
constant and corresponds to simultaneous multiplication of the rest mass\ $%
M_{0}$ and the charge $\bar{\rho}_{\infty }$ by the same constant.
\end{remark}

In Section \ref{secAR} we construct a non-trivial example of sequences of
the KG\ equations and their solutions\ which concentrate at rectilinear
accelerating trajectories.

\subsection{Properties of concentrating solutions}

Everywhere in this section we assume that we have a trajectory $\mathbf{\hat{%
r}}(t)$ with its concentrating neighborhoods $\hat{\Omega}_{n}\left( \mathbf{%
\hat{r}},R_{n}\right) $.\ In this section we often use the following
elementary inequalities for the densities\ defined in Section \ref{sreldis}:%
\begin{equation}
\left\vert \mathbf{P}\right\vert \leq \frac{\chi ^{2}}{m\mathrm{c}^{2}}%
\left\vert \partial _{t}\psi +\frac{\mathrm{i}q}{\chi }\varphi \psi
\right\vert \left\vert \nabla \psi \right\vert \leq C\chi ^{2}\left\vert
\partial _{t}\psi \right\vert ^{2}+C\chi ^{2}\left\vert \nabla \psi
\right\vert ^{2}+C\left\vert \psi \right\vert ^{2},  \label{Pless}
\end{equation}%
\begin{equation}
\left\vert \mathbf{J}\right\vert \leq \frac{\chi q}{m}\left\vert \nabla \psi
\right\vert \left\vert \psi \right\vert \leq C\left( \chi \left\vert \nabla
\psi \right\vert +\left\vert \psi \right\vert \right) \left\vert \psi
\right\vert \leq C^{\prime }\chi ^{2}\left\vert \nabla \psi \right\vert
^{2}+C^{\prime }\left\vert \psi \right\vert ^{2},  \label{Jless}
\end{equation}%
\begin{equation}
\left\vert \rho \right\vert \leq \frac{\chi q}{m\mathrm{c}^{2}}\left\vert 
\tilde{\partial}_{t}\psi \right\vert \left\vert \psi \right\vert \leq C\chi
^{2}\left\vert \partial _{t}\psi \right\vert ^{2}+C\left\vert \psi
\right\vert ^{2}.  \label{rholess}
\end{equation}%
Similarly,%
\begin{equation}
\left\vert \mathcal{E}\right\vert +\left\vert \mathcal{L}\right\vert \leq
C\chi ^{2}\left\vert \partial _{t}\psi \right\vert ^{2}+C\chi ^{2}\left\vert
\nabla \psi \right\vert ^{2}+C\chi ^{2}\left\vert G\left( \psi ^{\ast }\psi
\right) \right\vert +C\left\vert \psi \right\vert ^{2}.  \label{Eless}
\end{equation}%
We define the restricted to $\Omega \left( \mathbf{\hat{r}}(t),R_{n}\right) $%
\ charge $\bar{\rho}_{n}$\ by the formula 
\begin{equation}
\bar{\rho}_{n}\left( t\right) =\int_{\Omega \left( \mathbf{\hat{r}}%
(t),R_{n}\right) }\rho _{n}\,\mathrm{d}^{3}x.  \label{locch}
\end{equation}

\begin{lemma}
\label{Lchargeconst}Let solutions $\psi $ of the KG\ equation (\ref{KGex})\
concentrate at\ $\mathbf{\hat{r}}(t)$. Then there exists a number\ $\bar{\rho%
}_{\infty }$ and a subsequence of the solutions $\psi _{n}$ such that 
\begin{equation}
\bar{\rho}_{n}\left( t\right) \rightarrow \bar{\rho}_{\infty }\ \ \text{%
uniformly\ for}\ \ T_{-}\leq t\leq T_{+}.  \label{rhoinf}
\end{equation}%
\ 
\end{lemma}

\begin{proof}
\ According to (\ref{rholess}) and (\ref{locpsibound}), for any $t_{0}\in %
\left[ T_{-},T_{+}\right] $\ $\ \bar{\rho}_{n}\left( t_{0}\right) $ is a
bounded sequence, and hence it always contains a converging subsequence. We
pick such a subsequence, denote its limit by $\bar{\rho}_{\infty }$ and
integrate the continuity equation (\ref{contex}): 
\begin{equation}
\partial _{t}\int_{\Omega _{n}}\rho _{n}\,\mathrm{d}^{3}x-\int_{\partial
\Omega _{n}}\mathbf{\hat{v}}\cdot \mathbf{\bar{n}}\rho _{n}\,\mathrm{d}%
^{2}\sigma +\int_{\partial \Omega _{n}}\mathbf{\bar{n}}\cdot \mathbf{J}\,%
\mathrm{d}^{2}\sigma =0.  \label{intcont}
\end{equation}%
Integrating with respect to time from $t_{0}$ to $t$ and using (\ref{rholess}%
), (\ref{Jless})\ and (\ref{psiouta}), we obtain 
\begin{gather*}
\left\vert \bar{\rho}_{n}\left( t\right) -\bar{\rho}_{\infty }\right\vert
\leq \left\vert \bar{\rho}_{n}\left( t\right) -\bar{\rho}_{n}\left(
t_{0}\right) \right\vert +\left\vert \bar{\rho}_{n}\left( t_{0}\right) -\bar{%
\rho}_{\infty }\right\vert \leq \left( T_{+}-T_{-}\right) \left\vert \mathbf{%
\hat{v}}\right\vert \int_{\partial \Omega _{n}}\left\vert \rho
_{n}\right\vert \,\mathrm{d}^{2}\sigma \\
+\left( T_{+}-T_{-}\right) \int_{\partial \Omega _{n}}\left\vert \mathbf{J}%
\right\vert \,\mathrm{d}^{3}x\leq C\int_{\partial \Omega _{n}}\left( \chi
^{2}\left\vert \nabla _{0,x}\psi \right\vert ^{2}+C_{1}^{\prime }\left\vert
\psi \right\vert ^{2}\right) \,\mathrm{d}^{2}\sigma +\left\vert \bar{\rho}%
_{n}\left( t_{0}\right) -\bar{\rho}_{\infty }\right\vert \rightarrow 0.
\end{gather*}%
Therefore (\ref{rhoinf}) holds.
\end{proof}

\begin{theorem}[restricted energy convergence]
\label{Ldiscrepancy1}Let solutions $\psi $ of\ the KG\ equation (\ref{KGex})
concentrate at $\mathbf{\hat{r}}(t)$. Let $\bar{\rho}_{n}\left( t_{0}\right)
\rightarrow \bar{\rho}_{\infty }$. Then the restricted energy $\mathcal{\bar{%
E}}_{n}\left( t\right) $ converges uniformly to the limit restricted energy $%
\mathcal{\bar{E}}_{\infty }$%
\begin{equation}
\mathcal{\bar{E}}_{n}\left( t\right) \rightarrow \mathcal{\bar{E}}_{\infty
}\left( t\right) \text{\ uniformly on\ }\left[ T_{-},T_{+}\right] ,
\label{Endef}
\end{equation}%
\ where 
\begin{equation}
\mathcal{\bar{E}}_{\infty }\left( t\right) =\mathcal{\bar{E}}_{\infty
}\left( t_{0}\right) -\bar{\rho}_{\infty }\int_{t_{0}}^{t}\partial _{t}%
\mathbf{\hat{r}}\cdot \nabla \varphi _{\infty }\,\mathrm{d}t^{\prime }.
\label{Einf}
\end{equation}
\end{theorem}

\begin{proof}
Integrating (\ref{equen2}) with respect to $x$ and\ $t$ we obtain 
\begin{equation}
\mathcal{\bar{E}}_{n}\left( t\right) -\mathcal{\bar{E}}_{\infty }\left(
t_{0}\right) =-\int_{t_{0}}^{t}\nabla \varphi _{\infty }\cdot \int_{\Omega
_{n}}\mathbf{J}\,\mathrm{d}^{3}x\,\mathrm{d}t^{\prime }+Q_{1}
\label{dtenhat}
\end{equation}%
where $\varphi _{\infty }$ is defined by (\ref{fiinf}),\ and 
\begin{equation*}
Q_{1}=\int_{t_{0}}^{t}\int_{\partial \Omega _{n}}\left( \mathbf{\hat{v}\cdot 
\bar{n}}\mathcal{E}_{n}-\mathrm{c}^{2}\mathbf{\bar{n}}\cdot \mathbf{P}%
\right) \,\mathrm{d}^{2}\sigma \mathrm{d}t^{\prime
}-\int_{t_{0}}^{t}\int_{\Omega _{n}}\nabla \left( \varphi -\varphi _{\infty
}\right) \cdot \mathbf{J}\,\mathrm{d}^{3}x\mathrm{d}t^{\prime }+\mathcal{%
\bar{E}}_{n}\left( t_{0}\right) -\mathcal{\bar{E}}_{\infty }\left(
t_{0}\right) .
\end{equation*}%
We denote $\varphi -\varphi _{\infty }=\varphi _{\mathrm{b}}$,\ and, using
the continuity equation (\ref{contex}),\ we obtain%
\begin{equation*}
\int_{\Omega _{n}}\nabla \varphi _{\mathrm{b}}\cdot \mathbf{J}\,\mathrm{d}%
^{3}x=\partial _{t}\int_{\Omega _{n}}\varphi _{\mathrm{b}}\rho \,\mathrm{d}%
^{3}x+\int_{\partial \Omega _{n}}\left( \mathbf{\bar{n}}\cdot \mathbf{J}-%
\mathbf{\hat{v}\cdot \bar{n}}\rho \right) \varphi _{\mathrm{b}}\,\mathrm{d}%
^{2}\sigma -\int_{\Omega _{n}}\partial _{t}\varphi _{\mathrm{b}}\rho \,%
\mathrm{d}^{3}x.
\end{equation*}%
Therefore, according to (\ref{Alocr}), (\ref{rholess}), (\ref{Jless}), (\ref%
{locpsibound}) and (\ref{psiouta}) 
\begin{gather}
\left\vert \int_{t_{0}}^{t}\int_{\Omega _{n}}\nabla \varphi _{\mathrm{b}%
}\cdot \mathbf{J}\,\mathrm{d}^{3}x\right\vert \leq 2\max_{T_{-}\leq t\leq
T_{+},x\in \Omega _{n}}\left\vert \varphi _{\mathrm{b}}\right\vert
\max_{T_{-}\leq t\leq T_{+}}\int_{\Omega _{n}}\left\vert \rho \right\vert \,%
\mathrm{d}^{3}x  \label{grfij} \\
+C\max_{T_{-}\leq t\leq T_{+}}\int_{\partial \Omega _{n}}\left( \left\vert 
\mathbf{J}\right\vert +\left\vert \rho \right\vert \right) \varphi _{\mathrm{%
b}}\,\mathrm{d}^{2}\sigma +\max_{T_{-}\leq t\leq T_{+},x\in \Omega
_{n}}\left\vert \partial _{t}\varphi _{\mathrm{b}}\right\vert
\max_{T_{-}\leq t\leq T_{+}}\int_{\Omega _{n}}\left\vert \rho \right\vert \,%
\mathrm{d}^{3}x\rightarrow 0.  \notag
\end{gather}%
Similarly, we estimate the first term in $Q_{1}$ 
\begin{equation*}
\left\vert \int_{t_{0}}^{t}\int_{\partial \Omega _{n}}\mathbf{\hat{v}\cdot 
\bar{n}}\mathcal{E}_{n}-\mathrm{c}^{2}\mathbf{\bar{n}}\cdot \mathbf{P}\,%
\mathrm{d}^{2}\sigma \mathrm{d}t^{\prime }\right\vert \leq C\max_{T_{-}\leq
t\leq T_{+}}\int_{\partial \Omega _{n}}\left( \left\vert \mathcal{E}%
_{n}\right\vert +\left\vert \mathbf{P}\right\vert \right) \,\mathrm{d}%
^{2}\sigma \mathrm{d}t^{\prime }
\end{equation*}%
and conclude that $\left\vert Q_{1}\right\vert \rightarrow 0.$

Multiplication of the continuity equation (\ref{contex})\ by the vector $%
\mathbf{x}-\mathbf{\hat{r}}$\ produces the following expression for $\mathbf{%
J}$: 
\begin{equation}
\partial _{t}\left( \rho \left( \mathbf{x}-\mathbf{\hat{r}}\right) \right)
+\partial _{t}\mathbf{\hat{r}}\rho +\tsum\nolimits_{j}\nabla _{j}\left(
\left( \mathbf{x}-\mathbf{\hat{r}}\right) \mathbf{J}_{j}\right) =\mathbf{J.}
\label{qPr}
\end{equation}%
Integrating we see that 
\begin{equation*}
\int_{\Omega _{n}}\mathbf{J}\,\mathrm{d}^{3}x=\mathbf{\hat{v}}\bar{\rho}%
+\partial _{t}\int_{\Omega _{n}}\left( \rho \left( \mathbf{x}-\mathbf{\hat{r}%
}\right) \right) \,\mathrm{d}^{3}x+\int_{\partial \Omega _{n}}\left( \mathbf{%
x}-\mathbf{\hat{r}}\right) \left( \mathbf{\bar{n}}\cdot \mathbf{J-\hat{v}%
\cdot \bar{n}}\rho \right) \,\mathrm{d}^{2}\sigma
\end{equation*}%
and\ 
\begin{equation}
\int_{t_{0}}^{t}\nabla \varphi _{\infty }\cdot \int_{\Omega _{n}}\mathbf{J}\,%
\mathrm{d}^{3}x\mathrm{d}t^{\prime }=\int_{t_{0}}^{t}\nabla \varphi _{\infty
}\cdot \partial _{t}\mathbf{\hat{r}}\bar{\rho}_{\infty }\mathrm{d}t^{\prime
}+Q_{10}  \label{grfiinf}
\end{equation}%
where%
\begin{equation*}
Q_{10}=\left. \int_{\Omega _{n}}\left( \rho \left( \mathbf{x}-\mathbf{\hat{r}%
}\right) \right) \,\mathrm{d}^{3}x\right\vert
_{t_{0}}^{t}+\int_{t_{0}}^{t}\int_{\partial \Omega _{n}}\left( \mathbf{x}-%
\mathbf{\hat{r}}\right) \left( \mathbf{\bar{n}}\cdot \mathbf{J-\hat{v}\cdot 
\bar{n}}\rho \right) \,\mathrm{d}^{2}\sigma dt^{\prime
}+\int_{t_{0}}^{t}\nabla \varphi _{\infty }\cdot \partial _{t}\mathbf{\hat{r}%
}\left( \bar{\rho}-\bar{\rho}_{\infty }\right) \mathrm{d}t^{\prime }
\end{equation*}%
Combining with (\ref{dtenhat}), we obtain 
\begin{equation}
\mathcal{\bar{E}}_{n}\left( t\right) -\mathcal{\bar{E}}_{\infty }\left(
t_{0}\right) =-\int_{t_{0}}^{t}\nabla \varphi _{\infty }\cdot \partial _{t}%
\mathbf{\hat{r}}\left( t\right) \bar{\rho}_{\infty }dt^{\prime }+Q_{1}-Q_{10}
\label{dtEQ1}
\end{equation}%
Using (\ref{fihatb}), (\ref{rholess}), (\ref{Jless}), (\ref{trbound}), (\ref%
{locpsibound}), we conclude that 
\begin{gather}
\left\vert Q_{10}\right\vert \leq 2R_{n}\sup_{T_{-}\leq t\leq
T_{+}}\int_{\Omega _{n}}\left\vert \rho \right\vert \,\mathrm{d}^{3}x+\left(
T_{+}-T_{-}\right) R_{n}\sup_{T_{-}\leq t\leq T_{+}}\left\vert \mathbf{\hat{v%
}}\right\vert \int_{\partial \Omega _{n}}\left\vert \rho \right\vert \,%
\mathrm{d}^{2}\sigma  \label{Q10to0} \\
+\left( T_{+}-T_{-}\right) R_{n}\int_{\partial \Omega _{n}}\left\vert 
\mathbf{J}\right\vert \,\mathrm{d}^{2}\sigma +\left( T_{+}-T_{-}\right)
\sup_{T_{-}\leq t\leq T_{+}}\left\vert \nabla \varphi _{\infty }\right\vert
\cdot \mathbf{\hat{v}}\left\vert \bar{\rho}-\bar{\rho}_{\infty }\right\vert
\rightarrow 0.  \notag
\end{gather}%
We obtain (\ref{Endef}) and (\ref{Einf}) from (\ref{dtEQ1}).
\end{proof}

\begin{remark}
Equations (\ref{fre1}) for the space components of the relativistic 4-vector
are usually complemented (see, for instance, \cite{Barut}, \cite{Pauli RT}, 
\cite{Moller}) with the time component 
\begin{equation}
\frac{\mathrm{d}}{\mathrm{d}t}\left( M\mathrm{c}^{2}\right) =\mathbf{f}\cdot 
\mathbf{v.}  \label{ftime}
\end{equation}%
Formula (\ref{Einf}), obviously, has the form of integrated equation (\ref%
{ftime}).
\end{remark}

We define the adjacent ergocenter $\mathbf{r}_{n}$ by the formula 
\begin{equation}
\mathbf{r}_{n}\left( t\right) =\frac{1}{\mathcal{\bar{E}}_{n}}\int_{\Omega
\left( \mathbf{\hat{r}}(t),R_{n}\right) }\mathbf{x}\mathcal{E}_{n}\,\mathrm{d%
}^{3}x.  \label{rnom}
\end{equation}

\begin{lemma}
\label{Lcentersconverge}Let solutions of\ the KG\ equation (\ref{KGex})
concentrate at\ $\mathbf{\hat{r}}(t)$. Then the adjacent ergocenters $%
r_{n}\left( t\right) $ of the solutions converge to\ $\mathbf{\hat{r}}(t)$\
uniformly on the time interval $\left[ T_{-},T_{+}\right] $.
\end{lemma}

\begin{proof}
\ We infer from (\ref{Endef0}) and (\ref{rnom}) that 
\begin{equation}
\frac{1}{\mathcal{\bar{E}}_{n}}\int_{\Omega _{n}}\left( \mathbf{x}-\mathbf{r}%
_{n}\right) \mathcal{E}_{n}\,\mathrm{d}^{3}x=\mathbf{r}_{n}-\mathbf{r}_{n}%
\frac{1}{\mathcal{\bar{E}}_{n}}\int_{\Omega _{n}}\mathcal{E}_{n}\,\mathrm{d}%
^{3}x=0.  \label{xrn1}
\end{equation}%
From\ (\ref{locpsibound}) we obtain 
\begin{equation}
|\int_{\Omega _{n}}\left( \mathbf{x}-\mathbf{\hat{r}}\right) \mathcal{E}%
_{n}\,\mathrm{d}^{3}x|\leq R_{n}\int_{\Omega _{n}}|\mathcal{E}_{n}|\,\mathrm{%
d}^{3}x\leq CR_{n}\rightarrow 0.
\end{equation}%
Using (\ref{Egrc}) we conclude that 
\begin{equation}
\left\vert \mathbf{\hat{r}}-\mathbf{r}_{n}\right\vert =\frac{1}{\mathcal{%
\bar{E}}_{n}}\left\vert \int_{\Omega _{n}}\left( \mathbf{x}-\mathbf{r}%
_{n}\right) -\left( \mathbf{x}-\mathbf{\hat{r}}\right) \mathcal{E}_{n}\,%
\mathrm{d}^{3}x\right\vert \leq C_{1}R_{n}\rightarrow 0.  \label{rminrn}
\end{equation}
\end{proof}

\begin{lemma}
\label{Lestdtr}Let solutions $\psi =\psi _{n}$ of\ the KG\ equation (\ref%
{KGex}) concentrate at $\mathbf{\hat{r}}(t)$. Then%
\begin{equation}
\left\vert \mathbf{v}_{n}\right\vert =\left\vert \partial _{t}\mathbf{r}%
_{n}\right\vert \leq C_{4}  \label{vleC}
\end{equation}%
and for any $t_{0}$ there exists a subsequence such that $\mathbf{v}%
_{n}\left( t_{0}\right) \ $converges.
\end{lemma}

\begin{proof}
Multiplying (\ref{equen2}) by $\left( \mathbf{x}-\mathbf{r}\right) $, we
obtain%
\begin{equation}
\partial _{t}\left( \left( \mathbf{x}-\mathbf{r}\right) \mathcal{E}\right) +%
\mathcal{E}\partial _{t}\mathbf{r}=-\left( \mathbf{x}-\mathbf{r}\right) 
\mathrm{c}^{2}\nabla \cdot \mathbf{P-}\left( \mathbf{x}-\mathbf{r}\right)
\nabla \varphi \cdot \mathbf{J.}  \label{dtxmre0}
\end{equation}%
Integration over $\Omega _{n}$ yields 
\begin{equation*}
\int_{\Omega _{n}}\partial _{t}\left( \left( \mathbf{x}-\mathbf{r}\right) 
\mathcal{E}\right) \,\mathrm{d}^{3}x+\partial _{t}\mathbf{r}\int_{\Omega
_{n}}\mathcal{E}\,\mathrm{d}^{3}x=-\int_{\Omega _{n}}\left( \mathbf{x}-%
\mathbf{r}\right) \mathrm{c}^{2}\nabla \cdot \mathbf{P}\,\mathrm{d}^{3}x%
\mathbf{-}\int_{\Omega _{n}}\left( \mathbf{x}-\mathbf{r}\right) \nabla
\varphi \cdot \mathbf{J}\,\mathrm{d}^{3}x.
\end{equation*}%
From the definition of the ergocenter $\mathbf{r}$\ we infer that%
\begin{equation}
\int_{\Omega _{n}}\partial _{t}\left( \left( \mathbf{x}-\mathbf{r}\right) 
\mathcal{E}\right) \,\mathrm{d}^{3}x+\int_{\partial \Omega _{n}}\left( 
\mathbf{x}-\mathbf{r}\right) \mathbf{\hat{v}}\cdot \mathbf{\bar{n}}\mathcal{E%
}\,\mathrm{d}^{2}\sigma =0.  \label{dtxmre}
\end{equation}%
Therefore 
\begin{equation*}
\partial _{t}\mathbf{r}\mathcal{E}_{n}=\int_{\partial \Omega _{n}}\left( 
\mathbf{x}-\mathbf{r}\right) \mathbf{\hat{v}}\cdot \mathbf{\bar{n}}\mathcal{E%
}\,\mathrm{d}^{2}\sigma -\mathrm{c}^{2}\int_{\Omega _{n}}\left( \mathbf{x}-%
\mathbf{r}\right) \nabla \cdot \mathbf{P}\,\mathrm{d}^{3}x-\int_{\Omega
_{n}}\left( \mathbf{x}-\mathbf{r}\right) \nabla \varphi \cdot \mathbf{J}\,%
\mathrm{d}^{3}x.
\end{equation*}%
Using (\ref{Abound}), (\ref{Egrc}) and (\ref{rminrn})\ we obtain the
inequality 
\begin{equation*}
\left\vert \partial _{t}\mathbf{r}\right\vert \leq C\int_{\partial \Omega
_{n}}\left( \left\vert \mathcal{E}\right\vert +\left\vert \mathbf{P}%
\right\vert \right) \,\mathrm{d}^{2}\sigma +\int_{\Omega _{n}}\left(
\left\vert \mathbf{J}\right\vert +\left\vert \mathbf{P}\right\vert \right) \,%
\mathrm{d}^{3}x
\end{equation*}%
yielding (\ref{vleC}).
\end{proof}

\subsection{ Proof of Theorem \protect\ref{TconcEinstein}\label{Sproof}}

We prove in this section that if solutions to the KG equation concentrate at
a trajectory $\mathbf{\hat{r}}\left( t\right) $ then the limit restricted
energy and the trajectory must satisfy Einstein's formula and the
relativistic version of Newton's law. The proof is based on two facts.
First, according to Lemma \ref{Lcentersconverge} the trajectory $\mathbf{%
\hat{r}}\left( t\right) $ is the limit of adjacent ergocenters $\mathbf{r}%
_{n}\left( t\right) $. Second, the adjacent ergocenters satisfy equations
(see Theorem \ref{Ldiscrepancy2}) which yield in the limit the relativistic
Newton's law and Einstein's formula.

\begin{theorem}
\label{Ldiscrepancy2}Let solutions $\psi _{n}$ of the KG\ equation (\ref%
{KGex}) concentrate at $\mathbf{\hat{r}}(t)$, let $\mathbf{r}\left( t\right)
=\mathbf{r}_{n}\left( t\right) $ be adjacent ergocenters defined by (\ref%
{rnom}). Then for any $t_{0},t\in \left[ T_{-},T_{+}\right] $ 
\begin{equation}
\frac{1}{\mathrm{c}^{2}}\mathcal{\bar{E}}_{n}\left( t\right) \partial _{t}%
\mathbf{r}\left( t\right) -\frac{1}{\mathrm{c}^{2}}\mathcal{\bar{E}}%
_{n}\left( t_{0}\right) \partial _{t}\mathbf{r}\left( t_{0}\right)
=-\int_{t_{0}}^{t}\bar{\rho}_{\infty }\nabla \varphi _{\infty }dt^{\prime }+%
\mathbf{\delta }_{f}\   \label{dtern1a}
\end{equation}%
where$\ \mathbf{\delta }_{f}\rightarrow 0$ uniformly on $\left[ T_{-},T_{+}%
\right] $.
\end{theorem}

\begin{proof}
Integrating the momentum conservation equation (\ref{momeq1}) with respect
to time and over $\Omega \left( \mathbf{\hat{r}}(t),R_{n}\right) =\Omega
_{n} $ we obtain 
\begin{gather*}
\int_{\Omega _{n}}\mathbf{P}\left( t\right) \,\mathrm{d}^{3}x-\int_{\Omega
_{n}}\mathbf{P}\left( t_{0}\right) \,\mathrm{d}^{3}x-\int_{\Omega
_{n}}\int_{t_{0}}^{t}\mathbf{F}dt^{\prime }\,\mathrm{d}^{3}x \\
+\int_{\Omega _{n}}\int_{t_{0}}^{t}\nabla \mathcal{L}_{1}\left( \psi \right)
dt^{\prime }\,\mathrm{d}^{3}x+\frac{\chi ^{2}}{2m}\int_{\Omega
_{n}}\int_{t_{0}}^{t}\tsum\nolimits_{j}\partial _{j}\left( \partial _{j}\psi
\nabla \psi ^{\ast }+\partial _{j}\psi ^{\ast }\nabla \psi \right)
dt^{\prime }\,\mathrm{d}^{3}x=0.
\end{gather*}%
Since $\Omega _{n}=\Omega \left( \mathbf{\hat{r}}(t),R_{n}\right) $ does not
depend on $t^{\prime }$,\ we can change order of integration\ and get%
\begin{equation}
\int_{\Omega _{n}}\mathbf{P}\left( t\right) \,\mathrm{d}^{3}x-\int_{\Omega
_{n}}\mathbf{P}\left( t_{0}\right) \,\mathrm{d}^{3}x+Q_{0}=\int_{t_{0}}^{t}%
\int_{\Omega _{n}}\mathbf{F}\,\mathrm{d}t^{\prime },  \label{PQF}
\end{equation}%
\begin{equation}
Q_{0}=\int_{t_{0}}^{t}\int_{\partial \Omega _{n}}\mathbf{\bar{n}}\mathcal{L}%
_{1}\left( \psi \right) \,\mathrm{d}^{2}\sigma \mathrm{d}t^{\prime }+\frac{%
\chi ^{2}}{2m}\int_{t_{0}}^{t}\int_{\partial \Omega _{n}}\left( \mathbf{\bar{%
n}\cdot }\nabla \psi \nabla \psi ^{\ast }+\mathbf{\bar{n}\cdot }\nabla \psi
^{\ast }\nabla \psi \right) \,\mathrm{d}^{2}\sigma \mathrm{d}t^{\prime }.
\label{Q0}
\end{equation}%
Using (\ref{Eless}) and (\ref{psiouta}), we infer that 
\begin{equation}
\left\vert Q_{0}\right\vert \leq \left\vert t-t_{0}\right\vert
\max_{T_{-}\leq s\leq T_{+}}\int_{\partial \Omega _{n}}\left( \left\vert 
\mathcal{L}_{1}\left( \psi \right) \right\vert +\chi ^{2}m^{-1}\left\vert
\nabla \psi \right\vert ^{2}\right) \,\mathrm{d}^{2}\sigma \rightarrow 0.
\label{Q00}
\end{equation}%
Note that 
\begin{equation}
\int_{\Omega _{n}}\mathbf{P}\left( t_{0}\right) \,\mathrm{d}%
^{3}x=\int_{\Omega \left( \mathbf{\hat{r}}(t_{0}),R_{n}\right) }\mathbf{P}%
\left( t_{0}\right) \,\mathrm{d}^{3}x+Q_{01},  \label{PQ01}
\end{equation}%
\begin{equation}
Q_{01}=\int_{\Omega \left( \mathbf{0},R_{n}\right) }\mathbf{P}\left( t_{0},%
\mathbf{y+\hat{r}}(t)\right) -\mathbf{P}\left( t_{0},\mathbf{y+\hat{r}}%
(t_{0})\right) \,\mathrm{d}^{3}y.  \label{Q01}
\end{equation}%
Obviously 
\begin{gather*}
Q_{01}=\int_{\Omega \left( \mathbf{0},R_{n}\right) }\int_{t_{0}}^{t}\partial
_{s}\mathbf{P}\left( t_{0},\mathbf{y+\mathbf{\hat{r}}}\left( s\right)
\right) \,\mathrm{d}s\mathrm{d}^{3}y=\int_{t_{0}}^{t}\int_{\Omega \left( 
\mathbf{0},R_{n}\right) }\partial _{t}\mathbf{\mathbf{\hat{r}}}\left(
s\right) \cdot \nabla \mathbf{P}\left( t_{0},\mathbf{y+\mathbf{\hat{r}}}%
\left( s\right) \right) \,\mathrm{d}^{3}y\mathrm{d}s \\
=\int_{t_{0}}^{t}\int_{\partial \Omega \left( \mathbf{0},R_{n}\right) }%
\mathbf{\bar{n}\cdot \mathbf{\hat{v}}}\left( s\right) \mathbf{P}\left( t_{0},%
\mathbf{y+\mathbf{\hat{r}}}\left( s\right) \right) \,\mathrm{d}^{2}\sigma 
\mathrm{d}s.
\end{gather*}%
Using (\ref{Pless}) and (\ref{psiouta}) we infer that 
\begin{equation*}
\left\vert Q_{01}\right\vert \leq \left\vert t-t_{0}\right\vert
\max_{T_{-}\leq s\leq T_{+}}\int_{\partial \Omega _{n}}\left\vert \mathbf{P}%
\right\vert \,\mathrm{d}^{2}\sigma \rightarrow 0.
\end{equation*}%
We obtain from the relation (\ref{dtxmre0}) the following expression for $%
\mathbf{P}$: 
\begin{equation}
\mathbf{P}=\frac{1}{\mathrm{c}^{2}}\partial _{t}\left( \left( \mathbf{x}-%
\mathbf{r}\right) \mathcal{E}\right) +\frac{1}{\mathrm{c}^{2}}\mathcal{E}%
\partial _{t}\mathbf{r}+\tsum\nolimits_{j}\partial _{j}\cdot \left( \mathbf{P%
}_{j}\left( \mathbf{x}-\mathbf{r}\right) \right) +\frac{1}{\mathrm{c}^{2}}%
\left( \nabla \varphi \cdot \mathbf{J}\right) \left( \mathbf{x}-\mathbf{r}%
\right) .  \label{dtgrex}
\end{equation}%
We use this expression to write the first term in (\ref{PQF}) as follows:\ 
\begin{equation}
\int_{\Omega _{n}}\mathbf{P}\left( t\right) \,\mathrm{d}^{3}x=\frac{1}{%
\mathrm{c}^{2}}\partial _{t}\mathbf{r}\int_{\Omega _{n}}\mathcal{E}\,\mathrm{%
d}^{3}x+Q_{02}  \label{PQ02}
\end{equation}%
where 
\begin{equation}
Q_{02}=\int_{\Omega _{n}}\left( \frac{1}{\mathrm{c}^{2}}\partial _{t}\left(
\left( \mathbf{x}-\mathbf{r}\right) \mathcal{E}\right)
+\tsum\nolimits_{j}\partial _{j}\cdot \left( \mathbf{P}_{j}\left( \mathbf{x}-%
\mathbf{r}\right) \right) +\frac{1}{\mathrm{c}^{2}}\left( \nabla \varphi
\cdot \mathbf{J}\right) \left( \mathbf{x}-\mathbf{r}\right) \right) \,%
\mathrm{d}^{3}x.  \label{Q02}
\end{equation}%
Using (\ref{dtxmre}) we obtain\ 
\begin{equation*}
Q_{02}=-\frac{1}{\mathrm{c}^{2}}\int_{\partial \Omega _{n}}\left( \mathbf{x}-%
\mathbf{r}\right) \mathbf{\mathbf{\hat{v}}}\cdot \mathbf{\bar{n}}\mathcal{E}%
\,\mathrm{d}^{3}x+\int_{\partial \Omega _{n}}\left( \mathbf{x}-\mathbf{r}%
\right) \mathbf{\bar{n}\cdot P}\,\mathrm{d}^{3}x+\int_{\Omega _{n}}\frac{1}{%
\mathrm{c}^{2}}\left( \nabla \varphi \cdot \mathbf{J}\right) \left( \mathbf{x%
}-\mathbf{r}\right) \,\mathrm{d}^{3}x.
\end{equation*}%
Using (\ref{rminrn}), (\ref{Pless}), (\ref{Eless}) and (\ref{psiouta})\ we
conclude that%
\begin{equation*}
\left\vert Q_{02}\right\vert \leq CR_{n}\int_{\partial \Omega
_{n}}\left\vert \mathcal{E}\right\vert \,\mathrm{d}^{3}x+CR_{n}\int_{%
\partial \Omega _{n}}\left\vert \mathbf{P}\right\vert \,\mathrm{d}%
^{3}x+CR_{n}\max_{\Omega _{n}}\left\vert \nabla \varphi \right\vert
\int_{\Omega _{n}}\left\vert \mathbf{J}\right\vert \,\mathrm{d}%
^{3}x\rightarrow 0.
\end{equation*}%
Quite similarly\ to (\ref{PQ02}), 
\begin{equation}
\int_{\Omega \left( \mathbf{\hat{r}}(t_{0}),R_{n}\right) }\mathbf{P}\left(
t_{0}\right) \,\mathrm{d}^{3}x=\frac{1}{\mathrm{c}^{2}}\partial _{t}\mathbf{r%
}\left( t_{0}\right) \mathcal{\bar{E}}\left( t_{0}\right) +Q_{012}
\label{PQ012}
\end{equation}%
where\ $Q_{012}\rightarrow 0.$ Now we write the last term in (\ref{PQF}) in
the form%
\begin{gather}
\int_{t_{0}}^{t}\int_{\Omega _{n}}\mathbf{F}\,\mathrm{d}^{3}xdt^{\prime
}=-\int_{t_{0}}^{t}\nabla \varphi _{\infty }\int_{\Omega _{n}}\rho _{\infty
}\,\mathrm{d}^{3}xdt^{\prime }+Q_{03},  \label{FQ03} \\
Q_{03}=-\int_{t_{0}}^{t}\int_{\Omega _{n}}\left( \rho _{n}\nabla \varphi
-\rho _{\infty }\nabla \varphi _{\infty }\right) \,\mathrm{d}^{3}x.  \notag
\end{gather}%
Obviously%
\begin{equation*}
\int_{\Omega _{n}}\left( \rho _{n}\nabla \varphi -\rho _{\infty }\nabla
\varphi _{\infty }\right) \,\mathrm{d}^{3}x=\int_{\Omega _{n}}\left( \nabla
\varphi -\nabla \varphi _{\infty }\right) \rho _{n}\,\mathrm{d}^{3}x+\nabla
\varphi _{\infty }\left( \bar{\rho}_{n}-\bar{\rho}_{\infty }\right) ,
\end{equation*}%
and\ by Lemma \ref{Lchargeconst}\ and (\ref{Alocr}) 
\begin{equation*}
\left\vert Q_{03}\right\vert \leq \left\vert T_{+}-T_{-}\right\vert \max_{%
\hat{\Omega}_{n}}\left\vert \nabla \varphi -\nabla \varphi _{\infty
}\right\vert \int_{\Omega _{n}}\left\vert \rho _{n}\right\vert \,\mathrm{d}%
^{3}x+\left\vert T_{+}-T_{-}\right\vert \max_{t}\left\vert \nabla \varphi
_{\infty }\right\vert \left\vert \rho _{n}-\rho _{\infty }\right\vert
\rightarrow 0.
\end{equation*}%
From (\ref{PQF}) using (\ref{PQ02}), (\ref{PQ01}), (\ref{PQ012}), (\ref{FQ03}%
) we obtain%
\begin{equation*}
\frac{1}{\mathrm{c}^{2}}\mathcal{\bar{E}}_{n}\left( t\right) \partial _{t}%
\mathbf{r}\left( t\right) +Q_{02}-\frac{1}{\mathrm{c}^{2}}\mathcal{\bar{E}}%
_{n}\left( t_{0}\right) \partial _{t}\mathbf{r}\left( t_{0}\right) -Q_{01}\
-Q_{012}+Q_{0}=-\int_{t_{0}}^{t}\bar{\rho}_{\infty }\nabla \varphi _{\infty
}dt^{\prime }+Q_{03}
\end{equation*}%
which implies (\ref{dtern1a}) with $\mathbf{\delta }%
_{f}=Q_{03}-Q_{0}+Q_{012}+Q_{01}-Q_{02}$.
\end{proof}

\begin{proof}[Proof of Theorem \protect\ref{TconcEinstein}]
According to Lemma \ref{Lcentersconverge},\ the adjacent ergocenters $%
\mathbf{r}_{n}\left( t\right) $ converge to $\mathbf{\hat{r}}(t)$. We take $%
t_{0}$\ from Definition \ref{Dlocconverge}\ and\ choose\ such a subsequence
that according to Lemma \ref{Lestdtr} and \ref{Lchargeconst} $\partial _{t}%
\mathbf{r}_{n}\left( t_{0}\right) \rightarrow v_{\infty }$, $\bar{\rho}%
_{n}\rightarrow \bar{\rho}_{\infty }$.

We multiply (\ref{dtern1a}) by $\mathrm{c}^{2}/\mathcal{\bar{E}}_{n}$ and
obtain 
\begin{equation}
\partial _{t}\mathbf{r}_{n}\left( t\right) =\frac{\mathrm{c}^{2}}{\mathcal{%
\bar{E}}_{\infty }}\int_{t_{0}}^{t}\mathbf{f}_{\infty }dt^{\prime }+\frac{%
\mathrm{c}^{2}}{\mathcal{\bar{E}}_{\infty }}\mathbf{\delta }_{f}+\frac{1}{%
\mathcal{\bar{E}}_{\infty }}\mathcal{\bar{E}}_{n}\left( t_{0}\right)
\partial _{t}\mathbf{r}_{n}\left( t_{0}\right) .  \label{intdtern1}
\end{equation}%
Integration of the above equation yields 
\begin{equation*}
\mathbf{r}_{n}\left( t\right) -\mathbf{r}_{n}\left( t_{0}\right)
=\int_{t_{0}}^{t}\frac{\mathrm{c}^{2}}{\mathcal{\bar{E}}_{\infty }}\mathbf{%
\delta }_{f}dt^{\prime \prime }+\int_{t_{0}}^{t}\frac{\mathrm{c}^{2}}{%
\mathcal{\bar{E}}_{\infty }}\int_{t_{0}}^{t^{\prime \prime }}\mathbf{f}%
_{\infty }dt^{\prime }dt^{\prime \prime }+\int_{t_{0}}^{t}\frac{1}{\mathcal{%
\bar{E}}_{\infty }\left( t^{\prime \prime }\right) }dt^{\prime \prime }%
\mathcal{\bar{E}}_{n}\left( t_{0}\right) \partial _{t}\mathbf{r}_{n}\left(
t_{0}\right) .
\end{equation*}%
Observe that the sequence $1/\mathcal{\bar{E}}_{n}$ converges uniformly to $%
1/\mathcal{\bar{E}}_{\infty }$\ according to Lemma \ref{Ldiscrepancy1} and (%
\ref{Egrc}), $\mathbf{r}_{n}\left( t\right) $ converges to $\mathbf{\hat{r}}%
\left( t\right) $\ and $\mathbf{\delta }_{f}$ converge to zero.\ Hence
taking the limit of the above we get 
\begin{equation}
\mathbf{\hat{r}}\left( t\right) -\mathbf{\hat{r}}\left( t_{0}\right)
=\int_{t_{0}}^{t}\frac{\mathrm{c}^{2}}{\mathcal{\bar{E}}_{\infty }}%
\int_{t_{0}}^{t^{\prime \prime }}\mathbf{f}_{\infty }dt^{\prime }dt^{\prime
\prime }+\int_{t_{0}}^{t}\frac{1}{\mathcal{\bar{E}}_{\infty }\left(
t^{\prime \prime }\right) }dt^{\prime \prime }\mathcal{\bar{E}}_{\infty
}\left( t_{0}\right) v_{\infty }.  \label{rrtte1}
\end{equation}%
Taking the time derivative of the equality (\ref{rrtte1}), multiplying the
result by $\mathcal{\bar{E}}_{\infty }\left( t\right) /\mathrm{c}^{2}$ and
taking the time derivative once more we find that\ $\mathbf{\hat{r}}\left(
t\right) $ satisfies (\ref{New2})\ and $\mathbf{\hat{r}}\left( t_{0}\right) =%
\mathbf{r}_{\infty }\left( t_{0}\right) $, $\partial _{t}\mathbf{\hat{r}}%
\left( t_{0}\right) =v_{\infty }$.

Notice that equation (\ref{Einf}) implies 
\begin{equation}
\partial _{t}\mathcal{\bar{E}}_{\infty }=\partial _{t}\mathbf{\hat{r}}\cdot 
\mathbf{f}_{\infty }.  \label{Elim}
\end{equation}%
Multiplying (\ref{New2}) by $2M\partial _{t}\mathbf{\hat{r}}$ where $M=%
\mathcal{\bar{E}}_{\infty }\left( t\right) /\mathrm{c}^{2}$\ and using (\ref%
{Elim}),\ we obtain 
\begin{equation}
\partial _{t}\left( M\partial _{t}\mathbf{\hat{r}}\right) ^{2}=2M\partial
_{t}\mathbf{\hat{r}\cdot f}_{\infty }=2\mathrm{c}^{2}M\partial _{t}M.
\label{Elim1}
\end{equation}%
This equation after time integration implies 
\begin{equation}
M^{2}-\mathrm{c}^{-2}M^{2}\left( \partial _{t}\mathbf{\hat{r}}\right)
^{2}=M_{0}^{2},  \label{Mconst}
\end{equation}%
where $M_{0}^{2}$ is a constant of integration; consequently, we obtain (\ref%
{Massfor}). Note that $M_{0}=\gamma ^{-1}\left( t_{0}\right) \mathrm{c}^{-2}%
\mathcal{\bar{E}}_{\infty }\left( t_{0}\right) $ is uniquely determined by
the value $\mathcal{\bar{E}}_{\infty }\left( t_{0}\right) $ from (\ref{Econv}%
). Hence the limit of the restricted energy $\mathcal{\bar{E}}_{\infty
}\left( t\right) $ does not depend on particular\ subsequences $\partial _{t}%
\mathbf{r}_{n}\left( t_{0}\right) $, $\bar{\rho}_{n}$ and on $\bar{\rho}%
_{\infty }$.\ Based on this we conclude that any subsequence $\mathcal{\bar{E%
}}_{n}\left( t\right) $ has\ the same limit and the convergence in (\ref%
{Endef}) holds for the given concentrating sequence.
\end{proof}

\begin{corollary}
\label{Cuniq0}Assume that $\partial _{t}^{2}\mathbf{\hat{r}}$ is not
identically zero on $\left[ T_{-},T_{+}\right] $. Then (\ref{rhoinf}) holds\
for any subsequence of the concentrating sequence and the sequence $\bar{\rho%
}_{n}\left( t\right) $ converges to $\bar{\rho}_{\infty }$.
\end{corollary}

\begin{proof}
We have to prove that every convergent subsequence $\bar{\rho}_{n}\left(
t_{0}\right) $ in Lemma \ref{Lchargeconst} converges to the same limit. If
we have a convergent subsequence $\bar{\rho}_{n}\left( t_{0}\right) $ which
converges to $\bar{\rho}_{\infty }$, we obtain (\ref{Massfor}) and (\ref%
{New2}). According to formula (\ref{Massfor}), we can rewrite (\ref{New2})
in the form 
\begin{equation}
M_{0}\gamma \partial _{t}^{2}\mathbf{\hat{r}}+M_{0}\partial _{t}\gamma
\partial _{t}\mathbf{\hat{r}}=-\bar{\rho}_{\infty }\nabla \varphi _{\infty }.
\label{New21}
\end{equation}%
This relation uniquely determines $\bar{\rho}_{\infty }$ if\ $\nabla \varphi
_{\infty }\left( t\right) $ is not\ identically zero. The left-hand side
vanishes on an interval only if\ $\partial _{t}\mathbf{\hat{v}}+\ \partial
_{t}\ln \gamma \mathbf{\hat{v}=0}$\ which is possible only for $\partial _{t}%
\mathbf{\hat{v}=0}$ on the interval. Therefore, for accelerated motion $%
\nabla \varphi _{\infty }$ is not identically zero on the interval and $\bar{%
\rho}_{\infty }$ is\ \ uniquely determined\ by $\mathbf{\hat{r}}$,$t_{0}$
and $M_{0}$ where\ $M_{0}$ by is uniquely defined by (\ref{Massfor})
according to (\ref{Econv}).
\end{proof}

In the following Section \ref{secAR} we present a class of examples where
solutions of the KG\ equation concentrate at a trajectory which describes an
accelerated motion.

\section{Rectilinear accelerated motion of the wave with a fixed shape \label%
{secAR}}

We proved in the preceding section that if solutions of the KG equation
concentrate at a trajectory, then the trajectory and the energy satisfy
Newton's equation where the mass is defined by Einstein's formula. The
results of the previous section are valid for general concentrating
solutions of the KG equation. In this section we present a specific class of
examples for which the concentration assumptions hold. The constructions are
explicit and provide for examples of relativistic accelerated motion with
the internal energy determined from the Klein-Gordon Lagrangian. To be able
to construct an explicit example we make the following simplifying
assumptions: (i) the motion is rectilinear; (ii) the nonlinearity is
logarithmic\ as in (\ref{Gpa}); (iii)\ the shape $\left\vert \psi
\right\vert $ is fixed, namely it is Gaussian; (iv) the parameter $\zeta
=\zeta _{n}$ in (\ref{acan}) satisfies condition $\zeta _{n}\rightarrow 0$.
For a given rectilinear trajectory we find a sequence of parameters,
potentials and solutions of the KG equation which concentrate at the
trajectory. Now we formulate the assumptions in detail.

Recall that the parameters $\mathrm{c}$, $m$, $q$ are fixed. We assume that
the sequence $a=a_{n},\chi =\chi _{n}$ satisfies the conditions (\ref{anto0}%
), (\ref{acan}) and an additional condition 
\begin{equation}
\zeta =\zeta _{n}=\frac{a_{\mathrm{C}}}{a}=\frac{\chi _{n}}{a_{n}m\mathrm{c}}%
\rightarrow 0.  \label{zetzet0}
\end{equation}%
Trajectories $\mathbf{\hat{r}}\left( t\right) $ which describe \emph{%
rectilinear} accelerated\ motion have the form 
\begin{equation}
\mathbf{\hat{r}}\left( t\right) =(0,0,r\left( t\right) ),\quad -\infty
<t<\infty .  \label{rtinf}
\end{equation}%
We consider a fixed trajectory $r\left( t\right) $ such that corresponding
velocity $v=\partial _{t}r$\ is two times continuously differentiable and
has uniformly bounded derivatives: 
\begin{equation}
\left\vert v\left( t\right) \right\vert +\left\vert \partial _{t}v\left(
t\right) \right\vert +\left\vert \partial _{t}^{2}v\left( t\right)
\right\vert \leq C,\quad -\infty <t<\infty .  \label{betep1}
\end{equation}%
We also impose a weaker version of the above restriction, 
\begin{equation}
\sup_{\tau }\left( \left\vert \partial _{\tau }\beta \right\vert +\left\vert
\partial _{\tau }^{2}\beta \right\vert \right) \leq \hat{\epsilon},\quad 
\text{where }\tau =\frac{\mathrm{c}}{a}t,  \label{betep0}
\end{equation}%
$\beta $ is the \emph{normalized velocity, }namely 
\begin{equation}
\beta =v/\mathrm{c},\quad v=\partial _{t}r,  \label{vv0v1}
\end{equation}%
this version is sufficient in many estimates. Since the parameter $%
a_{n}\rightarrow 0$,\ the assumption (\ref{betep0}) is less restrictive than
(\ref{betep1}).

The velocity $v=\partial _{t}r$ is\ assumed to be smaller than the speed of
light $\mathrm{c}$, namely that 
\begin{equation}
\left\vert \beta \left( t\right) \right\vert \leq \hat{\epsilon}_{1}<1,\quad
-\infty <t<\infty ;  \label{cv10}
\end{equation}%
we assume also that the normalized velocity does not vanish: 
\begin{equation}
\left\vert \beta \left( t\right) \right\vert \geq \check{\beta}>0,\quad
-\infty <t<\infty .  \label{betcheck}
\end{equation}%
Here is the main result of this section.\ 

\begin{theorem}
\label{Ttraex}For any trajectory $\mathbf{\hat{r}}\left( t\right) =\left(
0,0,r\left( t\right) \right) $ where $r$\ satisfies (\ref{betep1}), (\ref%
{betep0}), (\ref{cv10}), (\ref{betcheck}), for any $T_{-},T_{+}$ there
exists a sequence $a_{n},R_{n},\chi _{n},$ and potentials $\varphi _{n}$
such that the KG equations are localized in $\hat{\Omega}\left( \mathbf{\hat{%
r}},R_{n}\right) $. There exists a sequence of solutions of the KG equation
which concentrates at $\mathbf{\hat{r}}$.\ 
\end{theorem}

\begin{proof}
The statement follows from Theorem \ref{Ttrajcon}\ which is proven at the
end of this section. The potentials and solutions are described in the
following sections.
\end{proof}

\subsection{Reduction to one dimension}

When the external potential $\varphi $ depends only on $t$ and $x_{3}$, the
equation (\ref{KGex}) in the three dimensional space with a logarithmic
nonlinearity (\ref{paf3}) can be reduced to a problem in one dimensional
space by the following substitution 
\begin{equation}
\psi =\pi ^{-1/2}a^{-1}\exp \left( -a^{-2}\left( x_{1}^{2}+x_{2}^{2}\right)
/2\right) \psi _{1D}\left( t,x_{3}\right) ,  \label{psi1d}
\end{equation}%
with $\psi _{1D}$ being dependent only on $x_{3}$ and $t$. The corresponding
reduced $1D$ KG equation\ for $\psi =\psi _{1D}$ with \emph{one spatial
variable} has the form 
\begin{equation}
-\mathrm{c}^{-2}\tilde{\partial}_{t}^{2}\psi +\partial _{3}^{2}\psi
-G_{a}^{\prime }\left( \psi ^{\ast }\psi \right) \psi -\kappa _{0}^{2}\psi
=0,  \label{KG1D}
\end{equation}%
where 
\begin{equation*}
\kappa _{0}=m\mathrm{c}/\chi ,\quad \tilde{\partial}_{t}=\partial _{t}+%
\mathrm{i}q\varphi /\chi ,\quad \varphi =\varphi \left( t,x_{3}\right) ,
\end{equation*}%
and the $1D$ logarithmic nonlinearity takes the form 
\begin{equation}
G_{a}^{\prime }\left( \left\vert \psi \right\vert ^{2}\right)
=G_{a,1D}^{\prime }\left( \left\vert \psi \right\vert ^{2}\right) =-a^{-2} 
\left[ \ln \left( \pi ^{1/2}\left\vert \psi \right\vert ^{2}\right) +1\right]
-a^{-2}\ln a.  \label{logb}
\end{equation}

From now on we write $x$ instead of $x_{3}$ for the notational simplicity.
We look for the external potential $\varphi $\ in the form 
\begin{equation}
\varphi \left( t,x\right) =\varphi _{\mathrm{ac}}\left( t,x\right) +\varphi
_{\mathrm{b}}\left( t,x;\zeta \right) ,  \label{fiex}
\end{equation}%
where the \emph{accelerating potential }$\varphi _{\mathrm{ac}}$ \emph{is
linear in }$y$, namely\ 
\begin{equation}
\varphi _{\mathrm{ac}}=\varphi _{\mathrm{0}}\left( t\right) +\varphi _{%
\mathrm{ac}}^{\prime }y,\quad y=x-r,  \label{fiac}
\end{equation}%
and $\varphi _{\mathrm{b}}\left( t,x;\zeta \right) $ is a small \emph{%
balancing potential}.\ The coefficient $\varphi _{\mathrm{ac}}^{\prime
}=\partial _{x}\varphi _{\mathrm{ac}}$ is determined by the trajectory $%
r\left( t\right) $ according to the formula%
\begin{equation}
\partial _{t}\left( m\gamma v\right) +q\partial _{x}\varphi _{\mathrm{ac}%
}=0,\quad v=\partial _{t}r,  \label{new11}
\end{equation}%
which has the form of\emph{\ relativistic law of motion} (\ref{fre1}).\ The
potential $\varphi _{\mathrm{ac}}$ coincides in our construction with\ the
limit potential $\varphi _{\infty }$ in (\ref{fiinf}).\ The coefficient $%
\varphi _{\mathrm{0}}\left( t\right) $ can be prescribed as an arbitrary
function with bounded derivatives.

According to the above formula, $\varphi _{\mathrm{ac}}$ directly relates to
the acceleration of the charge, and we call it "accelerating" potential, $%
\varphi _{\mathrm{ac}}$ does not depend on the small parameter $\zeta $. The
remaining part of the external potential $\varphi $ in the KG equation (\ref%
{KG1D}) is a small "balancing" potential $\varphi _{\mathrm{b}}$ which
allows the charge to exactly preserve its form as it accelerates. Below we
find such a\ potential $\varphi _{\mathrm{b}}$\ that the Gaussian wave
function with the center $r\left( t\right) $ is a solution of (\ref{KG1D})
in the strip 
\begin{equation}
\Xi \left( \theta _{n}\right) =\left\{ \left( t,x\right) :\left\vert
x-r\left( t\right) \right\vert \leq \theta _{n}a_{n}\right\} ,\quad \theta
_{n}\rightarrow \infty ,\quad \theta _{n}a_{n}=R_{n}\rightarrow 0.
\label{Ksin}
\end{equation}%
The balancing potential vanishes asymptotically, that is $\varphi _{\mathrm{b%
}}\left( t,x;\ \zeta \right) \rightarrow 0$ as $\zeta \rightarrow 0$, and
the forces it produces also become vanishingly small compared with the
electric force $-q\partial _{x}\varphi _{\mathrm{ac}}\left( x\right) $\ in
the strip $\Xi \left( \theta \right) $.\ Note that\ such an accelerated
motion with a preserved shape is possible only for a properly chosen
potential $\varphi _{\mathrm{b}}$.

\subsection{Equation in a moving frame}

As the first step of the construction of the potential $\varphi _{\mathrm{b}%
} $ we rewrite the KG equation (\ref{KG1D}) in a moving frame. We take $%
\mathbf{\hat{r}}\left( t\right) $ as the new origin and make the following
change of variables: 
\begin{equation}
x_{3}=r\left( t\right) +y,\quad \mathbf{x}=\mathbf{\hat{r}}\left( t\right) +%
\mathbf{y}.  \label{xry}
\end{equation}%
The\ 1D KG equation (\ref{KG1D}) then takes the form 
\begin{gather}
-\mathrm{c}^{-2}\left( \partial _{t}+\mathrm{i}q\varphi /\chi -v\mathbf{%
\partial }_{y}\right) \left( \partial _{t}+\mathrm{i}q\varphi /\chi -v%
\mathbf{\partial }_{y}\right) \psi  \label{KGy} \\
+\partial _{y}^{2}\psi -G_{a}^{\prime }\left( \psi ^{\ast }\psi \right) \psi
-a_{C}^{-2}\psi =0,  \notag
\end{gather}%
The 1D logarithmic nonlinearity $G_{a}^{\prime }=G_{a,1D}^{\prime }$ is
defined by (\ref{logb}), and the electric potential $\varphi \ $has the form
(\ref{fiex}).\ We assume that the solution $\psi \left( t,x\right) $ has the
Gaussian form, namely 
\begin{gather}
\psi =\psi _{1D}=\hat{\Psi}\exp (\mathrm{i}\hat{S}),  \label{psipsi} \\
\hat{S}=\omega _{0}\mathrm{c}^{-2}\gamma vy-s\left( t\right) -S\left(
t,y\right) ,\quad \omega _{0}=\mathrm{c}/a_{\mathrm{C}}=m\mathrm{c}^{2}/\chi
\notag
\end{gather}%
where we explicitly define the real valued function $\hat{\Psi}$: 
\begin{equation}
\hat{\Psi}=a^{-1/2}\Psi \left( a^{-1}y\right)  \label{psigam}
\end{equation}%
where%
\begin{equation}
\Psi \left( t,z\right) =\pi ^{-1/4}\mathrm{e}^{\sigma -z^{2}/2},
\label{psi1sig}
\end{equation}%
with 
\begin{equation}
\sigma =\sigma \left( t\right) =\ln \gamma ^{-1/2},\quad \gamma =\left(
1-\beta ^{2}\right) ^{-1/2}.  \label{siggam}
\end{equation}%
We define $\sigma $ by the above formula to satisfy (\ref{imips}).
Substitution of (\ref{psipsi}) into (\ref{KGy}) yields%
\begin{gather}
-\frac{1}{\mathrm{c}^{2}}\left( \partial _{t}+\mathrm{i}\partial _{t}\left(
\gamma v\right) \frac{\omega _{0}}{\mathrm{c}^{2}}y-\mathrm{i}\partial _{t}s+%
\frac{\mathrm{i}q\varphi _{\mathrm{0}}}{\chi }-\mathrm{i}v^{2}\frac{\omega
_{0}}{\mathrm{c}^{2}}\gamma \right.  \notag \\
\left. -\mathrm{i}\partial _{t}S+\mathrm{i}v\partial _{y}S+\frac{\mathrm{i}%
q\varphi _{\mathrm{ac}}^{\prime }}{\chi }y+\frac{\mathrm{i}q\varphi _{%
\mathrm{b}}}{\chi }-v\partial _{y}\right) ^{2}\hat{\Psi}  \notag \\
+\left( \partial _{y}-\mathrm{i}\partial _{y}S+\mathrm{i}v\frac{\omega _{0}}{%
\mathrm{c}^{2}}\gamma \right) ^{2}\hat{\Psi}-G_{a}^{\prime }(|\hat{\Psi}%
|^{2})\hat{\Psi}-\kappa _{0}^{2}\hat{\Psi}=0.  \label{KGy1}
\end{gather}%
To eliminate the leading, independent of the small parameter $\zeta ,$\
terms in the above equation, we require (\ref{new11}) to hold together with
the following equation: 
\begin{equation}
-\partial _{t}s+q\varphi _{\mathrm{0}}/\chi -v^{2}\omega _{0}\mathrm{c}%
^{-2}\gamma =-\gamma \omega _{0}.  \label{sprime}
\end{equation}%
Observe that the expressions in (\ref{KGy1}) eliminated in view of equations
(\ref{sprime}) and equations (\ref{new11}) do not depend on $\zeta $.\ Note
that the constant part $\varphi _{\mathrm{0}}\left( t\right) $ of the
accelerating potential can be prescribed arbitrarily since we always can
choose the phase shift $s\left( t\right) $ so that the equation (\ref{sprime}%
) holds. Equation (\ref{KGy}) combined with equations (\ref{sprime}) and (%
\ref{new11}) can be transformed into 
\begin{gather}
-\frac{1}{\mathrm{c}^{2}}\left( \partial _{t}-\mathrm{i}\gamma \frac{\mathrm{%
c}}{a_{\mathrm{C}}}-\mathrm{i}\partial _{t}S+\mathrm{i}v\partial _{y}S+\frac{%
\mathrm{i}q\varphi _{\mathrm{b}}}{\chi }-v\partial _{y}\right) ^{2}\hat{\Psi}
\notag \\
+\left( \partial _{y}-\mathrm{i}\partial _{y}S+\mathrm{i}\frac{1}{a_{\mathrm{%
C}}}\frac{v}{\mathrm{c}}\gamma \right) ^{2}\hat{\Psi}-G_{a}^{\prime }(|\hat{%
\Psi}|^{2})\hat{\Psi}-\frac{1}{a_{\mathrm{C}}^{2}}\hat{\Psi}=0.  \label{Kpsy}
\end{gather}%
In the following sections,\ we find the quantities $\varphi _{\mathrm{b}}$
and $S$ which are of order $\zeta ^{2}$ and satisfy this equation.

\subsection{Equations for auxiliary phases}

In this subsection we introduce two auxiliary phases and reduce the problem
of determination of the potential\ $\varphi _{\mathrm{b}}$ and the phase $S$
to a first-order partial differential equation for a single unknown phase.
Solving this equation can be reduced to integration along characteristics
allowing for a rather detailed mathematical analysis.

It is convenient to introduce rescaled dimensionless variables $\mathbf{z}%
,\tau $:\ 
\begin{equation}
\tau =\frac{\mathrm{c}}{a}t,\quad \mathbf{z}=\frac{\zeta }{a_{\mathrm{C}}}%
\mathbf{y}=\frac{1}{a}\mathbf{y},  \label{zztau}
\end{equation}%
where we set $z_{3}=z$.\ We introduce now auxiliary phases $Z$ and $\Phi $%
\begin{gather}
Z=\zeta \mathbf{\partial }_{z}S,  \label{denZ} \\
\Phi =-\zeta \mathbf{\partial }_{\tau }S+\zeta \beta \mathbf{\partial }_{z}S+%
\frac{qa_{\mathrm{C}}\varphi _{\mathrm{b}}}{\mathrm{c}\chi },  \label{denZ2}
\end{gather}%
and these phases will be our new unknown variables. Obviously, if we find $Z$
and $\Phi $, we can find $S$ by the integration in $z$ and setting $S=0$ at $%
z=0$, see (\ref{Sfor}). After that $\varphi _{\mathrm{b}}$ can be found from
(\ref{denZ2}). Consequently, to find a small $\varphi _{\mathrm{b}}$ we need
to find small $Z$, $\Phi $.

Equation (\ref{Kpsy}) takes the following form:%
\begin{gather}
-\left( \zeta \partial _{\tau }+\mathrm{i}\Phi -\mathrm{i}\gamma -\beta
\zeta \partial _{z}\right) ^{2}\Psi  \label{eqips} \\
+\left( \zeta \partial _{z}-\mathrm{i}Z+\mathrm{i}\beta \gamma \right)
^{2}\Psi -\zeta ^{2}G_{1}^{\prime }\left( \Psi ^{2}\right) \Psi -\Psi =0, 
\notag
\end{gather}%
where $\hat{\Psi}$ is given by (\ref{psigam}), (\ref{psi1sig}). We look for
a solution of (\ref{eqips}) in the strip $\Xi \left( \theta \right) $ in the
time-space: 
\begin{equation}
\Xi \left( \theta \right) =\left\{ \left( \tau ,z\right) :-\infty <\tau
<\infty ,\quad \left\vert z\right\vert <\theta \right\} .  \label{Ksi}
\end{equation}%
We expand (\ref{eqips}) with respect to $\Phi ,Z$ and rewrite equation (\ref%
{eqips}) in the form 
\begin{gather}
Q\Psi -\mathrm{i}\Phi \left( \zeta \partial _{\tau }-\mathrm{i}\gamma -\beta
\zeta \partial _{z}\right) \Psi -\mathrm{i}\left( \zeta \partial _{\tau }-%
\mathrm{i}\gamma -\beta \zeta \partial _{z}\right) \left( \Phi \Psi \right)
+\Phi ^{2}\Psi  \label{Qeq} \\
+\mathrm{i}Z\left( \zeta \partial _{z}+\mathrm{i}\beta \gamma \right) \Psi +%
\mathrm{i}\left( \zeta \partial _{z}+\mathrm{i}\beta \gamma \right) \left(
Z\Psi \right) -Z^{2}\Psi =0,  \notag
\end{gather}%
where we denote by $Q\Psi $ the term which does not involve $\Phi $ and $Z$
explicitly: 
\begin{equation}
Q=\frac{1}{\Psi }\left( -\left( \zeta \partial _{\tau }-\mathrm{i}\gamma
-\beta \zeta \partial _{z}\right) ^{2}\Psi +\left( \zeta \partial _{z}+%
\mathrm{i}\beta \gamma \right) ^{2}\Psi \right) -\zeta ^{2}G_{1}^{\prime
}\left( \Psi ^{2}\right) -1.  \label{Qdef}
\end{equation}%
Using (\ref{psi1sig}) and (\ref{siggam}) we conclude that the imaginary part
of $Q$ is zero:\ 
\begin{equation}
\func{Im}Q=2\zeta \gamma \partial _{\tau }\sigma +\zeta \partial _{\tau
}\gamma =0.  \label{imips}
\end{equation}%
Hence, \ $Q=\func{Re}Q\ $ can be written in the form 
\begin{equation}
Q=\zeta ^{2}\left( -\partial _{\tau }^{2}\Psi +\partial _{\tau }\beta
\partial _{z}\Psi +2\beta \partial _{\tau }\partial _{z}\Psi +\gamma
^{-2}\partial _{z}^{2}\Psi -G_{1}^{\prime }\left( \Psi ^{2}\right) \Psi
\right) /\Psi .  \label{ReQ2}
\end{equation}%
Now we show explicitly the dependence of $Q$ on $z$. Using (\ref{logb}) we
easily verify that the function $\Psi $ in (\ref{psi1sig}) satisfies an
equation\ similar to (\ref{stp}): 
\begin{equation}
\partial _{z}^{2}\Psi -G_{a}^{\prime }\left( \Psi \Psi ^{\ast }\right) \Psi
=2\sigma \Psi .  \label{G1D}
\end{equation}%
Taking into account (\ref{G1D}) we see that 
\begin{equation}
Q=\zeta ^{2}\left[ -\left( \partial _{\tau }^{2}\sigma +\left( \partial
_{\tau }\sigma \right) ^{2}\right) -z\partial _{\tau }\beta -2\beta
z\partial _{\tau }\sigma -\beta ^{2}\zeta ^{2}\left( z^{2}-1\right) +2\sigma %
\right] .  \label{Qzetz}
\end{equation}%
Now we rewrite the complex equation (\ref{Qeq}) as a system of two real
equations. The real part of (\ref{Qeq}) divided by $\Psi $ yields the
following quadratic equation 
\begin{equation}
Q-2\gamma \Phi +\Phi ^{2}-2\beta \gamma Z-Z^{2}=0.  \label{ReQeq}
\end{equation}%
The solution $Z$\ which is small for small $\Phi ,Q$\ is given by the
formula 
\begin{equation}
Z=\Theta \left( \Phi \right) =-\beta \gamma +\left( \Phi ^{2}-2\gamma \Phi
+\beta ^{2}\gamma ^{2}+Q\right) ^{1/2}\beta /\left\vert \beta \right\vert .
\label{ZfiQ}
\end{equation}%
The imaginary part of (\ref{Qeq}) divided by $\zeta \Psi $ yields equation 
\begin{equation}
-2\Phi \left( \partial _{\tau }-\beta \partial _{z}\right) \ln \Psi -\left(
\partial _{\tau }-\beta \partial _{z}\right) \Phi +\partial _{z}Z+2Z\partial
_{z}\ln \Psi =0  \label{ImQeq0}
\end{equation}%
where the coefficients are expressed in terms of $\Psi $ defined by (\ref%
{psigam}), (\ref{psi1sig}): 
\begin{equation}
\ln \Psi =\sigma -z^{2}/2,\text{\quad }\partial _{z}\ln \Psi =-z,\text{\quad 
}\partial _{\tau }\ln \Psi =\partial _{\tau }\sigma .  \label{dlog}
\end{equation}%
To determine a \emph{small} solution $\Phi $ of (\ref{ImQeq0}), (\ref{ZfiQ})
in the strip $\Xi $ we impose the condition%
\begin{equation}
\Phi =0\text{\ if\ }z=0,\text{\quad }-\infty <\tau <\infty .  \label{Zy0}
\end{equation}%
A solution $\Phi $ of the equation (\ref{ImQeq0}), where $Z\ =\Theta \left(
\Phi \right) $ satisfies (\ref{ZfiQ}),\ is a solution of the following
quasilinear first-order equation 
\begin{equation}
\partial _{\tau }\Phi -\beta \mathbf{\partial }_{z}\Phi -\Theta _{\Phi
}\left( \Phi \right) \mathbf{\partial }_{z}\Phi =-2\Phi \zeta \left(
\partial _{\tau }-\beta \partial _{z}\right) \ln \Psi +2\Theta \partial
_{z}\ln \Psi +\Theta _{z}  \label{Zfieq2}
\end{equation}%
where $\Theta _{\Phi }$ and $\Theta _{z}$\ are the partial derivatives of $%
\Theta \left( \Phi \right) $ (\ $\Theta $ depends on $z$ via $Q$) with $\Phi 
$ subjected to condition (\ref{Zy0}).

To prove the existence of a solution in a wide enough strip $\Xi $ and study
its properties we use the method of characteristics.

\subsection{Construction and properties of the auxiliary potential}

We introduce the characteristic equations for (\ref{Zfieq2}):\ 
\begin{equation}
dz/ds=-\Theta _{\Phi }\left( \Phi \right) -\beta \left( \tau \right) ,
\label{exLC1}
\end{equation}%
\begin{equation}
d\tau /ds=1,  \label{exLC1a}
\end{equation}%
\begin{equation}
d\Phi /ds=-2\Phi \left( \partial _{\tau }-\beta \partial _{z}\right) \ln
\Psi +2\Theta \left( \Phi \right) \partial _{z}\ln \Psi +\Theta _{z},
\label{exLC3}
\end{equation}%
with the initial data%
\begin{equation}
\tau _{s=0}=\tau _{0},\quad z_{s=0}=0.  \label{exLC2}
\end{equation}%
From (\ref{Zy0}) on the line $z=0$\ we derive the initial condition 
\begin{equation}
\Phi _{s=0}=0.  \label{LCin}
\end{equation}%
In the above equations the quantity $\Phi $\ is an independent variable
which coincides with $\Phi =\Phi \left( \tau ,z\right) $ on the
characteristic curves:\ 
\begin{equation}
\Phi =\Phi \left( \tau \left( \tau _{0},s\right) ,z\left( \tau _{0},s\right)
\right) .  \label{Fichar}
\end{equation}%
Equation (\ref{exLC1a}) can be solved explicitly:%
\begin{equation*}
\tau =\tau _{0}+s.
\end{equation*}%
Often, abusing notation, we will write $\Phi \left( \tau _{0},s\right) $ for
the solution of (\ref{exLC1})-(\ref{exLC3}) which is found independently
from the formula (\ref{Fichar}). The function $\Theta \left( \Phi \right) $
and its partial derivatives $\Theta _{z}$\ and $\Theta _{\Phi }$ are
determined by (\ref{ZfiQ}) and are well-defined and smooth\ as long as 
\begin{equation}
D\left( \tau ,z,\Phi \right) =\Phi ^{2}-2\gamma \left( \tau \right) \Phi
+\beta ^{2}\gamma ^{2}\left( \tau \right) +Q\left( \tau ,z\right) >0.
\label{Dgr0}
\end{equation}%
Then formula\ (\ref{ZfiQ}) determines the function $\Theta $ as an analytic
function of\ $Q,\Phi $. Taking into account (\ref{betep0}) we see that the
right-hand side of the system (\ref{exLC1}) (\ref{exLC3}) is a two time
continuously differentiable function of variables $\tau ,z,\Phi $ in the
domain in $\mathbb{R}^{3}$\ defined by the inequality (\ref{Dgr0}).

From the classical theorem on the existence, uniqueness and regular
dependence on the initial data and parameters of solutions of the Cauchy
problem for a system\ of ordinary differential equations (see, for instance, 
\cite{CoddL}), we readily obtain the following statement.

\begin{lemma}
\label{Lcharloc0}Let the initial data 
\begin{equation}
\tau _{s=0}=\tau _{0},\quad z_{s=0}=z_{0},\quad \Phi _{s=0}=\Phi _{0}
\label{inigen}
\end{equation}%
be such that the function $D$ defined by(\ref{Dgr0}) is positive, namely $%
D\left( \tau _{0},z_{0},\Phi _{0}\right) >0$. Then the system (\ref{exLC1}),
(\ref{exLC3}) with the initial condition (\ref{inigen})\ has a unique
solution $\tau \left( s\right) ,z\left( s\right) ,\Phi \left( s\right) $\
defined on a maximal (finite or infinite) interval $s_{-}<s<s_{+}$. If the
value\ $s_{\pm }$\ is finite,\ then%
\begin{equation}
\text{either\ }\lim_{s\rightarrow s_{\pm }}\left( \left\vert \Phi
\right\vert ^{2}+z^{2}\right) =\infty \text{\ \ or\ \ }\lim_{s\rightarrow
s_{\pm }}D\left( \tau ,z,\Phi \right) =0.  \label{Fizinfty}
\end{equation}%
\ The solution is\ a twice continuously differentiable function of $s$, of
the initial data\ $\tau _{0},z_{0},\Phi _{0},\ $and of the parameter $\zeta $%
, and $\tau \left( s\right) =\tau _{0}+s$.
\end{lemma}

\subsubsection{Properties of the characteristic curves}

The characteristic system (\ref{exLC1})--(\ref{exLC3}), (\ref{LCin})\
involves the function $Z=\Theta \left( \Phi \right) $\ defined by (\ref{ZfiQ}%
). We give sufficient conditions on the variables which ensure that (\ref%
{Dgr0}) holds.

\begin{proposition}
If%
\begin{equation}
\left\vert \Phi \right\vert \leq \check{\beta}^{2}/4  \label{File14}
\end{equation}%
and 
\begin{equation}
Q\geq -\check{\beta}^{2}/4  \label{Qle14}
\end{equation}%
then%
\begin{equation}
D\left( \tau ,z,\Phi \right) \geq \Phi ^{2}+\ \check{\beta}^{2}/4,
\label{Dgr1b}
\end{equation}%
and (\ref{Dgr0}) holds. Conditions (\ref{File14}) and (\ref{Qle14}) are
fulfilled if 
\begin{gather}
\zeta ^{1/3}\left\vert z\right\vert \leq 1,  \label{zetsm1} \\
\zeta ^{2}\sup_{\tau }\left\vert \partial _{\tau }^{2}\sigma +\left(
\partial _{\tau }\sigma \right) ^{2}-2\sigma -\beta ^{2}\zeta
^{2}\right\vert +\zeta ^{4/3}\sup_{\tau }\left\vert \partial _{\tau }\beta
+2\beta \partial _{\tau }\sigma \right\vert +\zeta ^{10/3}\leq \check{\beta}%
^{2}/8.  \label{zetsm2}
\end{gather}
\end{proposition}

\begin{proof}
If (\ref{File14}) and (\ref{Qle14}) hold, we take into account that $\gamma
\geq 1$ and conclude that 
\begin{equation}
\Phi ^{2}-2\gamma \Phi +Q+\beta ^{2}\gamma ^{2}\geq \Phi ^{2}+\ \check{\beta}%
^{2}/4>0.  \label{discrgr}
\end{equation}%
According to (\ref{Qzetz}) condition (\ref{Qle14}) takes the form 
\begin{equation*}
\zeta ^{2}\left( \partial _{\tau }^{2}\sigma +\left( \partial _{\tau }\sigma
\right) ^{2}-2\sigma -\beta ^{2}\zeta ^{2}\right) +z\zeta ^{2}\left(
\partial _{\tau }\beta +2\beta \partial _{\tau }\sigma \right) +\zeta
^{4}z^{2}\beta ^{2}\leq \check{\beta}^{2}/4.
\end{equation*}%
Hence (\ref{Qle14})\ is fulfilled if (\ref{zetsm1}), (\ref{zetsm2}) hold.
\end{proof}

Therefore $\Theta \left( \Phi \right) $ is a regular function of $\Phi $
which satisfies inequality (\ref{File14}) in the strip $\Xi \left( \theta
\right) $ defined by (\ref{Ksi}) as long as $\zeta $\ satisfies (\ref{zetsm2}%
).\ 

It is convenient to introduce the following notation. Let $2_{+}$ be a fixed
number which is arbitrary close to $2$ and $1_{+}$\ be arbitrarily\ close to 
$1$\ and satisfy the inequality 
\begin{equation}
1<2_{+}/2<1_{+}<2.  \label{pl2}
\end{equation}%
Below we obtain estimates of solutions of the characteristic equations in
the strip $\Xi \left( \bar{\theta}\right) $ using notation 
\begin{equation}
\bar{\theta}=\bar{\theta}\left( \zeta \right) =\ln ^{1/2_{+}}\zeta ^{-1}.
\label{thebar}
\end{equation}%
It is also\ convenient to introduce the following functions: 
\begin{equation}
b\left( \tau \right) =\beta ^{-1}\left( \tau \right) -\beta \left( \tau
\right) ,\quad \;B\left( \tau _{0},s\right) =\int_{\tau _{0}}^{\tau
_{0}+s}b\left( \tau \right) d\tau .  \label{btaus}
\end{equation}%
Using the inequality%
\begin{equation}
0<\hat{\epsilon}_{1}^{-1}-\hat{\epsilon}_{1}\leq \left\vert \beta
^{-1}-\beta \right\vert =\left( \beta ^{-1}-\beta \right) \beta /\left\vert
\beta \right\vert \leq \check{\beta}^{-1}-\check{\beta}  \label{betbet}
\end{equation}%
we obtain that 
\begin{equation}
\left( \hat{\epsilon}_{1}^{-1}-\hat{\epsilon}_{1}\right) \left\vert
s\right\vert \leq B\left( \tau _{0},s\right) \leq \left\vert s\right\vert
\left( \check{\beta}^{-1}-\check{\beta}\right) .  \label{bbet}
\end{equation}

\begin{lemma}
\label{Lfizet1}Let $\zeta \leq 1/C_{0}$\ where $C_{0}$ is sufficiently
large. Let\ $\tau ,z,\Phi $\ be a solution to (\ref{exLC1})- (\ref{LCin})
and assume that $z\left( s\right) ,\Phi \left( s\right) $\ are defined on an
interval $s_{1-}<s<s_{1+}$ with values in the strip $\Xi \left( \bar{\theta}%
\right) $, and that the following estimate holds 
\begin{equation}
\left\vert \Phi \left( \tau \left( s\right) ,z\left( s\right) \right)
\right\vert \leq \left\vert \zeta \right\vert \ \ \text{for }s_{1-}<s<s_{1+}.
\label{ests1s2}
\end{equation}%
\ Then there exist a constant$\ C_{5}>0$\ such that\ on the same interval\
we have: 
\begin{equation}
0<\left( \hat{\epsilon}_{1}^{-1}-\hat{\epsilon}_{1}\right) /2\leq \left\vert
dz/ds\right\vert \leq 2\left( \check{\beta}^{-1}-\check{\beta}\right) ,
\label{dzdsgr}
\end{equation}%
\begin{equation}
\left( \hat{\epsilon}_{1}^{-1}-\hat{\epsilon}_{1}\right) \left\vert
s\right\vert /2\leq \left\vert z\right\vert \leq 2\left\vert s\right\vert
\left( \check{\beta}^{-1}-\check{\beta}\right) ,  \label{zgrle}
\end{equation}%
\begin{equation}
\left\vert z-B\left( \tau _{0},s\right) \right\vert \leq C\zeta \left\vert
s\right\vert ,  \label{zminb}
\end{equation}%
\begin{equation}
\left\vert \Phi \right\vert \leq C_{5}\zeta ^{2}\left( \left\vert
z\right\vert ^{2}+\left\vert z\right\vert \right) \mathrm{e}^{\left\vert
z\right\vert ^{2}}.  \label{Filec3}
\end{equation}
\end{lemma}

\begin{proof}
If (\ref{ests1s2}) holds, conditions (\ref{zetsm1}), (\ref{zetsm2}) are
satisfied for $\zeta \leq 1/C_{0}$, implying that (\ref{Qle14}) and (\ref%
{Dgr1b}) are fulfilled.\ Equation (\ref{exLC1})\ can be written in the form 
\begin{equation}
dz/ds=\beta ^{-1}-\beta -\Theta _{\Phi }^{1}\left( \Phi \right)
\label{exLC11}
\end{equation}%
where 
\begin{gather}
\Theta _{\Phi }^{1}=\Theta _{\Phi }\left( \Phi \right) +\beta ^{-1}
\label{Thetfif} \\
=\frac{\Phi }{\left( \left( \Phi -\gamma \right) ^{2}-1+Q\right) ^{1/2}}%
\frac{\beta }{\left\vert \beta \right\vert }+\frac{\Phi ^{2}-2\gamma \Phi +Q%
}{\left( \left( \left( \Phi -\gamma \right) ^{2}-1+Q\right) ^{1/2}+\gamma
\left\vert \beta \right\vert \right) \left( \left( \Phi -\gamma \right)
^{2}-1+Q\right) ^{1/2}}.  \notag
\end{gather}%
Evidently $\Theta _{\Phi }^{1}$ has the form 
\begin{equation}
\Theta _{\Phi }^{1}=\Phi U_{1}+QU_{2},  \label{Thet1u}
\end{equation}%
where $U_{1}$ and $U_{2}$ are algebraic expressions which are analytic if (%
\ref{Dgr1b}) is fulfilled, and consequently they are bounded. Using (\ref%
{Qzetz})\ with $\left\vert z\right\vert \leq \bar{\theta}\ $and small $\zeta 
$ we conclude that 
\begin{equation}
\left\vert Q\right\vert =\zeta ^{2}\left\vert \partial _{\tau }^{2}\sigma
+\left( \partial _{\tau }\sigma \right) ^{2}-2\sigma -\beta ^{2}\zeta
^{2}+z\left( \partial _{\tau }\beta +2\beta \partial _{\tau }\sigma \right)
+\zeta ^{2}z^{2}\beta ^{2}\right\vert \leq C_{0}\zeta ^{2}\left( \left\vert
z\right\vert +1\right) ,  \label{Qlnz}
\end{equation}%
and obtain an estimate 
\begin{equation}
\left\vert \Theta _{\Phi }^{1}\right\vert \leq C_{1}\zeta \text{\ for }%
s_{1-}<s<s_{1+}.  \label{Tfi}
\end{equation}%
Hence, (\ref{exLC1})\ implies that 
\begin{equation}
\left\vert dz/ds-\beta ^{-1}+\beta \right\vert \leq C\zeta ^{2-\delta }\text{%
.}  \label{Tfi1}
\end{equation}%
\ We take $\zeta $ so small that $2C\zeta \leq \min (\hat{\epsilon}_{1}^{-1}-%
\hat{\epsilon}_{1},\check{\beta}^{-1}-\check{\beta})$, and using (\ref%
{betbet}) we obtain (\ref{dzdsgr}). Then integration implies (\ref{zgrle})\
and (\ref{zminb}).

Observe that equation (\ref{exLC3}) can be written in the form 
\begin{equation*}
\frac{d\Phi }{ds}=-2\Phi \partial _{\tau }\ln \Psi +2\left( \Theta \left(
\Phi \right) -\Phi \Theta _{\Phi }\left( \Phi \right) \right) \partial
_{z}\ln \Psi +2\Phi \left( \Theta _{\Phi }\left( \Phi \right) +\beta \right)
\partial _{z}\ln \Psi +\Theta _{z},
\end{equation*}%
which, in turn, with the use of (\ref{exLC1}) can be rewritten as 
\begin{equation}
\frac{d\Phi }{ds}=-2\Phi \partial _{\tau }\ln \Psi +2\left( \Theta \left(
\Phi \right) -\Phi \Theta _{\Phi }\left( \Phi \right) \right) \partial
_{z}\ln \Psi -2\Phi \frac{dz}{ds}\partial _{z}\ln \Psi +\Theta _{z}.
\end{equation}%
According to (\ref{Qle14})\ $\beta ^{2}\gamma ^{2}+Q>0$,\ and\ we rewrite (%
\ref{exLC3})\ in the form 
\begin{equation}
\frac{d\Phi }{ds}=-\Phi \frac{d\ln \Psi ^{2}}{ds}+2\Theta ^{1}\left( \Phi
\right) \partial _{z}\ln \Psi ^{2}-2\left( \beta \gamma -\left( \beta
^{2}\gamma ^{2}+Q\right) ^{1/2}\beta /\left\vert \beta \right\vert \right)
\partial _{z}\ln \Psi ^{2}+\Theta _{z}  \label{dfis}
\end{equation}%
where%
\begin{equation}
\Theta ^{1}\left( \Phi \right) =\Theta \left( \Phi \right) -\Phi \Theta
_{\Phi }\left( \Phi \right) +\beta \gamma -\left( \beta ^{2}\gamma
^{2}+Q\right) ^{1/2}\beta /\left\vert \beta \right\vert .  \label{Thet1}
\end{equation}%
Multiplying by $\Psi ^{2}$ we obtain 
\begin{equation}
\frac{d\left( \Psi ^{2}\Phi \right) }{ds}=2\Theta _{1}\left( \Phi \right)
z-2\left( \beta \gamma -\left( \beta ^{2}\gamma ^{2}+Q\right) ^{1/2}\right)
z+\Theta _{z}\Psi ^{2}.  \label{Psi2fieq}
\end{equation}%
According to (\ref{ZfiQ}) 
\begin{gather}
\frac{\beta }{\left\vert \beta \right\vert }\Theta ^{1}\left( \Phi \right) =-%
\frac{\Phi ^{2}}{\left( \left( \Phi -\gamma \right) ^{2}-1+Q\right) ^{1/2}}%
\frac{\left( \beta ^{2}\gamma ^{2}+Q\right) ^{1/2}}{\left( \beta ^{2}\gamma
^{2}+Q\right) ^{1/2}+\left( \left( \Phi -\gamma \right) ^{2}-1+Q\right)
^{1/2}}  \label{Thet2} \\
-\gamma \Phi ^{2}\frac{\Phi -2\gamma }{\left[ \left( \beta ^{2}\gamma
^{2}+Q\right) ^{1/2}+\left( \left( \Phi -\gamma \right) ^{2}-1+Q\right)
^{1/2}\right] ^{2}\left( \left( \Phi -\gamma \right) ^{2}-1+Q\right) ^{1/2}}.
\notag
\end{gather}%
Hence, taking into account (\ref{Dgr1b})\ we obtain%
\begin{equation*}
\left\vert \Theta ^{1}\left( \Phi \right) \right\vert \leq C_{1}\Phi ^{2}.
\end{equation*}%
To estimate remaining terms in (\ref{Psi2fieq}) we note that in\ $\Xi \left( 
\bar{\theta}\right) $ 
\begin{equation}
\left\vert \beta \gamma -\left( \beta ^{2}\gamma ^{2}+Q\right)
^{1/2}\right\vert \leq C\left\vert Q\right\vert \leq C^{\prime }\zeta
^{2}\left\vert z\right\vert .  \label{theq}
\end{equation}%
Using (\ref{Dgr1b})\ we obtain\ in\ $\Xi \left( \bar{\theta}\right) $ an
estimate 
\begin{equation}
\left\vert \Theta _{z}\left( \Phi \right) \right\vert =\frac{\left\vert
\zeta ^{2}\left( \partial _{\tau }\beta +2\beta \partial _{\tau }\sigma
\right) +2\zeta ^{4}z\beta ^{2}\right\vert }{2\left( \Phi ^{2}-2\gamma \Phi
+\beta ^{2}\gamma ^{2}+Q\right) ^{1/2}}\leq C_{1}\zeta ^{2}.  \label{Thzfile}
\end{equation}%
The above estimates combined with (\ref{Psi2fieq}) and (\ref{Thzfile}) yield%
\begin{equation}
\left\vert \frac{d\left( \Psi ^{2}\Phi \right) }{ds}\right\vert \leq
C_{2}\Phi ^{2}\left\vert z\right\vert +C^{\prime \prime }\zeta
^{2}\left\vert z\right\vert +C_{3}\zeta ^{2}.  \label{dpsifi}
\end{equation}%
Using (\ref{zgrle}) we obtain 
\begin{equation*}
\left\vert \Psi ^{2}\Phi \right\vert \leq C_{4}\zeta ^{2}s^{2}+C_{3}\zeta
^{2}\left\vert s\right\vert ,
\end{equation*}%
and, using (\ref{zgrle}) and the definition of $\Psi $,\ we obtain (\ref%
{Filec3}).\ 
\end{proof}

\begin{theorem}
\label{Tlfizet2}Let\ $\delta $\ be an arbitrary small fixed number
satisfying $0<\delta <\frac{2_{+}-2}{32}$, and suppose that $\zeta \leq
1/C_{0}$ is sufficiently small. Let $\ s_{2-}<s<s_{2+}$\ be a maximal
interval on which a solution $\tau ,z,\Phi $ \ of (\ref{exLC1})- (\ref{LCin}%
)\ is such\ that $z\left( s\right) ,\tau \left( s\right) $ takes values in
the strip $\Xi \left( \bar{\theta}\right) $ where $\bar{\theta}=\ln
^{1/2_{+}}\zeta ^{-1}$ \ and $\left\vert \Phi \left( s\right) \right\vert
\leq \zeta $. Then 
\begin{equation}
z\left( s_{2+}\right) =\bar{\theta}\beta /\left\vert \beta \right\vert
,\quad z\left( s_{2-}\right) =-\bar{\theta}\beta /\left\vert \beta
\right\vert ,\quad \bar{\theta}=\ln ^{1/2_{+}}\zeta ^{-1},  \label{zs2p}
\end{equation}%
and $z\left( s\right) ,\Phi \left( s\right) $ satisfy\ inequalities (\ref%
{Dgr1b}), (\ref{zgrle}) on this interval\ as well as the inequality%
\begin{equation}
\left\vert \Phi \right\vert \leq \zeta ^{2-\delta },\quad \left\vert \frac{%
d\Phi }{ds}\right\vert \leq C_{7}\zeta ^{2-\delta }.  \label{Filec4}
\end{equation}%
The constants $C_{7}$, $C_{0}$ depend on $\delta $ but do not depend on $%
\tau _{0},\zeta $.
\end{theorem}

\begin{proof}
Consider\ an interval\ $\left( s_{1-},s_{1+}\right) $ satisfying the
conditions of Lemma \ref{Lfizet1}.\ According to Lemma\ \ref{Lcharloc0} such
an interval exists. Consider the maximal interval\ $\left(
s_{1-},s_{1+}\right) $\ such that $z\left( s\right) ,\Phi \left( s\right) $
takes values in the strip $\Xi \left( \bar{\theta}\right) $\ and conditions
of Lemma \ref{Lfizet1} are fulfilled, and denote the maximal value of $%
s_{1+} $\ by\ $s_{2+}$\ and the minimal value of $s_{1-}$ by $s_{2-}$.
According to Lemma \ref{Lfizet1} inequality (\ref{Dgr1b})\ is satisfied on
every $\left( s_{1-},s_{1+}\right) $, and hence it is satisfied on $\left(
s_{2-},s_{2+}\right) $. On such an interval (\ref{Filec3}) is fulfilled, and
taking into account that $\left\vert z\right\vert \leq \ln ^{1/2_{+}}\zeta
^{-1}$\ we obtain (\ref{Filec4}). Note that (\ref{Filec3})\ implies that\ 
\begin{equation}
\left\vert \Phi \right\vert \leq 2C_{5}\zeta ^{2}\ln ^{2/2_{+}}\zeta
^{-1}\exp \ln ^{2/2_{+}}\zeta ^{-1}\leq \zeta ^{2-\delta }/2  \label{filezet}
\end{equation}%
if\ $\zeta $\ is small enough to satisfy\ 
\begin{equation*}
C_{5}\ln ^{2/2_{+}}\zeta ^{-1}\exp \left( \ln ^{2/2_{+}}\zeta ^{-1}-\delta
\ln \zeta ^{-1}\right) \leq 1/2.
\end{equation*}%
Hence (\ref{ests1s2}) holds under the above condition on $\zeta $ which
evidently does not depend on $\tau _{0}$. We assert that (\ref{zs2p}) is
fulfilled. Indeed, assume the contrary. If 
\begin{equation*}
\left\vert z\left( s_{2+}\right) \right\vert <\ln ^{1/2_{+}}\zeta ^{-1},
\end{equation*}%
we can extend the solution to a small interval with $\left\vert s\right\vert
>\left\vert s_{2+}\right\vert $ The solution stays in $\Xi \left( \bar{\theta%
}\right) \ $by continuity and (\ref{filezet}) implies (\ref{ests1s2}) by
continuity for small $s-s_{2+}$.\ This contradicts the assumption that $%
s_{2+}$ is maximal. The quantity $s_{2-}$ can be treated similarly yielding (%
\ref{zs2p}).\ The first inequality (\ref{Filec4}) follows from (\ref{filezet}%
). To obtain the inequality for the derivative\ we use (\ref{dpsifi}) and (%
\ref{filezet}) as follows: 
\begin{equation*}
\left\vert d\Phi /ds\right\vert \leq C_{5}s\zeta ^{4-2\delta }+C_{3}\zeta
^{2}+C_{6}\left\vert \Phi \right\vert \leq C_{7}\zeta ^{2-\delta }.
\end{equation*}
\end{proof}

We denote by $s_{2\pm }\left( \tau _{0}\right) $\ the value of $s_{2\pm }$
given by (\ref{zs2p}) which was found in Theorem\ \ref{Tlfizet2}. According
to (\ref{dzdsgr}) and the implicit function theorem the function $s_{2\pm
}\left( \tau _{0}\right) $ is differentiable. Let us introduce a set 
\begin{equation}
\Xi ^{\prime }\left( \bar{\theta}\right) =\left\{ \left( \tau _{0},s\right)
\in \mathbb{R}^{2}:\tau _{0}\in \mathbb{R},\quad s_{2-}\leq s\leq
s_{2+}\right\} ,  \label{Ksip}
\end{equation}%
where according to (\ref{zgrle})%
\begin{equation}
C^{-1}\ln ^{1/2_{+}}\zeta ^{-1}\leq \left\vert s_{2\pm }\right\vert \leq
C\ln ^{1/2_{+}}\zeta ^{-1}.  \label{s2lnz}
\end{equation}

\subsubsection{Properties of the auxiliary potentials}

Based on the properties of solutions to the characteristic equations we
establish here properties of solutions to the quasilinear equation (\ref%
{Zfieq2}). Notice that at $\zeta =0$\ the characteristic equations are
linear and the case of small $\zeta $ can be considered in $\Xi \left( \bar{%
\theta}\right) $ as a small perturbation.

We make use below\ of $C^{l}$ norms of a function of two variables defined
as follows:%
\begin{eqnarray}
\left\Vert u\right\Vert _{C^{l}\left( \Omega \right) } &=&\sup_{y\in \Omega
}\tsum\nolimits_{\left\vert \alpha \right\vert \leq l}\left\vert \partial
^{\alpha }u\left( y\right) \right\vert ,\text{\ where }y=\left(
y_{1},y_{2}\right) ,  \label{Cl} \\
\partial ^{\alpha } &=&\partial _{1}^{\alpha _{1}}\partial _{2}^{\alpha
_{2}},\alpha =\left( \alpha _{1},\alpha _{2}\right) ,\left\vert \alpha
\right\vert =\alpha _{1}+\alpha _{2}.  \notag
\end{eqnarray}%
If $\Omega =\Xi ^{\prime }\left( \bar{\theta}\right) $ we set in the formula
above $y=\left( \tau _{0},s\right) $, whereas if $\Omega =\Xi \left( \bar{%
\theta}\right) $ we set\ $y=\left( \tau ,z\right) $. Let $B\left( \tau
_{0},s\right) $ be defined by (\ref{btaus}). Obviously,%
\begin{equation}
\left\vert \partial ^{\alpha }B\left( \tau _{0},s\right) \right\vert \leq C,%
\text{\quad }1\leq \left\vert \alpha \right\vert \leq 2.  \label{dab}
\end{equation}

\begin{lemma}
\label{Lestzfi}Under conditions of Theorem \ref{Tlfizet2} functions $z\left(
\tau _{0},s;\zeta \right) $ and $\Phi \left( \tau _{0},s;\zeta \right) $
satisfy in $\Xi ^{\prime }\left( \bar{\theta}\right) ,$ $\bar{\theta}=\ln
^{1/2}\zeta ^{-1/2_{+}},$ estimates 
\begin{equation}
\left\Vert z-B\right\Vert _{C^{l}\left( \Xi ^{\prime }\left( \bar{\theta}%
\right) \right) }\leq C\zeta ^{2-\delta l-\delta },\quad l=0,1,2;
\label{zest3}
\end{equation}%
\begin{equation}
\left\Vert \Phi \left( \tau _{0},s\right) \right\Vert _{C^{l}\left( \Xi
^{\prime }\left( \bar{\theta}\right) \right) }\leq C\zeta ^{2-\delta
l-\delta },\quad l=0,1,2.  \label{fiest3}
\end{equation}
\end{lemma}

\begin{proof}
The derivation of the above estimates is straightforward but tedious, and we
present only the principal steps. The estimates are derived by induction in $%
l$. For $l=0$ we use Theorem \ref{Tlfizet2}, and then inequality (\ref%
{fiest3}) follows from (\ref{Filec4}).\ According to (\ref{Filec4})
inequality (\ref{ests1s2}) is fulfilled and then inequality (\ref{zest3})
for $l=0$\ follows from (\ref{zminb}) and (\ref{zgrle}). Consider now $l>0$
assuming that (\ref{zest3}), (\ref{fiest3})\ hold for $l-1$. Equations (\ref%
{exLC11}), (\ref{dfis}) can be written in the form%
\begin{equation}
\frac{d\left( z-B\right) }{ds}=-\Phi U_{11}-QU_{12},  \label{dz1}
\end{equation}%
\begin{equation}
\frac{d\Phi }{ds}=-\Phi \left( \partial _{\tau }\sigma -zU_{20}\right) +\Phi
^{2}U_{21}z+QU_{22}z+\zeta ^{2}U_{23}z+\zeta ^{2}U_{24},  \label{dfis1}
\end{equation}%
where $U_{ij}$ are algebraic functions of variables\ $\Phi ,Q,\gamma ,\beta $%
. These variables are bounded, and the derivatives of $\gamma ,\beta ,\sigma 
$ are bounded as well by (\ref{betep0}). The derivatives of the solutions,
namely $z^{\prime }=\partial ^{\alpha }z,\Phi ^{\prime }=\partial ^{\alpha
}\Phi $, $\left\vert \alpha \right\vert =l$, satisfy the equations obtained
by application of $\partial ^{\alpha }$\ to (\ref{exLC11}), (\ref{dfis}).
Since (\ref{Dgr1b}) is fulfilled, the coefficients $U_{ij}$\ and their
derivatives\ with respect to\ $\Phi ,Q,\gamma ,\beta $\ are bounded. The
only unbounded variables are\ $z\ $and $s$ according to (\ref{zgrle}), and
their upper bounds are \ respectively $\bar{\theta}$\ and\ $C\bar{\theta}$.
Observe that $z$\ enters $U_{11}$ and $U_{12},U_{21}$ only through $Q$.\ The
derivatives of $U_{ij}$\ or of $Q$\ up to $l$-th order cannot\ involve
powers of $z$\ higher than $z^{2l}$.\ Since $Q$\ given by (\ref{Qzetz})\
involves\ the factor $\zeta ^{2}$\ and $\left\vert z\right\vert \leq \bar{%
\theta}$\ in $\Xi ^{\prime }\left( \bar{\theta}\right) $,\ the derivatives\
of $Q$ of order $l$ with respect to $z$ or $\tau $\ are smaller than $\zeta
^{2-\delta }$ for small $\zeta $, and derivatives of $U_{ij}$ with respect
to\ $z$ or $\tau $\ are smaller than $\zeta ^{2-l\delta }$. We apply the
Leibnitz formula to the derivatives of (\ref{dz1}) and (\ref{dfis1}).\ Using
the induction assumption 
\begin{equation*}
\left\Vert \Phi \right\Vert _{C^{l-1}\left( \Xi ^{\prime }\left( \bar{\theta}%
\right) \right) }\leq C\zeta ^{2-l\delta },\qquad \left\Vert z-B\right\Vert
_{C^{l-1}\left( \Xi ^{\prime }\left( \bar{\theta}\right) \right) }\leq
C\zeta ^{2-l\delta },
\end{equation*}%
and notation (\ref{btaus}), we can write the equations in the form 
\begin{equation}
\frac{dz^{\prime }}{ds}=\partial ^{\alpha }b+U_{11}^{\prime }\Phi ^{\prime
}+U_{12}^{\prime }z^{\prime }+U_{13,\alpha },  \label{zp1}
\end{equation}%
\begin{equation}
\frac{d\Phi ^{\prime }}{ds}=U_{21}^{\prime }\Phi ^{\prime }+U_{22}^{\prime
}z^{\prime }+U_{23,\alpha },  \label{fip1}
\end{equation}%
where 
\begin{eqnarray}
\left\vert U_{11}^{\prime }\right\vert &\leq &C_{1},\qquad U_{12}^{\prime
}\leq \zeta ^{2-\delta },\qquad U_{13,\alpha }\leq \zeta ^{2-l\delta },
\label{zp11} \\
\left\vert U_{21}^{\prime }\right\vert &\leq &C_{0}\bar{\theta},\qquad
\left\vert U_{22}^{\prime }\right\vert \leq C\zeta ^{2-l\delta },\qquad
U_{23,\alpha }\leq C\zeta ^{2-l\delta }.  \notag
\end{eqnarray}%
\ \ The initial data for $z^{\prime }=\partial ^{\alpha }z$, $\Phi ^{\prime
}=\partial ^{\alpha }\Phi $ can also be analyzed by induction, since $%
\partial _{\tau _{0}}^{i}\Phi _{s=0}=0$, $\partial _{\tau _{0}}^{i}z_{s=0}=0$
and $\partial _{s}^{j}\partial _{\tau _{0}}^{i}\Phi \ $can\ be expressed
from the equation for $\partial _{s}^{j-1}\partial _{\tau _{0}}^{i}\Phi $.
Note that $\partial ^{\alpha }b=\partial _{\tau }^{\left\vert \alpha
\right\vert }b=\partial _{s}^{\left\vert \alpha \right\vert }b$ and (\ref%
{zp1}) can be rewritten in the form 
\begin{equation}
\frac{d}{ds}\left( z^{\prime }-\partial _{s}^{\left\vert \alpha \right\vert
-1}b\right) =U_{11}^{\prime }\Phi ^{\prime }+U_{12}^{\prime }\left(
z^{\prime }-\partial _{s}^{\left\vert \alpha \right\vert -1}b\right)
+U_{12}^{\prime }\partial _{s}^{\left\vert \alpha \right\vert
-1}b+U_{13,\alpha }.  \label{zp12}
\end{equation}%
Hence, we get\ from the induction assumptions estimates of the initial data
for small $\zeta $ 
\begin{equation*}
\left\vert \partial ^{\alpha }z_{s=0}-\partial _{s}^{\left\vert \alpha
\right\vert -1}b_{s=0}\right\vert \leq C_{1}\zeta ^{2-l\delta },\quad
\left\vert \partial ^{\alpha }\Phi _{s=0}\right\vert \leq C_{2}\zeta
^{2-l\delta },\quad \left\vert \alpha \right\vert =l.
\end{equation*}%
From the system (\ref{zp12}), (\ref{fip1}) we easily obtain for small $\zeta 
$ the following estimate 
\begin{equation}
\left\vert z^{\prime }-\partial _{s}^{\left\vert \alpha \right\vert
-1}b\right\vert +\left\vert \Phi ^{\prime }\right\vert \leq C_{3}\zeta
^{2-l\delta }\exp \left( C_{0}^{\prime }\bar{\theta}s\right) \leq C_{3}\zeta
^{2-l\delta }\exp \left( C_{0}^{\prime }\ln ^{2/2_{+}}\zeta ^{-1}\right)
\leq \zeta ^{2-l\delta -\delta }  \label{rough}
\end{equation}%
which implies desired (\ref{zest3}) and (\ref{fiest3}).
\end{proof}

Notice that solutions of the characteristic equations determine the function 
$\Phi \left( \tau \left( \tau _{0},s\right) ,z\left( \tau _{0},s\right)
\right) $ on the characteristic curves as a function of parameters $\tau
_{0},s$. Characteristic equations (\ref{exLC1})-(\ref{LCin}) also determine
a function\ $\ \Pi \left( \tau _{0},s\right) =\left( \tau \left( \tau
_{0},s\right) \text{, }z\left( \tau _{0},s\right) \right) $\ in $\left( \tau
,z\right) $-plane. To obtain the function $\Phi $ of independent variables $%
\left( \tau ,z\right) $ we have to find the inverse of $\Pi $, that is $%
\left( \tau _{0},s\right) =\Pi ^{-1}\left( \tau ,z\right) $ in a strip about
the line $z=0$.

\begin{lemma}
\label{Limageth}The image of the mapping $\Pi $ contains the strip\ $\Xi
\left( \bar{\theta}\right) $ if\ $\zeta \leq 1/C$ for sufficiently large
constant $C$.\ 
\end{lemma}

\begin{proof}
The mapping $\Pi $ maps the straight line $\left\{ \tau _{0},s=0\right\} $
onto the straight line $\left\{ \tau ,z=0\right\} $. The straight line $%
\left\{ \tau _{0}=\tau _{00},s\right\} $\ is mapped onto the curve $\tau
=\tau _{00}+s,z=z\left( \tau _{00},s\right) $. This curve intersects
straight line $\left\vert z_{0}\right\vert =\bar{\theta}$\ and,
consequently, any straight line $z=z_{0}$\ with $\left\vert z_{0}\right\vert
\leq \bar{\theta}$. This intersection is\ transversal according to (\ref%
{dzdsgr}). Hence the point of intersection $\tau =p\left( \tau
_{00},z_{0}\right) $\ continuously and differentiably depends on $\tau _{00}$%
.\ Formula $\tau =\tau _{00}+s$\ implies that $p\left( \tau
_{00},z_{0}\right) \rightarrow \pm \infty $\ as $\tau _{00}\rightarrow \pm
\infty $. Since $p\left( \tau _{00},z_{0}\right) $\ is a continuous
function, it takes all intermediate values on the straight line, therefore
the image of the mapping $\Pi $\ contains every straight line $z=z_{0}$ with 
$\left\vert z_{0}\right\vert \leq \bar{\theta}$.
\end{proof}

Now we want to prove that the mapping $\Pi $ is one-to-one on\ $\Xi ^{\prime
}\left( \bar{\theta}\right) $ and that its inverse has uniformly bounded
derivatives in the strip $\Xi \ \left( \bar{\theta}\right) $. The
characteristic system depends on the small parameter $\zeta $. Therefore $%
\Pi \left( \tau _{0},s\right) =\Pi \left( \tau _{0},s;\zeta \right) $ and
its differential is given by the matrix%
\begin{equation}
\Pi ^{\prime }\left( \tau _{0},s;\zeta \right) =\left( 
\begin{array}{cc}
\partial \tau /\partial \tau _{0} & \partial \tau /\partial s \\ 
\partial z/\partial \tau _{0} & \partial z/\partial s%
\end{array}%
\right) =\left( 
\begin{array}{cc}
1 & 1 \\ 
\partial z/\partial \tau _{0} & \partial z/\partial s%
\end{array}%
\right) .  \label{Pipr}
\end{equation}%
If the matrix determinant is not zero, the inverse is given by the formula 
\begin{equation*}
\Pi ^{\prime -1}\left( \tau _{0},s;\zeta \right) =\frac{1}{\partial
z/\partial s-\partial z/\partial \tau _{0}}\left( 
\begin{array}{cc}
\partial z/\partial s & -1 \\ 
-\partial z/\partial \tau _{0} & 1%
\end{array}%
\right) .
\end{equation*}%
For $\zeta =0$ we obtain from (\ref{exLC1})-(\ref{exLC3}) a simpler system
for the resulting approximation $\mathring{\Phi}$: 
\begin{gather}
\frac{d\tau }{ds}=1,\quad dz/ds=\left( \beta ^{-1}-\beta \right) \left( \tau
_{0}+s\right) ,  \label{exLC10} \\
\frac{d\mathring{\Phi}}{ds}=-2\mathring{\Phi}\left( \partial _{\tau }-\beta
\partial _{z}\right) \ln \Psi -2\beta ^{-1}\mathring{\Phi}\partial _{z}\ln
\Psi .  \label{exLC31}
\end{gather}%
The solution of (\ref{exLC10}) is given by the formula 
\begin{equation}
\tau =\tau _{0}+s,\quad z=z_{0}=B\left( \tau _{0},s\right) ,  \label{ztau0}
\end{equation}%
where $B\left( \tau _{0},s\right) ,b\left( \tau \right) $ are given by (\ref%
{btaus}). The differential of $\Pi \left( \tau _{0},s;0\right) $\ is given
by the matrix%
\begin{equation}
\Pi ^{\prime }\left( \tau _{0},s;0\right) =\left( 
\begin{array}{cc}
1 & 1 \\ 
b\left( \tau _{0}+s\right) -b\left( \tau _{0}\right) & b\left( \tau
_{0}+s\right)%
\end{array}%
\right)  \label{Pipr0}
\end{equation}%
with the determinant $\det \Pi \left( \tau _{0},s;0\right) =b\left( \tau
_{0}\right) $.\ The matrix is invertible according to (\ref{betbet}), and
the inverse matrix is uniformly bounded.\ 

When $\zeta >0$ is small, we consider the system (\ref{exLC1})- (\ref{exLC3}%
), (\ref{LCin})\ as a perturbation of the system with $\zeta =0$\ and the
differential $\Pi ^{\prime }\left( \tau _{0},s;\zeta \right) $ is also a
small perturbation of $\Pi ^{\prime }\left( \tau _{0},s;0\right) $.

\begin{lemma}
\label{Lpiinverse}Let conditions of Theorem \ref{Tlfizet2} be satisfied.
Then\ in $\Xi \left( \bar{\theta}\right) $ 
\begin{equation}
\left\vert \Pi ^{\prime }\left( \tau _{0},s;\zeta \right) -\Pi ^{\prime
}\left( \tau _{0},s;0\right) \right\vert \leq C\zeta ^{2-\delta }
\label{Pilez}
\end{equation}%
for $\left\vert \zeta \right\vert \leq 1/C_{0},$ the matrices $\Pi ^{\prime
}\left( \tau _{0},s;\zeta \right) $ are invertible, the inverse matrices $%
\Pi ^{\prime -1}\left( \tau _{0},s;\zeta \right) $ have continuously
differentiable elements and their derivatives are uniformly bounded in $\Xi
\left( \bar{\theta}\right) .$The mapping $\Pi $ is one-to-one from $\Xi
^{\prime }\left( \bar{\theta}\right) $ to $\Xi \left( \bar{\theta}\right) $,
the mappings $\Pi $ and $\Pi ^{-1}$\ are two times differentiable with
uniformly bounded derivatives.
\end{lemma}

\begin{proof}
We use Lemma \ref{Lestzfi}\ and infer from (\ref{zest3}) with $l=1$ that 
\begin{equation*}
\left\vert \partial z/\partial \tau _{0}-\partial z_{0}/\partial \tau
_{0}\right\vert +\left\vert \partial z/\partial s-\partial z_{0}/\partial
s\right\vert \leq C\zeta ^{2-\delta }.
\end{equation*}%
This inequality implies inequality (\ref{Pilez}).\ Inequality (\ref{Pilez}),
in turn,\ implies that $\Pi ^{\prime }\left( \zeta \right) =\Pi ^{\prime
}\left( \tau _{0},s;\zeta \right) $ is a uniformly small perturbation of the
invertible matrix (\ref{Pipr0}) $\Pi ^{\prime }\left( 0\right) $, and hence $%
\Pi ^{\prime -1}\left( \zeta \right) $ is a uniformly small perturbation of
the matrix\ $\Pi ^{\prime -1}\left( 0\right) $. The matrices $\Pi ^{\prime
}\left( \zeta \right) $ and $\Pi ^{\prime -1}\left( \zeta \right) $ have
continuously differentiable entries. Since $\left\vert \Pi ^{\prime
-1}\left( 0\right) \right\vert $ is uniformly bounded and the derivatives of
entries of $\Pi ^{\prime }\left( 0\right) $\ are uniformly bounded and
continuous in $\Xi \left( \bar{\theta}\right) $, then the derivatives of
entries of $\Pi ^{\prime -1}\left( 0\right) $ are uniformly bounded and
continuous in $\Xi \left( \bar{\theta}\right) $ as well.\ The mapping $\Pi
\left( \zeta \right) $\ is one-to-one since it is a local diffeomorphism
between $\Xi ^{\prime }\left( \bar{\theta}\right) $\ and $\Xi \left( \bar{%
\theta}\right) $ and\ the image is a simply connected domain. Since $\Pi $
is two times continuously differentiable with uniformly bounded derivatives
and\ $\Pi ^{\prime -1}\left( \zeta \right) $ is uniformly bounded, the
inverse mapping $\Pi ^{-1}\left( \zeta \right) $\ is two times continuously
differentiable\ with uniformly bounded derivatives.
\end{proof}

\begin{theorem}
For any $\delta >0$ there exists such $C_{0}$\ \ that if $\zeta \leq 1/C_{0}$%
\ then there exists a solution $\Phi \left( \tau ,z\right) $ of the
quasilinear equation (\ref{Zfieq2})\ defined in $\Xi \left( \bar{\theta}%
\right) $. This solution is twice continuously differentiable in\ the strip $%
\Xi \left( \bar{\theta}\right) $\ and its derivatives are uniformly bounded
and small for $\zeta \leq 1/C_{0}$, namely%
\begin{equation}
\left\Vert \Phi \left( \tau ,z\right) \right\Vert _{C^{2}\left( \Xi \left( 
\bar{\theta}\right) \right) }\leq C_{1}\zeta ^{2-3\delta }.  \label{Fiztau}
\end{equation}
\end{theorem}

\begin{proof}
The solution $\Phi \left( \tau ,z\right) $ is defined by formula (\ref%
{Fichar})\ which can be written in the form 
\begin{equation}
\Phi \left( \tau ,z;\zeta \right) =\Phi \left( \Pi ^{-1}\left( \tau
_{0},s;\zeta \right) ,\zeta \right) .  \label{Fitz}
\end{equation}%
The function $\Phi \left( \tau ,z;\zeta \right) $\ is well-defined in $\Xi
\left( \bar{\theta}\right) $\ according to Lemmas \ref{Limageth}\ and \ref%
{Lpiinverse}. Its differentiability properties follow from\ properties of $%
\Phi \left( \tau _{0},s\ ;\zeta \right) $ described in (\ref{fiest3}) and\
properties of $\Pi ^{-1}\left( \tau _{0},s;\zeta \right) $ described in
Lemma \ref{Lpiinverse}. It is a solution to (\ref{Zfieq2})\ according to the
construction of $\Phi \left( \tau _{0},s;\zeta \right) $ as a solution of\ (%
\ref{exLC1})-(\ref{exLC3}).
\end{proof}

The second auxiliary phase $Z$\ is given by the formula (\ref{ZfiQ}). Since (%
\ref{Dgr1b}) holds in $\Xi \left( \bar{\theta}\right) $,\ $Z\left( \tau
,z\right) $ is also twice continuously differentiable in\ the strip $\Xi
\left( \bar{\theta}\right) $ and has uniformly bounded derivatives. The
potential $\varphi _{\mathrm{b}}$ and phase $S$ can be found from (\ref{denZ}%
), (\ref{denZ2}), namely 
\begin{equation}
S=\zeta ^{-1}\int_{0}^{z}Z\left( \tau ,z_{1}\right) dz_{1}=\zeta
^{-1}\int_{0}^{z}\Theta \left( \Phi \right) dz_{1},  \label{Sfor}
\end{equation}%
and%
\begin{equation}
\varphi _{\mathrm{b}}=\frac{m\mathrm{c}^{2}}{q}\left( \Phi +\mathbf{\partial 
}_{\tau }\int_{0}^{z}\Theta \left( \Phi \right) dz_{1}-\beta \Theta \left(
\Phi \right) \right) .  \label{fi2for}
\end{equation}

\begin{lemma}
\label{Lsfi2} The phase function $S$ defined by (\ref{Sfor}) and the
potential $\varphi _{\mathrm{b}}$ defined by (\ref{fi2for}) satisfy the
estimates 
\begin{equation}
\left\Vert S\right\Vert _{C^{2}\left( \Xi \left( \bar{\theta}\right) \right)
}\leq C_{1}\zeta ^{1-4\delta },  \label{sest}
\end{equation}%
\begin{equation}
\left\Vert \varphi _{\mathrm{b}}\right\Vert _{C^{1}\left( \Xi \left( \bar{%
\theta}\right) \right) }\leq C_{2}\zeta ^{2-4\delta }.  \label{fi2est}
\end{equation}
\end{lemma}

\begin{proof}
The above estimates follow from (\ref{Fiztau})\ and the boundedness in $\Xi
^{\prime }\left( \bar{\theta}\right) $ of derivatives of the function $%
\Theta \left( \Phi \right) $ which enters representations (\ref{Sfor}) and (%
\ref{fi2for}).
\end{proof}

\subsection{Verification of the concentration conditions}

In this section we fix an interval $T_{-}\leq t\leq T_{+}$\ and define
sequences $a=a_{n}$, $\zeta =\zeta _{n}$, $R=R_{n}$.\ We verify then that
the KG equation is localized at the trajectory $\mathbf{\hat{r}}\left(
t\right) =\left( 0,0,r\left( t\right) \right) $ and that solutions $\psi $
defined by (\ref{psi1d}), (\ref{psipsi}) concentrate at $\mathbf{\hat{r}}%
\left( t\right) $\ as in Theorem \ref{Ttrajcon}).

\begin{proposition}
\label{Lestnormfi1}Let $\varphi \left( t\right) $\ be defined by (\ref{fiex}%
)-(\ref{new11})\ where $\varphi _{\mathrm{b}}$\ is given by (\ref{fi2for})\
and $\varphi _{\mathrm{0}}\left( t\right) $ satisfy (\ref{fi0est}), $\varphi
_{\infty }=\varphi _{\mathrm{ac}}\left( t\right) $\ given by (\ref{fiac}), (%
\ref{new11}), $R_{n}=\bar{\theta}a_{n}\rightarrow 0$. Let also 
\begin{equation}
a_{\mathrm{C}}\leq Ca^{2}.  \label{aca2}
\end{equation}%
Then $\varphi _{n}$\ satisfy (\ref{Abound})\ \ and (\ref{Alocr}), and
condition (\ref{fihatb}) holds.
\end{proposition}

\begin{proof}
Note that $\varphi \left( t,x\right) $ does not depend on $x_{1},x_{2}$ and\
estimates in $\hat{\Omega}\left( \mathbf{\hat{r}}(t),R_{n}\right) $ follow
from estimates in $\Xi \left( \bar{\theta}\right) $:%
\begin{gather}
\max_{\hat{\Omega}\left( \mathbf{\hat{r}}(t),R_{n}\right) }|\varphi \left(
t,x\right) \mathbf{|\leq }\sup_{\tau ,z\in \Xi \left( \bar{\theta}\right)
}|\varphi \left( \tau ,z\right) \mathbf{|,}  \label{xtzt} \\
\max_{\hat{\Omega}\left( \mathbf{\hat{r}}(t),R_{n}\right) }|\partial
_{t}\varphi \left( t,x\right) \mathbf{|}\leq \frac{\mathrm{c}}{a}\sup_{\tau
,z\in \Xi \left( \bar{\theta}\right) }|\partial _{\tau }\varphi \mathbf{%
|,\quad }\max_{\hat{\Omega}\left( \mathbf{\hat{r}}(t),R_{n}\right) }|\nabla
\varphi \left( t,x\right) \mathbf{|}\leq \frac{1}{a}\sup_{\tau ,z\in \Xi
\left( \bar{\theta}\right) }|\nabla _{z}\varphi \mathbf{|.}  \notag
\end{gather}%
According to (\ref{fiex})%
\begin{equation*}
|\varphi \left( t,x\right) \mathbf{|}\leq \left\vert \varphi _{\mathrm{0}%
}\right\vert +|\varphi _{\mathrm{ac}}^{\prime }\mathbf{\mathbf{|}}\left\vert
y\right\vert +\left\vert \varphi _{\mathrm{b}}\right\vert \leq |\varphi _{%
\mathrm{0}}\mathbf{|+}\left\vert \varphi _{\mathrm{ac}}^{\prime }\right\vert
R_{n}\mathbf{+}\left\vert \varphi _{\mathrm{b}}\right\vert .
\end{equation*}%
Since $\varphi _{\mathrm{ac}}^{\prime }$ is defined by(\ref{new11}),\ we
conclude using (\ref{fi2est}), (\ref{fi0est}) that $|\varphi \mathbf{|}\leq
C_{1}$\ in $\Omega \left( \mathbf{\hat{r}}(t),R_{n}\right) $. Using (\ref%
{aca2}) we obtain estimate of derivatives%
\begin{equation}
|\partial _{t}\varphi \left( t,x\right) \mathbf{|}\leq |\partial _{t}\varphi
_{\mathrm{0}}|+|\partial _{t}\varphi _{\mathrm{ac}}^{\prime }\mathbf{\mathbf{%
|}}\left\vert y\right\vert +\mathrm{c}a^{-1}\left\vert \partial _{\tau
}\varphi _{\mathrm{b}}\right\vert \leq C+C_{1}a^{-1}\zeta ^{2-4\delta }\leq
C+C_{2}a^{1-4\delta }\leq C_{3},  \label{fic2}
\end{equation}%
\begin{equation}
|\nabla _{x}\varphi \left( t,x\right) \mathbf{|}\leq |\varphi _{\mathrm{ac}%
}^{\prime }\mathbf{\mathbf{|}}+a^{-1}\left\vert \nabla _{z}\varphi _{\mathrm{%
b}}\right\vert \leq C^{\prime }+C_{1}^{\prime }a^{-1}\zeta ^{2-4\delta }\leq
C_{3}^{\prime }  \label{fic3}
\end{equation}%
Hence, (\ref{Abound}) holds.\ According to (\ref{fiinf}), (\ref{fiex}) and (%
\ref{fiac}) $\varphi -\varphi _{\infty }\ =\varphi _{\mathrm{b}}$.\ Using (%
\ref{fi2est}), (\ref{xtzt}),\ and observing that (\ref{aca2}) implies $\zeta
\leq Ca$,\ we conclude that\ 
\begin{equation*}
|\varphi _{\mathrm{b}}\left( t,x\right) \mathbf{|}+|\nabla _{0,x}\varphi _{%
\mathrm{b}}\left( t,x\right) \mathbf{\mathbf{\mathbf{|}}\leq }\sup_{\tau
,z\in \Xi \left( \bar{\theta}\right) }(|\varphi _{\mathrm{b}}\left( \tau
,z\right) +\frac{1}{a}|\partial _{\tau }\varphi _{\mathrm{b}}\mathbf{|}+%
\frac{1}{a}|\nabla _{z}\varphi _{\mathrm{b}}\mathbf{|)}\leq C_{4}a^{-1}\zeta
^{2-4\delta }\rightarrow 0
\end{equation*}%
yielding relations (\ref{Alocr}). To obtain (\ref{fihatb}) we use (\ref%
{betep1})\ and (\ref{new11}).
\end{proof}

The solution of the KG equation\ is given by the formula (\ref{psipsi}), (%
\ref{psi1d}) where $s\left( t\right) $ is given by (\ref{sprime})\ with
condition $s\left( 0\right) =0$.

\begin{proposition}
\label{Lestnormpsi1}Let\ $\psi $ be defined by (\ref{psi1d}), (\ref{psipsi})
where\ $s,S$\ are defined by (\ref{sprime}), (\ref{new11}), let $R_{n}=\bar{%
\theta}a_{n}$. Then in $\hat{\Omega}\left( \mathbf{\hat{r}}(t),R_{n}\right) $
\begin{equation}
\left\vert \partial _{t}\psi \right\vert ^{2}+\left\vert \nabla \psi
\right\vert ^{2}\leq C_{1}a_{\mathrm{C}}^{-2}\mathring{\psi}^{2},
\label{dtpsi0}
\end{equation}%
\begin{equation}
G(\left\vert \psi \right\vert ^{2})\leq C_{2}\zeta ^{2-\delta }a_{\mathrm{C}%
}^{-2}\mathring{\psi}^{2},  \label{Gpsi0}
\end{equation}%
where 
\begin{equation*}
\mathring{\psi}=\mathring{\psi}\left( \mathbf{y}/a\right) =\mathring{\psi}%
\left( \mathbf{z}\right) =\pi ^{-3/4}a^{-3/2}\mathrm{e}^{-\left\vert \mathbf{%
y}\right\vert ^{2}a^{-2}/2}.
\end{equation*}
\end{proposition}

\begin{proof}
According to (\ref{psi1d}), (\ref{psipsi}) $\psi =\mathrm{e}^{\mathrm{i}\hat{%
S}\left( t,\left( x-r\right) \right) }\hat{\Psi}$\ with $\hat{\Psi}=\mathrm{e%
}^{\sigma }\mathring{\psi}\left( \mathbf{x}-\mathbf{\hat{r}}\right) $. For
such solutions we use the change of variables (\ref{xry}) and relations (\ref%
{new11}) and (\ref{sprime}),\ and we obtain similarly to (\ref{KGy1})\ and (%
\ref{Kpsy}) 
\begin{equation*}
\left\vert \partial _{0}\psi \right\vert ^{2}=a_{\mathrm{C}}^{-2}\zeta
^{2}\left\vert \partial _{\tau }\mathring{\psi}-\beta \partial \mathring{\psi%
}/\partial z_{3}\right\vert ^{2}+a_{\mathrm{C}}^{-2}\left\vert \gamma -\zeta
\partial _{\tau }S+\zeta \beta \partial _{z}S\right\vert ^{2}\left\vert 
\mathring{\psi}\right\vert ^{2},
\end{equation*}%
\begin{gather*}
\left\vert \nabla _{x}\psi \right\vert ^{2}=a^{-2}\left\vert \partial 
\mathring{\psi}/\partial z_{1}\right\vert ^{2}+a^{-2}\left\vert \partial 
\mathring{\psi}/\partial z_{2}\right\vert ^{2}+a^{-2}\left\vert \partial 
\mathring{\psi}/\partial z_{3}\right\vert ^{2}+a_{\mathrm{C}}^{-2}\left(
-Z+\beta \gamma \right) ^{2}\left\vert \mathring{\psi}\right\vert ^{2} \\
=a^{-2}\left\vert \nabla _{z}\mathring{\psi}\right\vert ^{2}+a_{\mathrm{C}%
}^{-2}\left( -Z+\beta \gamma \right) ^{2}\left\vert \mathring{\psi}%
\right\vert ^{2}.
\end{gather*}%
Note that in $\hat{\Omega}\left( \mathbf{\hat{r}}(t),R_{n}\right) $ 
\begin{equation}
\left\vert \partial \mathring{\psi}/\partial z_{i}\right\vert ^{2}\leq
C\left( 1+R_{n}^{2}\right) \left\vert \mathring{\psi}\right\vert ^{2}.
\label{dpsi0}
\end{equation}%
Using (\ref{sest})\ and (\ref{dpsi0}) we get (\ref{dtpsi0}). According to (%
\ref{paf3}) 
\begin{equation*}
G\left( \left\vert \psi \right\vert ^{2}\right) =-a^{-2}e^{2\sigma }%
\mathring{\psi}^{2}\left( z\right) \left[ -\mathbf{z}^{2}+2\sigma +\ln \pi
^{3/2}+2\right] ,
\end{equation*}%
yielding the desired inequality (\ref{Gpsi0}).
\end{proof}

\begin{proposition}
\label{Lendens1}Let\ $\psi $ be defined by (\ref{psi1d}), (\ref{psipsi})
where\ $s,S$\ are defined by (\ref{sprime}), (\ref{new11}). Then%
\begin{gather}
\mathcal{E}\left( \mathbf{\hat{r}}+a\mathbf{z}\right) =\frac{m\mathrm{c}^{2}%
}{2}\left[ \zeta ^{2}\left( \sigma +\beta z\right) ^{2}+2\zeta ^{2}\mathbf{z}%
^{2}-\zeta ^{2}\left( 2\sigma +\ln \pi ^{3/2}+2\right) \right] e^{2\sigma }%
\mathring{\psi}^{2}\left( z\right)  \label{Edens} \\
+\frac{m\mathrm{c}^{2}}{2}\left[ \left( \Phi -\gamma \right) ^{2}+\left(
\beta \gamma -Z\right) ^{2}+1\right] e^{2\sigma }\mathring{\psi}^{2}\left(
z\right)  \notag
\end{gather}
\end{proposition}

\begin{proof}
Using the change of variables (\ref{xry}), (\ref{zztau}) (\ref{denZ}), (\ref%
{denZ2})\ we obtain similarly to (\ref{KGy1})\ and (\ref{Kpsy}) 
\begin{equation*}
|\tilde{\partial}_{t}\psi |^{2}=\frac{\zeta ^{2}\mathrm{c}^{2}}{a_{\mathrm{C}%
}^{2}}\left\vert \left( \partial _{\tau }-\beta \partial _{z}\right) 
\mathring{\psi}\right\vert ^{2}+\frac{\mathrm{c}^{2}}{a_{\mathrm{C}}^{2}}%
\left( -\gamma +\Phi \right) ^{2}\left\vert \mathring{\psi}\right\vert ^{2}.
\end{equation*}%
The energy density\ in $\tau ,z$ variables has the form 
\begin{gather*}
\mathcal{E}\left( \mathbf{\hat{r}}+a\mathbf{z}\right) =\frac{\chi ^{2}}{2m}%
\left[ \frac{\zeta ^{2}}{a_{\mathrm{C}}^{2}}\left\vert \left( \partial
_{\tau }-\beta \partial _{z}\right) \mathring{\psi}\right\vert ^{2}+\frac{%
\zeta ^{2}}{a_{\mathrm{C}}^{2}}\left( \left\vert \nabla _{z}\mathring{\psi}%
\right\vert ^{2}+G\left( \left\vert \mathring{\psi}\right\vert ^{2}\right)
\right) \right] \\
+\frac{\chi ^{2}}{2m}\left[ \frac{1}{a_{\mathrm{C}}^{2}}\left( -\gamma +\Phi
\right) ^{2}\left\vert \mathring{\psi}\right\vert ^{2}+\frac{1}{a_{\mathrm{C}%
}^{2}}\left( \beta \gamma -Z\right) ^{2}\left\vert \mathring{\psi}%
\right\vert ^{2}+\frac{m^{2}\mathrm{c}^{2}}{\chi ^{2}}\left\vert \mathring{%
\psi}\right\vert ^{2}\right] .
\end{gather*}%
Using\ (\ref{psi1sig}) we obtain (\ref{Edens}).
\end{proof}

\begin{proposition}
\label{Lintbound}Let $\zeta _{n}$ and $a_{n}$ be related by the formula 
\begin{equation}
\ln \zeta ^{-1}=\ln ^{1_{+}}a^{-1}.  \label{azzet}
\end{equation}%
Then conditions\ (\ref{psiouta}) and\ (\ref{locpsibound}) are satisfied.\ 
\end{proposition}

\begin{proof}
To obtain\ (\ref{locpsibound})\ we use Proposition \ref{Lestnormpsi1}:%
\begin{gather*}
a_{\mathrm{C}}^{2}\int_{\Omega \left( \mathbf{\hat{r}}(t),R_{n}\right)
}\left\vert \nabla _{0,x}\psi \left( t,x\right) \right\vert
^{2}+|G(\left\vert \psi \left( t,x\right) \right\vert ^{2})|\,\mathrm{d}%
^{3}x+\int_{\Omega \left( \mathbf{\hat{r}}(t),R_{n}\right) }\left\vert \psi
\left( t,x\right) \right\vert ^{2}\,\mathrm{d}^{3}x \\
\leq C_{2}\int_{\Omega \left( \mathbf{\hat{r}}(t),R_{n}\right) }\mathring{%
\psi}^{2}\,\mathrm{d}^{3}x=C_{2}a^{-3}\int_{\Omega \left( 0,R_{n}\right) }%
\mathring{\psi}^{2}\left( y/a\right) \,\mathrm{d}^{3}y\leq C_{4}\int_{%
\mathbb{R}^{3}}\exp \left( -\left\vert z\right\vert ^{2}\right) \mathrm{d}%
z\leq C_{5}.
\end{gather*}%
To obtain (\ref{psiouta}) we once again use Proposition \ref{Lestnormpsi1} 
\begin{gather*}
\int_{\partial \Omega \left( \mathbf{\hat{r}}(t),R_{n}\right) }(a_{\mathrm{C}%
}^{2}\left\vert \nabla _{0,x}\psi \left( t,x\right) \right\vert ^{2}+a_{%
\mathrm{C}}^{2}|G(\left\vert \psi \left( t,x\right) \right\vert
^{2})|+\left\vert \psi \left( t,x\right) \right\vert ^{2}\,\mathrm{d}%
^{2}\sigma \leq \int_{\partial \Omega \left( \mathbf{\hat{r}}%
(t),R_{n}\right) }C_{2}\mathring{\psi}^{2}\,\mathrm{d}^{2}\sigma \\
=\int_{\partial \Omega \left( \mathbf{\hat{r}}(t),R_{n}\right) }C_{2}%
\mathring{\psi}^{2}\,\mathrm{d}^{2}\sigma =C_{3}a^{2}a^{-3}\int_{\left\vert
z\right\vert =\bar{\theta}}e^{-\left\vert z\right\vert ^{2}}\,\mathrm{d}%
^{2}\sigma =C_{4}\exp \left( \ln a^{-1}-\ln ^{2/2_{+}}\zeta ^{-1}\right)
\end{gather*}%
Condition (\ref{azzet}) implies\ that\ $\ln a^{-1}=\ln ^{1/1_{+}}\zeta ^{-1}$
and by (\ref{pl2}) $2/2_{+}>1/1_{+}$. Hence 
\begin{equation*}
\exp \left( \ln a^{-1}-\ln ^{2/2_{+}}\zeta ^{-1}\right) \rightarrow 0,
\end{equation*}%
and (\ref{psiouta}) holds.\ 
\end{proof}

\begin{proposition}
\label{Ledens1}Let\ $\psi ^{\prime }$ be defined by (\ref{psi1d}), (\ref%
{psipsi}) where\ $s,S$\ are defined by (\ref{sprime}), (\ref{new11}), $R_{n}=%
\bar{\theta}a_{n}\rightarrow 0$. Then the following inequality holds\ in $%
\Xi \left( \bar{\theta}\right) $ 
\begin{equation}
\left\vert \mathcal{E}\left( \mathbf{\hat{r}}+a\mathbf{z}\right) -\gamma m%
\mathrm{c}^{2}\mathring{\psi}^{2}\left( \mathbf{z}\right) \right\vert \leq
C\zeta ^{2-\delta }\mathring{\psi}^{2}\left( \mathbf{z}\right) .
\label{Eminmc}
\end{equation}%
Conditions and (\ref{Endef}), (\ref{Egrc}) are satisfied with $\mathcal{\bar{%
E}}_{\infty }\left( t\right) =\gamma m\mathrm{c}^{2}$. The following
estimate holds%
\begin{equation*}
\left\vert \rho -q\mathring{\psi}^{2}\left( \mathbf{z}\right) \right\vert
\leq C\zeta ^{2-\delta }\mathring{\psi}^{2}\left( \mathbf{z}\right) ,
\end{equation*}%
and\ (\ref{rhoinf})\ holds as well with $\bar{\rho}_{\infty }=q$.
\end{proposition}

\begin{proof}
According to (\ref{Edens}), since $\gamma ^{2}+\beta ^{2}\gamma
^{2}+1=2\gamma ^{2},$\ we have 
\begin{gather*}
\mathcal{E}\left( \mathbf{\hat{r}}+a\mathbf{z}\right) =m\mathrm{c}^{2}\gamma
^{2}e^{2\sigma }\mathring{\psi}^{2}\left( z\right) +\frac{m\mathrm{c}^{2}}{2}%
\left( \left( -\gamma +\Phi \right) ^{2}-\gamma ^{2}+\left( \beta \gamma
-Z\right) ^{2}-\beta ^{2}\gamma ^{2}\right) e^{2\sigma }\mathring{\psi}%
^{2}\left( \mathbf{z}\right) \\
+\frac{m\mathrm{c}^{2}}{2}\zeta ^{2}\left[ \left( \sigma +\beta z\right)
^{2}+2\mathbf{z}^{2}-\left( 2\sigma +\ln \pi ^{3/2}+2\right) \right]
e^{2\sigma }\mathring{\psi}^{2}\left( \mathbf{z}\right) .
\end{gather*}%
Using (\ref{ZfiQ}) and (\ref{fiest3})\ we conclude that 
\begin{gather}
\frac{m\mathrm{c}^{2}}{2}\left\vert \left( \left( -\gamma +\Phi \right)
^{2}-\gamma ^{2}+\left( \beta \gamma -Z\right) ^{2}-\beta ^{2}\gamma
^{2}\right) e^{2\sigma }\mathring{\psi}^{2}\left( \mathbf{z}\right)
\right\vert  \notag \\
\leq C_{0}\left( \left\vert \Phi \right\vert \left( \left\vert \Phi
\right\vert +\gamma \right) +\left\vert Z\right\vert \left( \left\vert
Z\right\vert +\left\vert \beta \right\vert \gamma \right) \right) \mathring{%
\psi}^{2}\left( \mathbf{z}\right) \leq C_{0}^{\prime }\left\vert \Phi
\right\vert \mathring{\psi}^{2}\left( \mathbf{z}\right) \leq C_{0}^{\prime
\prime }\zeta ^{2-\delta }\mathring{\psi}^{2}\left( \mathbf{z}\right)
\end{gather}%
One can easily verify that in $\Xi \left( \bar{\theta}\right) $%
\begin{gather*}
\frac{m\mathrm{c}^{2}}{2}\zeta ^{2}\left\vert \left[ \left( \sigma +\beta
z\right) ^{2}+2\mathbf{z}^{2}-\left( 2\sigma +\ln \pi ^{3/2}+2\right) \right]
\mathrm{e}^{2\sigma }\mathring{\psi}^{2}\left( z\right) \right\vert \\
\leq C_{1}\zeta ^{2}\left( \left\vert \mathbf{z}\right\vert ^{2}+1\right) 
\mathring{\psi}^{2}\left( z\right) \leq C_{1}^{\prime }\zeta ^{2-\delta }%
\mathring{\psi}^{2}\left( \mathbf{z}\right) ,
\end{gather*}%
hence (\ref{Eminmc}) holds. Let us estimate now 
\begin{equation*}
\mathcal{\bar{E}}_{n}\left( t\right) =\int_{\Omega \left( \mathbf{\hat{r}}%
(t),R_{n}\right) }\mathcal{E}\,\mathrm{d}^{3}x=a^{3}\int_{\Xi \left( \bar{%
\theta}\right) }\mathcal{E}\left( \mathbf{\hat{r}}+a\mathbf{z}\right) \,%
\mathrm{d}^{3}z.
\end{equation*}%
Note that 
\begin{equation}
a^{3}\int_{\Xi \left( \bar{\theta}\right) }\mathring{\psi}^{2}\left( \mathbf{%
z}\right) d^{3}z=\pi ^{-3/2}\int_{\left\vert z\right\vert \leq \bar{\theta}}%
\mathrm{e}^{-\left\vert z\right\vert ^{2}}\,\mathrm{d}^{3}z=1-\pi
^{-3/2}\int_{\left\vert z\right\vert \geq \bar{\theta}}\mathrm{e}%
^{-\left\vert z\right\vert ^{2}}\,\mathrm{d}^{3}z.  \label{note}
\end{equation}%
Using (\ref{Eminmc}) and the above formula, we conclude that 
\begin{equation}
\mathcal{\bar{E}}_{n}\left( t\right) \rightarrow \gamma m\mathrm{c}^{2}
\label{emmc2}
\end{equation}%
uniformly for all $t$ and we get (\ref{Endef}). Since $\gamma m\mathrm{c}%
^{2}>0$ we get the energy estimate from below (\ref{Egrc}). To express $\rho 
$ we use (\ref{paf6})\ and in $\left( \tau ,z\right) $-variables obtain 
\begin{equation*}
\rho =-\frac{\chi q}{m\mathrm{c}^{2}}\left\vert \psi \right\vert ^{2}\func{Im%
}\frac{\tilde{\partial}_{t}\psi }{\psi }=-\frac{\chi q}{m\mathrm{c}^{2}}%
\mathrm{e}^{2\sigma }\mathring{\psi}^{2}\frac{\mathrm{c}}{a_{\mathrm{C}}}%
\left( \Phi -\gamma \right) =q\mathrm{e}^{2\sigma }\mathring{\psi}^{2}\gamma
-q\mathrm{e}^{2\sigma }\Phi \mathring{\psi}^{2}.
\end{equation*}%
Using (\ref{Fiztau}), (\ref{siggam}) and (\ref{note}) we obtain (\ref{rhoinf}%
).
\end{proof}

\begin{theorem}
\label{Ttrajcon}Let\ trajectory $r\left( t\right) $ satisfy relations (\ref%
{betep0})-(\ref{betcheck}), and let $\varphi \left( t\right) $\ be defined
by (\ref{fiex})-(\ref{new11})\ where $\varphi _{\mathrm{b}}$\ is given by (%
\ref{fi2for})\ and $\varphi _{\mathrm{0}}\left( t\right) $\ is a given
function which satisfies 
\begin{equation}
\left\vert \varphi _{\mathrm{0}}\right\vert +\left\vert \partial _{t}\varphi
_{\mathrm{0}}\right\vert +\left\vert \partial _{t}^{2}\varphi _{\mathrm{0}%
}\right\vert \leq C.  \label{fi0est}
\end{equation}%
Suppose also $\psi $ to be of the form (\ref{psi1d})\ with\ $\psi _{1D}$
defined by (\ref{psipsi}), where the phases $s\left( t\right) $, $S\left(
t,y\right) $\ are given by (\ref{sprime}), (\ref{Sfor}). Let $%
a_{n}\rightarrow 0$, $\zeta _{n}$ satisfy\ relations (\ref{azzet}), and%
\begin{equation}
\theta _{n}=\bar{\theta}=\ln ^{1/2_{+}}\zeta _{n}^{-1},\quad \chi _{n}=m%
\mathrm{c}a_{\mathrm{C,}n}=m\mathrm{c}\zeta _{n}a_{n},\quad R_{n}=\theta
_{n}a_{n}.  \label{defchi}
\end{equation}%
\ Then $\psi $ is a solution of the KG equation which concentrates at the
trajectory $\mathbf{\hat{r}}\left( t\right) =\left( 0,0,r\left( t\right)
\right) $.
\end{theorem}

\begin{proof}
According to (\ref{azzet}) $a^{-1}=\exp \ln ^{1/1_{+}}\zeta ^{-1}$, hence\ 
\begin{equation*}
R_{n}=\ln ^{1/2_{+}}\zeta _{n}^{-1}\exp \left( -\ln ^{1/1_{+}}\zeta
_{n}^{-1}\right) =\exp \left( \ln \ln ^{1/2_{+}}\zeta _{n}^{-1}-\ln
^{1/1_{+}}\zeta _{n}^{-1}\right) \rightarrow 0,
\end{equation*}%
implying that the contraction condition (\ref{Rnto0}) holds and we can
define concentrating neighborhood $\hat{\Omega}\left( \mathbf{\hat{r}}%
,R_{n}\right) $ by (\ref{Omhat}). Note that\ 
\begin{equation*}
a_{\mathrm{C}}/a^{2}=\zeta /a=\zeta \exp \ln ^{1/1_{+}}\zeta ^{-1}=\exp
\left( \ln ^{1/1_{+}}\zeta ^{-1}-\ln \zeta ^{-1}\right) \leq C,
\end{equation*}%
implying that inequality (\ref{aca2}) holds. Conditions of Definition \ref%
{DKGseq} are satisfied according to Proposition \ref{Lestnormfi1}.
Conditions of Definition \ref{Dlocconverge} are satisfied as well according
to Propositions \ref{Lintbound} and \ref{Ledens1}. Hence $\psi =\psi _{n}$
is a solution of the KG equation which concentrates at the trajectory $%
\mathbf{\hat{r}}\left( t\right) $.
\end{proof}

\begin{remark}
The accelerating force $-q\nabla \varphi $ defined by (\ref{new11})\ is of
order $1$, whereas by (\ref{fi2est}) the balancing force\ $-q\varphi _{%
\mathrm{b}}$\ is of order $\zeta ^{2-4\delta }$.\ Hence the balancing force,
while preserving the shape $\left\vert \psi \right\vert $,\ is vanishingly
small compared with the accelerating force.
\end{remark}

\begin{remark}
\label{Rgamma0}Similarly to solutions of the form (\ref{psi1d}), (\ref%
{psipsi}), (\ref{psi1sig})\ we can introduce one more parameter $\gamma
_{0}=\gamma \left( t_{0}\right) $ and modify (\ref{psi1sig}), (\ref{siggam})
as follows: 
\begin{equation}
\Psi \left( t,z\right) =\pi ^{-1/4}\mathrm{e}^{\sigma -\gamma
_{0}^{2}z^{2}/2},\quad \sigma =\ln \gamma ^{-1/2}-\ln \gamma _{0}^{-1/2}.
\label{gam0}
\end{equation}%
Such a modification introduces a fixed contraction which coincides with the
Lorentz contraction at $t_{0}=t_{0}$. All the analysis is quite similar and
the corresponding solutions concentrate at $\mathbf{\hat{r}}\left( t\right) $%
.\ In the case where $r\left( t\right) =vt$\ for $t<0$\ such a modification
with $t_{0}=-1$\ allows to obtain solutions which coincide with the free
uniform solutions with zero electric potential from Section \ref{Suniform}\
for $t<0$\ \ and can accelerate at positive times.
\end{remark}

\begin{remark}
\label{RLorcont}Here we discuss the robustness of the Lorentz contraction in
the accelerated motion. Let us consider a rectilinear motion with a
trajectory $r\left( t\right) $ such that the velocity takes a constant value 
$v_{0}$ for $t<0$\ and a different constant value $v_{1}\neq v_{0}$\ for $%
t\geq T_{1}>0$. Consequently there has to be a non-zero acceleration in the
interval $0<t<T_{1}$.\ Let us look at a solution $\psi $ \ of the form (\ref%
{psi1d}), (\ref{psipsi}), (\ref{psi1sig}) \ with $\Psi $ of the form (\ref%
{gam0}) and the parameter $\gamma _{0}=\left( 1-v_{0}^{2}/\mathrm{c}%
^{2}\right) ^{-1/2}$. Such a solution has a Gaussian shape $\left\vert \psi
\right\vert =\pi ^{-3/4}\gamma _{0}^{1/2}\gamma ^{-1/2}\mathrm{e}%
^{-z_{1}^{2}/2-z_{2}^{2}/2-\gamma _{0}^{2}z_{3}^{2}/2}$ with the velocity
dependent factor $\gamma _{0}^{1/2}\gamma ^{-1/2}$ which does not depend on $%
\zeta $.\ When $t<-1$ we have $\varphi =0$, and $\psi $ is given\ in $\Omega
\left( \mathbf{\hat{r}}(t),\theta _{n}a_{n}\right) $ by the same formula (%
\ref{mvch4}) as the solution for a free charge, with a Gaussian shape $%
\left\vert \psi \right\vert =\pi ^{-3/4}\mathrm{e}%
^{-z_{1}^{2}/2-z_{2}^{2}/2-\gamma _{0}^{2}z_{3}^{2}/2}$. When $t\geq T_{1}$,
the motion is uniform with velocity $v_{1}$, the accelerating part $\varphi
_{\mathrm{ac}}$ of the potential is zero and the balancing potential $%
\varphi _{\mathrm{b}}$ though not zero when $t\geq T_{0}$, but it is of
order $\zeta ^{2-4\delta }\ $in\ $\Omega \left( \mathbf{\hat{r}}(t),\theta
_{n}a_{n}\right) $\ and is vanishingly small as $\zeta \rightarrow 0$. For $%
t\geq T_{1}$ the Gaussian $\left\vert \psi \right\vert =\pi ^{-3/4}\gamma
_{0}^{1/2}\gamma _{1}^{-1/2}\mathrm{e}^{-z_{1}^{2}/2-z_{2}^{2}/2-\gamma
_{0}^{2}z_{3}^{2}/2}$ has the same Lorentz contraction \ factor $\gamma _{0}$%
\ but a different amplitude$\ \gamma _{0}^{1/2}\gamma _{1}^{-1/2}$. The
principal part $\omega _{0}\mathrm{c}^{-2}\gamma vy-\omega _{0}t/\gamma $ of
the phase $\hat{S}$ in (\ref{psipsi}) is the same as in (\ref{mvch4}), and
hence it involves Lorentz contraction with factor $\gamma $. Consequently,
the principal part of the solution for $t\geq T_{1}$ involves components
with different values of the contraction factor. Therefore, \emph{the
principal part of solution, while translating with the constant velocity} $%
v_{1}$,\ \emph{cannot be obtained by the Lorentz transformation from the
solution for a free charge with the velocity }$v_{0}$, whereas the phase can
be. Thus, in general, the transition from velocity $v_{0}$ to velocity $v_{1}
$ cannot be reduced to the Lorentz transformation. \ The fixed Lorentz
contraction $\gamma _{0}$ of the Gaussian shape factor\ is preserved while
the velocity changes thanks to the external electric force which causes the
acceleration of the charge and also results in the change of the amplitude
of the charge distribution.\newline
Observe that Theorem \ref{TconcEinstein}\ can be applied to the considered
example, and Einstein's formula is applicable at all times with the same
rest mass $m$. In particular, formula $M=m\gamma $ implies that $M=m\gamma
_{0}\ $for $t<0$ and $M=m\gamma _{1}\ $for $t>T_{1}$. This observation shows
that Einstein's formula holds all the time and fully applies to accelerating
regimes, whereas the Lorentz contraction formula can be applied only to some
characteristics of an accelerating charge distribution. Such a difference is
specific to accelerating regimes and is in a sharp contrast with the case of
a global uniform motion without external forces. 
\end{remark}

\section{Concentration of solutions of a linear KG\ equation\label{Sconclin}}

The results of Section\ \ref{SReldyn} on concentrating solutions\ are
directly applicable to solutions of the linear KG equation by setting $G=0$.
In this case Theorem \ref{TconcEinstein}\ can be applied, and solutions of
the linear KG\ equations can only concentrate at trajectories which satisfy
relativistic point equations.\ The size parameter $a$ now is not involved in
the equation but rather describes the localization of a sequence of
solutions of the linear equations. The significant difference is that\ the
linear KG equations do not have global localized solutions as described in
Sections \ref{Snonlinrest}\ and \ref{Suniform}. Consequently, one cannot
simply apply the relativistic Einstein argument for the 4-vector of the
global energy-momentum\ to the linear KG\ equation.\ As to the results of
Section \ref{secAR}\ they can be modified for the linear case as follows.
For a given trajectory $r\left( t\right) $ wave function $\psi $ is defined
by by the same formulas (\ref{psipsi}), (\ref{psi1d}). The identity (\ref%
{G1D})\ cannot be used, and in $Q$ defined by (\ref{ReQ2}) the term$\ \zeta
^{2}\partial _{z}^{2}\Psi $ will not cancel with $G^{\prime }$, resulting in
replacing the term\ $\beta ^{2}\zeta ^{2}\left( z^{2}-1\right) $\ in (\ref%
{Qzetz}) by the term $\gamma ^{-2}\zeta ^{2}\left( z^{2}-1\right) $; the
term\ $2\zeta ^{2}\sigma $ in (\ref{Qzetz})\ is absent. These arguments show
that $Q$ is a small perturbation\ of order $\zeta ^{2}$. Since we do not use
the structural details of $Q$\ but only its smallness, all the estimates in
Section \ref{secAR}\ can be carried out with this modification. Hence, a
small (of order $\zeta ^{2-4\delta }$) balancing potential $\varphi _{%
\mathrm{b}}$ exists in the linear case as well. Analyzing estimates made in
Proposition \ref{Ledens1}, one finds that the nonlinearity produces a
vanishing contribution to $E_{\infty }$. Thus the following modification of
Theorem \ref{Ttrajcon} holds.

\begin{theorem}
\label{Ttrajcon1}Let\ trajectory $r\left( t\right) $ satisfy (\ref{betep0})-(%
\ref{betcheck}), and let the KG equation be linear, namely $G^{\prime }=0$
in (\ref{KGex}) and (\ref{paf1}), (\ref{emtn3}).\ Suppose $\psi $ to be of
the form (\ref{psi1d})\ with\ $\psi _{1D}$ defined by (\ref{psipsi}), where
the phases $s\left( t\right) $, $S\left( t,y\right) $\ are given by (\ref%
{sprime}) (\ref{Sfor}). Let $a_{n}\rightarrow 0$, $\zeta _{n}$ satisfy\ (\ref%
{azzet}), $\chi _{n}$, $R_{n},\theta _{n}\ $be defined by (\ref{defchi}),
and $\varphi \left( t\right) $\ be defined by (\ref{fiex})-(\ref{new11})\
where $\varphi _{\mathrm{b}}$\ is given by (\ref{fi2for})\ and $\varphi _{%
\mathrm{0}}\left( t\right) $\ is a given function satisfying (\ref{fi0est}%
).Then $\psi $ is a solution of the linear KG equation which concentrates at
the trajectory $\mathbf{\hat{r}}\left( t\right) =\left( 0,0,r\left( t\right)
\right) $.\ 
\end{theorem}

\section{Appendix \label{apprelmath}}

For completeness, we present here verification of the energy and momentum
conservation equations. To verify the energy conservation equation, we
multiply (\ref{KGex}) by $\tilde{\partial}_{t}^{\ast }\psi ^{\ast }$\ where $%
\tilde{\partial}_{t}$ is defined by (\ref{dtex}). We obtain 
\begin{gather*}
-\mathrm{c}^{-2}\partial _{t}\left( \tilde{\partial}_{t}\psi \tilde{\partial}%
_{t}^{\ast }\psi ^{\ast }\right) +\nabla \left( \nabla \psi \tilde{\partial}%
_{t}^{\ast }\psi ^{\ast }+\nabla \psi ^{\ast }\tilde{\partial}_{t}\psi
\right) \\
-\left( \nabla \psi \nabla \tilde{\partial}_{t}^{\ast }\psi ^{\ast }+\nabla
\psi ^{\ast }\nabla \tilde{\partial}_{t}\psi \right) -\partial _{t}G\left(
\psi ^{\ast }\psi \right) -\kappa _{0}^{2}\partial _{t}\psi \psi ^{\ast }=0
\end{gather*}%
and\ rewrite in the form 
\begin{gather}
\partial _{t}\left( \mathrm{c}^{-2}\left( \tilde{\partial}_{t}\psi \tilde{%
\partial}_{t}^{\ast }\psi ^{\ast }\right) +\nabla \psi \nabla \psi ^{\ast
}+G\left( \psi ^{\ast }\psi \right) +\kappa _{0}^{2}\psi \psi ^{\ast
}\right) =  \label{equen1} \\
\nabla \left( \nabla \psi \tilde{\partial}_{t}^{\ast }\psi ^{\ast }+\nabla
\psi ^{\ast }\tilde{\partial}_{t}\psi \right) -2m\chi ^{-2}\nabla \varphi
\cdot \mathbf{J.}  \notag
\end{gather}%
Finally, using the definition of $\tilde{\partial}_{t},\mathbf{J}$ and $%
\mathbf{P}$, we rewrite it in the form of the energy conservation equation (%
\ref{equen2}).

To obtain the momentum conservation equation (\ref{momeq1}), we multiply
equation (\ref{KGex}) by $\nabla \psi ^{\ast }$, and then\ using the
identity $v\tilde{\partial}_{t}u+u\tilde{\partial}_{t}^{\ast }v=\partial
_{t}\left( uv\right) $,\ we obtain 
\begin{equation*}
-\frac{1}{\mathrm{c}^{2}}\partial _{t}\left( \tilde{\partial}_{t}\psi \nabla
\psi ^{\ast }\right) +\frac{1}{\mathrm{c}^{2}}\tilde{\partial}_{t}\psi 
\tilde{\partial}_{t}^{\ast }\nabla \psi ^{\ast }+\tilde{\nabla}^{2}\psi 
\tilde{\nabla}^{\ast }\psi ^{\ast }-G^{\prime }\left( \psi ^{\ast }\psi
\right) \psi \tilde{\nabla}^{\ast }\psi ^{\ast }-\kappa _{0}^{2}\psi \tilde{%
\nabla}^{\ast }\psi ^{\ast }=0.
\end{equation*}%
Obviously, $\tilde{\partial}_{t}\left( \nabla \psi \right) =\nabla \tilde{%
\partial}_{t}\psi -\frac{\mathrm{i}q}{\chi }\nabla \varphi \psi $.
Therefore, we arrive at the equation 
\begin{gather}
\partial _{t}\mathbf{P}+\frac{\chi ^{2}}{2m}\frac{1}{\mathrm{c}^{2}}\left( 
\tilde{\partial}_{t}\psi \left( \nabla \tilde{\partial}_{t}^{\ast }\psi +%
\frac{\mathrm{i}q}{\chi }\nabla \varphi \psi ^{\ast }\right) +\tilde{\partial%
}_{t}^{\ast }\psi ^{\ast }\left( \nabla \tilde{\partial}_{t}\psi -\frac{%
\mathrm{i}q}{\chi }\nabla \varphi \psi \right) \right)  \label{mom0} \\
+\frac{\chi ^{2}}{2m}\left( \nabla ^{2}\psi \nabla \psi ^{\ast }+\nabla
^{2}\psi ^{\ast }\nabla \psi \right) -\frac{\chi ^{2}}{2m}\nabla \left(
G\left( \psi ^{\ast }\psi \right) +\kappa _{0}^{2}\psi ^{\ast }\psi \right)
=0.  \notag
\end{gather}%
Substituting the identity 
\begin{equation}
\left( \nabla ^{2}\psi \nabla \psi ^{\ast }+\nabla ^{2}\psi ^{\ast }\nabla
\psi \right) =\partial _{j}\cdot \left( \partial _{j}\psi \nabla \psi ^{\ast
}+\partial _{j}\psi ^{\ast }\nabla \psi \right) -\nabla \left( \nabla \psi
\cdot \nabla \psi ^{\ast }\right)  \label{gr2gr}
\end{equation}%
into (\ref{mom0}) we obtain the \emph{momentum conservation law} for\emph{\ }%
$l$-th component of $\mathbf{P}$ in the following form\ 
\begin{gather*}
\partial _{t}\mathbf{P}_{l}+\frac{\chi ^{2}}{2m}\frac{1}{\mathrm{c}^{2}}%
\nabla _{l}\left( \tilde{\partial}_{t}\psi \tilde{\partial}_{t}^{\ast }\psi
\right) +\rho \nabla \varphi -\frac{\chi ^{2}}{2m}\nabla _{l}\left( G\left(
\psi ^{\ast }\psi \right) +\kappa _{0}^{2}\psi ^{\ast }\psi \right) \\
+\frac{\chi ^{2}}{2m}\left( \tsum\nolimits_{j}\nabla _{j}\left( \tilde{\nabla%
}_{j}\psi \tilde{\nabla}_{l}^{\ast }\psi ^{\ast }+\tilde{\nabla}_{j}^{\ast
}\psi ^{\ast }\tilde{\nabla}_{l}\psi \right) -\nabla
_{l}\tsum\nolimits_{j}\left( \tilde{\nabla}_{j}\psi \tilde{\nabla}_{j}^{\ast
}\psi ^{\ast }\right) \right) =0.
\end{gather*}%
Finally, using (\ref{paf1}),\ we rewrite the the above equation in the form\
(\ref{momeq1}).

\textbf{Acknowledgment.} The research was supported through Dr. A. Nachman
of the U.S. Air Force Office of Scientific Research (AFOSR), under grant
number FA9550-11-1-0163.

\end{document}